\documentclass{article}
\usepackage{geometry}

\usepackage{color}
\usepackage{url}
\usepackage{amsmath}
\usepackage{amsthm}
\usepackage{amssymb}
\usepackage{graphicx}
\usepackage{lmodern}

\usepackage{varioref}
\usepackage{hyperref}
\usepackage{cleveref}

\usepackage{thm-restate}

\usepackage{bbm}

\usepackage{comment}

\makeatletter

\providecommand{\tabularnewline}{\\}

\providecommand{\axiomname}{Axiom}
\providecommand{\corollaryname}{Corollary}
\providecommand{\theoremname}{Theorem}

\providecommand{\definitionname}{Definition}
\providecommand{\examplename}{Example}

\theoremstyle{plain}
\newtheorem{thm}{\protect\theoremname}
\newtheorem{prop}[thm]{Proposition}
\theoremstyle{definition}
\newtheorem{example}[thm]{\protect\examplename}
\theoremstyle{definition}
\newtheorem{defn}[thm]{\protect\definitionname}
\theoremstyle{definition}
\newtheorem{remark}[thm]{Remark}
\theoremstyle{plain}
\newtheorem{ax}[thm]{\protect\axiomname}
\theoremstyle{plain}
\newtheorem{cor}[thm]{\protect\corollaryname}
\newtheorem{lem}[thm]{Lemma}

\usepackage{mathabx}

\DeclareMathOperator*{\argmax}{arg\,max}

\@ifundefined{showcaptionsetup}{}{%
 \PassOptionsToPackage{caption=false}{subfig}}
\usepackage{subfig}
\makeatother

\usepackage[style=authoryear]{biblatex}

\DeclareFieldFormat{titlecase}{\MakeSentenceCase*{#1}}

\usepackage{ifthen}
\newboolean{commentsactivated}
\setboolean{commentsactivated}{true}
\newcommand{\jt}[1]{\ifthenelse{\boolean{commentsactivated}}{{\color{blue} {\em JT: #1 }}}{}}

\begin{document}

\title{Modeling evidential cooperation in large worlds}

\author{Johannes Treutlein}

\date{First written in 2018; major edits in 2023}

\maketitle

\begin{abstract}
\emph{Evidential cooperation in large worlds} (ECL) refers to the idea that humans and other agents can benefit by cooperating with similar agents with differing values in causally disconnected parts of a large universe. Cooperating provides agents with evidence that other similar agents are likely to cooperate too, resulting in gains from trade for all. This could be a crucial consideration for altruists.

I develop a game-theoretic model of ECL as an incomplete information bargaining problem. The model incorporates uncertainty about others' value systems and empirical situations, and addresses the problem of selecting a compromise outcome. Using the model, I investigate issues with ECL and outline open technical and philosophical questions.

I show that all cooperators must maximize the same weighted sum of utility functions to reach a Pareto optimal outcome. However, I argue against selecting a compromise outcome implicitly by normalizing utility functions. I review bargaining theory and argue that the Nash bargaining solution could be a relevant Schelling point. I introduce \emph{dependency equilibria} \parencite{Spohn2007-fp}, an equilibrium concept suitable for ECL, and generalize a folk theorem showing that the Nash bargaining solution is a dependency equilibrium. I discuss gains from trade given uncertain beliefs about other agents and analyze how these gains decrease in several toy examples as the belief in another agent decreases.

Finally, I discuss open issues in my model. First, the Nash bargaining solution is sometimes not coalitionally stable, meaning that a subset of cooperators can unilaterally improve payoffs by deviating from the compromise. I investigate conditions under which stable payoff vectors exist. Second, I discuss how to model agents' default actions without ECL.

\end{abstract}

\tableofcontents{}

\section{Introduction}

\emph{Evidential cooperation in large worlds} (ECL)\footnote{In previous work, this concept has been referred to as ``multiverse-wide
cooperation via superrationality'' (MSR).} \parencites{Oesterheld2017-qg}{Gloor2017}{Oesterheld2018} is a crucial consideration that could have important implications for the prioritization of altruists.
To illustrate the idea, consider a prisoner's dilemma between two artificial
agents with identical source code. Even if both agents cannot causally
interact, one agent's action provides them with strong
evidence about the other agent's action. Evidential
decision theory (EDT), as well as functional decision theory \parencite{Yudkowsky2017-vb} and some variants of causal decision theory (CDT) \parencite{Spohn2012-fo,Poellinger2013-we},
say that agents should take such evidence into account when making
decisions. In situations like the prisoner's dilemma with two identical
agents, they prescribe cooperation for this reason, an idea that is also called \emph{superrationality} \parencite{hofstadter1983dilemmas}. ECL is based on
the idea that humans on Earth are in a similar situation as such agents.

First, there probably is a large or infinite universe,
containing vast numbers of civilizations, inhabiting different, causally disconnected parts of the universe \parencite{tegmark2003parallel,tegmark2015our}. I refer to such a large universe as a \emph{multiverse}, and to causally disconnected parts of it as \emph{universes}, regardless of the specific structure of the universe (e.g., these parts could just be far-apart regions of space). Given their vast number, there are likely universes containing agents that are very similar to humans, such that humans'
actions are evidence about these agents' actions \parencite{macaskill2021evidentialist}.

Second, these
agents may pursue different goals, leading to possible gains from trade. For instance, pursuing a given goal in one universe may have diminishing returns, and agents may care about other universes as well. In that case, it may be beneficial for agents to trade by pursuing a mixture of everyone's goals in all universes. Since agents in different universes cannot
communicate and there is no way to enforce an agreement, this puts them in a collective prisoner's dilemma. Under the right conditions,
the abovementioned decision theories recommend that humans take the
preferences of other, similar agents in the multiverse into account,
in order to produce the evidence that these agents do in turn take humans'
preferences into account, leaving everyone better off.

According to \textcite[sec.~4]{Oesterheld2017-qg}, this idea could
have far-reaching implications for the prioritization of altruists. For instance, given 
ECL, some forms of moral advocacy could become ineffective: agents
advocating their particular values provides them with evidence that
others will do the same, potentially neutralizing each other's
efforts \parencite[sec.~4.2]{Oesterheld2017-qg}. Moreover, ECL could play a role in deciding which strategies
to pursue in AI alignment. If potential gains from cooperation
are vast, then it becomes more important to ensure that AI systems are aligned with humans' idealized philosophical views on decision theory and ECL.\footnote{Note that interventions to promote ECL could also backfire by exacerbating other risks from advanced AI \parencite[see][]{xu2021open}.}

In this report, I develop a game-theoretic model of ECL as an incomplete information bargaining problem, incorporating uncertainty about the values and empirical situations of potential superrational cooperators, and addressing the problem of selecting a compromise outcome. I clarify the conditions that make ECL feasible and analyze gains from trade given empirical uncertainty. Moreover, I discuss several technical and philosophical problems that arise.

Basic knowledge of game theory, such as normal form games, Nash equilibria, and the prisoner's dilemma \parencite[see][]{osborne1994course}, as well as decision theory and ECL (see \textcite{Gloor2017} for an introduction), will be helpful for understanding this report.

\subsection{Summary}
Here, I provide a short summary of the report, highlighting key contributions. Afterwards, I outline the organization of the remaining report, and briefly discuss related work.

\subsubsection{Game-theoretic models} 
I introduce three models: a bargaining model, a Bayesian game model, and a Bayesian bargaining model, combining the two previous models. The first two models are useful since many issues can more easily be discussed in the less general setting, and this structure may make the report easier to follow. However, it is possible to skip directly to the final model in \Cref{sec:Bargaining-with-incomplete}.

In a bargaining game, players have to agree on some compromise outcome, from a \emph{feasible set} of achievable payoff vectors. A \emph{disagreement point} specifies the outcome that is realized if no compromise is reached. I argue for modeling ECL as a bargaining problem, since (i) there is an inherent bargaining problem in determining a compromise between superrational cooperators that needs to be addressed, (ii) bargaining solutions that are supported by plausible axioms serve as Schelling points, and (iii) there are important parallels between \emph{acausal trade}\footnote{\url{https://www.lesswrong.com/tag/acausal-trade}} (where agents use mutual simulations to reach an agreement) and ECL, meaning that bargaining could be a relevant model for agents forming conditional beliefs over other agents' actions. I also address \textcite[Sec. 2.8]{Oesterheld2017-qg}'s suggested approach of pursuing a sum of normalized utility functions as a compromise utility function. I show that to achieve a Pareto optimal outcome, i.e., an outcome that cannot be improved upon without making any player worse off, everyone has to maximize the same compromise utility function. However, I argue for choosing a compromise based on a bargaining solution rather than a normalization method such as variance normalization, on the grounds that the latter can leave agents worse off than without the compromise. I review two popular bargaining solutions, the Nash bargaining solution (NBS) and the Kalai-Smorodinsky bargaining solution (KSBS), and conclude that the NBS could be a relevant Schelling point for ECL.

The Bayesian game formalism serves to incorporate incomplete information---that is, information about the values and available options of other players. Specifically, I use a modified version of \textcite{Harsanyi1967}'s type space formalism. In my model, there is a large number of players, living in different universes. Each player is assigned a \emph{type}, representing their values and empirical situation, according to some prior distribution \(p\). Players' posterior beliefs over types, after updating on their own type, represent their beliefs over other universes. Players' utility functions depend on the actions and types of all players. Relaxing the assumption of a common prior \(p\) is an important area for future work.

Finally, the Bayesian bargaining model implements a bargaining game on top of a Bayesian game, incorporating bargaining with incomplete information. The feasible set here is the set of expected utilities that can be produced by players for all the types, given the types' beliefs about other players.

\subsubsection{Gains from trade under uncertainty}

Two important assumptions in this report are \emph{additive separability} and \emph{anonymity}. Additive separability means each player's utility functions can be expressed as a sum of contributions from other players. This would be true for total utilitarians but false for average utilitarians valuing average well-being across the multiverse, for instance.
Anonymity means beliefs, utilities and strategies depend only on types, not on specific players. This means we do not distinguish between different universes.

Given these two assumptions, we can regard strategies as vectors \(\alpha\in A^T\), where \(T\) is the set of types and \(A\) the set of strategies for any type. The expected utility of a strategy for a type \(t\in T\) can be simplified to the expression
\begin{equation}\label{eu-intro}EU_t(\alpha)=u_{t,t}(\alpha) + (n-1)\sum_{t'\in T}p(t'\mid t)u_{t',t}(\alpha_{t'})\end{equation}
where \(n\) is the number of players, \(p(t'\mid t)\) is the belief of any player of type \(t\) that any other player has type \(t'\), and \(u_{t',t}(\alpha_{t'})\) is the utility provided by a player of type \(t'\) to a player of type \(t\). 
The first term is the utility produced by a player for themself, and the second term stands for the expected utility produced by all the other players.

Note that if \(n\) is large, the expected utility is dominated by the second term, meaning that the utilities produced by a player for themself in their own universe can be ignored. It follows that a potential compromise option \(\beta\) produces gains from trade for a type \(t\) if
\[\sum_{t'\in T\setminus \{t\}}p(t'\mid t)(u_{t',t}(\alpha_{t'})-u_{t',t}(\beta_t))\geq p(t\mid t)(u_{t,t}(\alpha)-u_{t,t}(\beta)).\]
That is, both potential gains from other types' cooperation as well as potential losses due to players of type \(t\) compromising are weighted by type \(t\)'s posterior beliefs over the types of other players. If the former outweigh the latter, then \(\beta\) leads to gains for players of type \(t\).

This shows that when it comes to gains from trade, what matters are players' posterior beliefs over other players' types. For instance, if types are certain that all players have the same type, i.e., \(p(t\mid t)=1\), then no trade is possible. If \(p(t'\mid t)=0\) for a specific type \(t'\), then that type cannot benefit players of type \(t\). If all types have the same posterior beliefs, then trade may in principle be possible, depending on the different types' options. In general, different beliefs can put a tax on trade.

I also consider a model that includes uncertainty over whether other players are superrationalists or sufficiently similar to enable ECL. However, I argue that such considerations can also be incorporated into the posterior beliefs \(p(t'\mid t)\), so this extension does not increase the generality of the model.

\subsubsection{Double decrease and Paretotopia}
Using \Cref{eu-intro}, we can analyze gains from trade in different toy models. I consider an example with a trade between two types, \(T=\{1,2\}\). Both types start out with an equal number of resources and can invest resources into either type's utility function. Resource investments have diminishing returns, leading to potential gains from trade. As a compromise outcome, I use the NBS. I consider square root as well as logarithmic returns to resources. \Cref{fig:intro1} shows \emph{individual feasible sets} in each case, which are the sets of expected utilities a player of each type can produce for both types. Gains from trade are larger given logarithmic utilities, since utilities diminish faster in this case.

Using this model, I analyze resource investments in the respective other type and gains from trade under the NBS, for different posterior beliefs \(p(t'\mid t)\) in the other type (assuming both types have the same prior weight \(p(1)=p(2)=\frac{1}{2}\)) (\Cref{fig:intro2}). As this belief goes down, gains from trade go down approximately quadratically in the square root utility case, leading to a ``double decrease'' as observed by \textcite{armstrong2017double}. However, in \textcite{drexler2019pareto}'s ``Paretotopia'' model with logarithmic returns to resource investments, gains from trade diminish more slowly with the belief in the other player.

\begin{figure}
\centering{}
\subfloat[Sqiare root utilities.]{\includegraphics[width=0.48\textwidth]{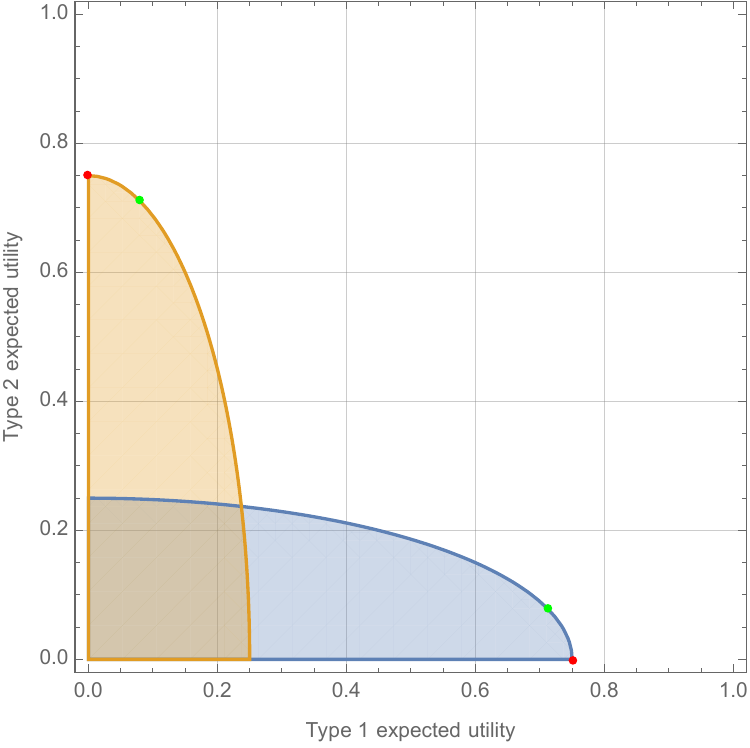}}
\hfill\subfloat[Logarithmic utilities]{\includegraphics[width=0.48\textwidth]{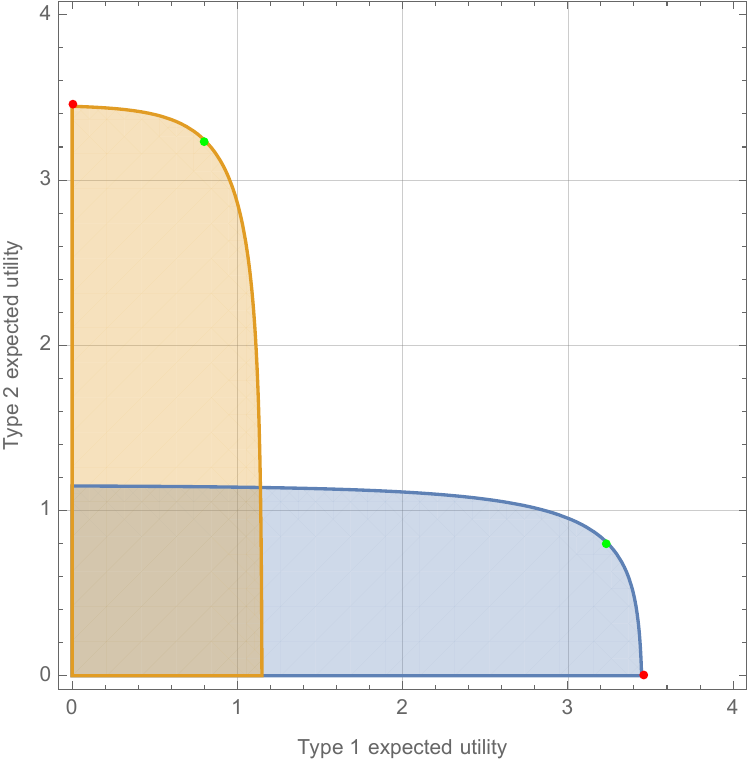}}
\caption{Sets of expected utility vectors that can be produced by either of two types investing resources in each other's value system. The disagreement point, at which each type only optimizes for their own values, is displayed in red, and the action corresponding to the NBS in green. Here, the prior of each type is \(p(t)=\frac{1}{2}\) and the posterior belief that any other player has the same type is \(p(t'\mid t)=\frac{1}{4}\). This means that gains from trade are smaller than with equal beliefs.}
\label{fig:intro1}
\end{figure}

\begin{figure}
\centering{}
\subfloat[Square root utilities.]{\includegraphics[width=0.48\textwidth]{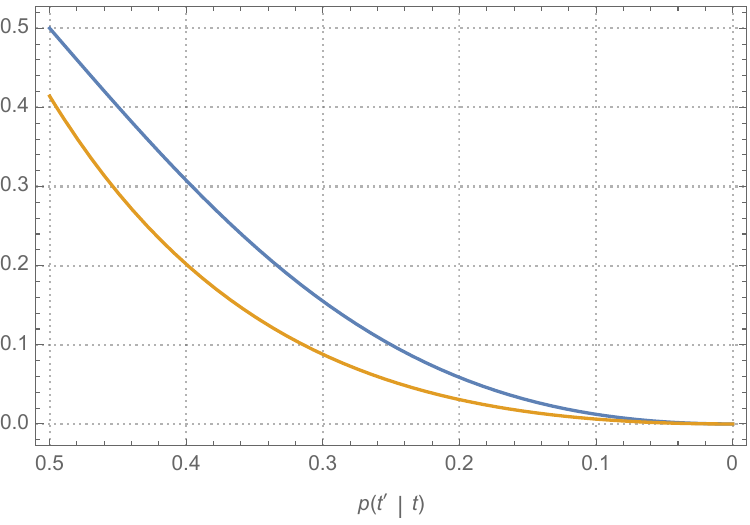}}
\hfill\subfloat[Logarithmic utilities]{\includegraphics[width=0.48\textwidth]{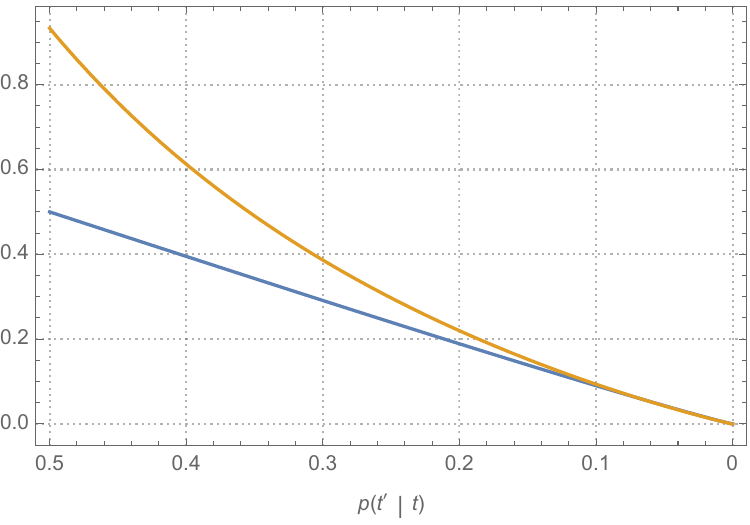}}
\caption{Share of expected utility received by the other type (blue) and percent gains from trade (orange). The belief \(p(t'\mid t)\) of any player of type \(t\) that other players have type \(t'\), for \(t'\neq t\), ranges between \(\frac{1}{2}\) and \(0\). When \(p(t'\mid t)=\frac{1}{2}\), all players have the same uniform posterior distribution over types, leading to maximal gains from trade. When \(p=0\), players think that all other players are of the same type, so no trade is possible. As \(p\) goes to zero, gains from trade decrease. In the square root utility case, they go down roughly quadratically, exhibiting a ``double decrease'' \parencite{armstrong2017double}. In the logarithmic utility case, they diminish more slowly.}
\label{fig:intro2}
\end{figure}

\subsubsection{Equilibrium concepts}
I introduce two equilibrium concepts for the Bayesian game and Bayesian bargaining game models. First, I introduce Bayesian Nash equilibria. In the additively separable case, these equilibria are trivial as each player is simply optimizing for their own values in their own universe, ignoring other players. Second, I introduce a generalization of \textcite{Spohn2007-fp}'s \emph{dependency equilibria} for Bayesian games and for continuous action spaces. A dependency equilibrium is a joint belief over the actions of all players, where every player's actions have optimal \emph{conditional} expected utility. Since it evaluates conditional probabilities and allows for dependencies between players' actions, dependency equilibria are suitable to model the superrational reasoning required for ECL. For instance, in a prisoner's dilemma, there is a dependency equilibrium in which both players cooperate.\footnote{There are several other equilibrium concepts in the literature with similar properties \parencite{al2015evidential,daley2017magical,halpern2018game}, which I have not looked at in this report.} My technical contributions are generalizing dependency equilibria to Bayesian games and to continuous action spaces. The latter is necessary for my bargaining model since players bargain over a continuous space of, e.g., \emph{independent} randomizations over actions, or continuous resource investments.

I prove several results about dependency equilibria in my model, including a generalization of \textcite{Spohn2007-fp}'s folk theorem for dependency equilibria, showing that any Pareto improvement over a Bayesian Nash equilibrium is a dependency equilibrium. As a corollary, it follows that the NBS with the Nash equilibrium disagreement point is a dependency equilibrium. I also show that a dependency equilibrium with independent action distributions is a Bayesian Nash equilibrium.

\subsubsection{Disagreement points}
I discuss the problem of choosing a disagreement point in ECL. Since ECL only involves choosing some compromise action based on some joint belief, without any actual bargaining, it is unclear what the relevant notion of non-compromise outcome should be. However, how to model agents' default options without ECL is an important question in general, not only in my bargaining model.

A natural option is the Bayesian Nash equilibrium, but there is also the \emph{threat point}, which is the equilibrium of a game in which players choose disagreement actions to improve their bargaining position. I review a plausible axiomatization of the threat point by \textcite{Nash1953} and show that the NBS with the threat disagreement point can sometimes lead to bargaining outcomes that are worse than a Nash equilibrium. I also show that the NBS with the threat disagreement point is still a dependency equilibrium. Coercion to join a compromise should not be relevant to ECL, since there are no explicit threats. However, threats might be relevant for the same reason bargaining in general is relevant to ECL. The question of disagreement points is an important area for future work.\footnote{It may be valuable to review recent work by \textcite{diffractor2022rose} on threat-resistant bargaining in this context.}

\subsubsection{Coalitional stability}
Finally, I discuss the issue of coalitional stability. A bargaining solution is coalitionally stable if it is in the \emph{core}, which is the set of payoff vectors in the feasible set such that no subgroup of players (coalition) can strictly increase payoffs for all of its members by unilaterally deviating from the compromise. Coalitional stability is an important criterion for a compromise solution for ECL since it seems plausible that players would choose to pursue a compromise with a subgroup of players if this leads to higher payoffs. Hence, if the \emph{grand coalition} of all players is not stable, this would lead to a difficult coalition finding problem, making ECL even more complicated to implement.

I show that the NBS with either the Nash or the threat disagreement point can sometimes be unstable. I then analyze the existence of core allocations. The core is known to be empty in general games \parencite[][ch.~13.2]{osborne1994course}. However, using a result by \textcite{scarf1967core}, I show that assuming additively separable utilities, the core is always nonempty. In this analysis, I assume worst-case responses by players outside the coalition, from among the possible Pareto optimal strategies they could pursue for themselves. (Specifically, I do not assume other players respond with threats against coalitions.)

Additionally, I show that if players outside the coalition respond with a Nash equilibrium, the core can be empty even given additive separability. This demonstrates that sometimes no stable bargaining solution exists that improves upon the Nash equilibrium disagreement point, a strong argument against this disagreement point. The intuition is that sometimes two players cooperating leads to negative externalities for a third player, leaving the third player worse off than with no cooperation. Motivated by my result on the existence of core allocations, I suggest an alternative disagreement point that guarantees stability.

\subsection{Outline}\begin{itemize}
\item In \Cref{sec:Preliminaries}, I discuss several assumptions and simplifications I make in the report.

\item In \Cref{sec:ECL-as-a-bargaining-problem}, I introduce a standard bargaining formalism. I argue that a bargaining problem is an appropriate model for ECL. After providing an example bargaining problem (\Cref{subsec:example-alice-bob}), I introduce the formal bargaining model and relevant notation (\Cref{ecl-bargaining-problem-setup}). I then discuss maximizing a sum of normalized utility functions as a compromise utility function (\Cref{subsec:Normalizing-utility-functions}). Next, I briefly review bargaining theory, introducing the Nash and Kalai-Smorodinsky bargaining solutions (\Cref{bargaining-theory}). Lastly, in \Cref{observations-bargaining}, I make some initial observations about the model.

\item In \Cref{sec:ECL-as-a-Bayesian-Game}, I introduce a Bayesian game model. In Sections~\ref{bayesian-game-formalism}--\ref{joint-strategy-distributions}, I introduce the formalism and notation. I then introduce Bayesian Nash equilibria and dependency equilibria (\Cref{joint-strategy-distributions}) and discuss extending the model to include uncertainty about decision procedures and similarity to other agents (\Cref{sec:Uncertainty-about-similarity}). Finally, I prove several equilibrium results (\Cref{subsec:ObservationsEquil}).

\item In \Cref{sec:Bargaining-with-incomplete}, I introduce Bayesian bargaining game, combining the previous models. I introduce the formal setup and notation in Sections~\ref{formal-setup-bayesian-bargaining}--\ref{strategies-bayesian-bargaining}. In \Cref{bargaining-theory-bayesian}, I introduce a version of the Nash bargaining solution adapted to my model. I then define equilibria in the model (\Cref{joint-distributions-equilibria-bayesian-bargaining}). Lastly, in \Cref{observations-final}, I discuss several takeaways: I provide equilibrium results, discuss how to think about gains from trade given uncertainty, and work through several toy examples, including \textcite{armstrong2017double}'s ``double decrease'' and \textcite{drexler2019pareto}'s ``Paretotopia'' model.

\item In \Cref{sec:Discussion-of-this}, I discuss two important issues: disagreement points (\Cref{subsec:The-question-of}) and coalitional stability (\Cref{fairness-and-coalitional}).

\item Finally, in \Cref{conclusion}, I conclude and outline possible future work.
\end{itemize}

\subsection{Related work}
A list of prior work on ECL can be found at \url{https://longtermrisk.org/msr}. No prior work introduces a formal game-theoretic model and discusses equilibria or bargaining theory. \textcite[Sec.~2.7]{Oesterheld2017-qg} includes a simple calculation establishing the plausibility of ECL but without modeling different players, beliefs, or utilities. \textcite[][Sec.~2.9.4]{Oesterheld2017-qg} introduces a variable for other players' decision theories, an idea I discuss in \Cref{sec:Uncertainty-about-similarity}. \textcite{treutlein2018three} introduces a simple model with variables for correlations, gains from trade, and number of cooperators, to establish a wager for ECL.

The most important related work is \textcite{armstrong2017acausal}'s sequence on acausal trade. He introduces a toy model where players have different utility functions and uncertainty about the existence of other players. Among other issues, he discusses how gains from trade change under different beliefs. I reproduce some of Armstrong's findings in \Cref{double-decrease}. Armstrong focuses on acausal trade and does not discuss relevance to ECL.
\section{Preliminaries\label{sec:Preliminaries}}

I make several assumptions and simplifications in this report:
\begin{enumerate}
\item I focus on EDT as a decision theory. In particular, I introduce a game-theoretic solution concept based on conditional expected utilities (see \Cref{subsec:Equilibrium-concepts}). While I believe my analysis also applies to other decision theories that take dependencies between similar
agents into account, I will not discuss this. My
analysis may also be relevant to readers with decision-theoretic
uncertainty, since there may be a wager to take ECL into account given any nonzero credence in EDT \parencite{macaskill2021evidentialist,treutlein2018three}. 
I do not model decision-theoretic uncertainty, but it could be added similarly to uncertainty about decision-theoretic similarity (see \Cref{sec:Uncertainty-about-similarity}).

\item I do  not address questions about the nature of dependencies between the decisions of different agents, how one could evaluate whether different agents' decisions are dependent, etc. However, I discuss modeling partial correlations or uncertain beliefs about dependencies in \Cref{sec:Uncertainty-about-similarity}.

\item I assume that there is only a finite set of
agents and the utilities of the options involved are all finite. This is a problem, since the most likely case in which the universe
is large enough to give rise to ECL is an infinite universe. It seems plausible that solutions to infinite
ethics will not change conclusions from my model \parencites[cf.][]{macaskill2021evidentialist}[][sec.\ 6.10]{Oesterheld2017-qg}. The assumption
of a finite set of agents is more problematic, since there likely exists a continuum of agents with a continuum of value systems. One may be able to discretize such a set and approximately recover the model discussed here, but it also seems possible that the general case leads to qualitatively new problems.

\item I assume that agents are Bayesian (conditional)
expected utility maximizers.

\item For simplicity, I model ECL as a one-off decision. For instance, this
could be a commitment to a policy or a decision to maximize some compromise
utility function in the future. I assume that it is possible to
commit oneself to this compromise, and that there won't be changes to the compromise based on new information about one's empirical situation in the universe. This is plausible if either the agents
\emph{can} actually commit themselves in this way, or if they just
never learn enough such that their assessment of the situation would
relevantly change. Note that this does not affect how agents \emph{arrive} at this compromise (whether by first-principles reasoning, by simulating agents in the multiverse, etc.; see the discussion in the next section).
\end{enumerate}

\section{Complete information bargaining model}
\label{sec:ECL-as-a-bargaining-problem}

In this section, I develop a model of ECL as a complete information
bargaining problem. A bargaining problem is a game between players
in which the players have some method of negotiating a binding agreement.
If everyone accepts the agreement, the actions specified by the agreement
are carried out. Otherwise, players carry out some disagreement action.
Complete information, as opposed to incomplete information, means
that everyone knows who the other players are, as well as their options and utility functions.

In ECL, players are uncertain about their superrational cooperators, so an incomplete information model would be more appropriate. Nevertheless, it is useful to start with complete information for simplicity, since many ideas from the complete information setup will transfer. I will relax the complete information assumption in the following sections.

A more critical assumption is that of using a bargaining model for ECL. ECL is based on the idea that an agent has some belief about other agents' actions,
conditional on their own action. The agent takes some
cooperative action, to produce the evidence that other agents also take more cooperative actions. This does not involve any explicit bargaining between the agents. Nevertheless, I believe using a bargaining model is useful for thinking about ECL.

First, the problem of choosing a compromise outcome in ECL has to be addressed in some way. In \Cref{subsec:Normalizing-utility-functions}, I discuss \parencite[Sec.\ 2.8]{Oesterheld2017-qg}'s suggested approach of maximizing a compromise utility function, consisting of a sum of normalized utility functions of all superrational cooperators. This is a valid approach, since every compromise that is Pareto optimal, i.e., that cannot be improved upon without making anyone worse off, is the result of maximizing \emph{some} common weighted sum of utility functions (see \Cref{thm-equal-weights}). However, I argue against choosing a compromise outcome implicitly by normalizing all agents' utility functions, e.g., according to variance, since that approach might leave some cooperators worse off than without a compromise. Formulating ECL as a bargaining problem and reviewing the relevant literature is a natural starting point for addressing the problem explicitly.

Second, a solution to a bargaining problem may serve as a Schelling point\footnote{\url{https://en.wikipedia.org/wiki/Focal_point_(game_theory)}} for superrational cooperators. Solutions can be supported by plausible axioms that could be universally agreed upon. Hence, bargaining theory can be one relevant reference point for determining which evidence one's actions provide about other agents' actions. It seems plausible that, if humans adopt some parsimonious solution, then other, similar agents will do the same.

Third, bargaining may be important because of a parallel between ECL and \emph{acausal trade}\footnote{\url{https://www.lesswrong.com/tag/acausal-trade}}. Acausal trade refers to the more general
idea that agents could be able to negotiate and enforce a cooperative outcome via mutual simulations, in the absence of any causal interaction. ECL is the special case in which, instead of mutual
simulation, similarity in decision algorithms or psychological processes
ensures a joint cooperative action. While humans might be able to engage in ECL, acausal trade is likely only feasible for superhuman AI systems.

I think there is no principled
distinction between acausal trade and ECL. Determining the conditional
beliefs about the actions of other agents involves, at least in principle,
similar questions as those concerning acausal trade. Conditioning on one's own decision process having some output, one needs to determine which
actions a similar but non-identical decision process in a similar
or symmetrical, but non-identical decision situation would output. At the same time, the other decision process is trying to make the same determination.
Due to such mutual dependencies between the actions of agents, one cannot divide the decision process clearly
into given conditional beliefs that specify which inferences to make
based on different actions, and the subsequent choice of the action
with the highest expected utility. Instead, one has to already make choices while
inferring the (logical) conditional credences. For instance, the inferred conditional distribution over opponent actions may be influenced by one's own commitment
to respond to opponent actions in a certain way \parencite[see][]{kokotajlo2019commitment,mennen2018wishful}.\footnote{{In a comment on an earlier draft, Max Daniel writes: ``If I understand this correctly, this seems important to me, and quite connected to some of the reasons why I feel skeptical about ECL having practical implications for humans. I also feel like it has been underemphasized in texts on ECL so far.''}}

It is prudent for humans to have conditional beliefs about the world,
including other agents, even without being able to entirely solve
this issue (which involves various open problems, for instance,
in logical uncertainty\footnote{\url{https://www.alignmentforum.org/tag/logical-uncertainty}}). In this situation, it makes sense to try
to improve one's guesses about ECL both based on reasoning
that is purely based on agents having beliefs about other agents,
and reasoning that involves hypothetical (acausal) bargaining.

Lastly, another way in which a bargaining problem is an inadequate model is that ECL is really a coalitional game. In a bargaining game,
there are only two possibilities: either all players agree to a proposed compromise action, or bargaining completely fails. If one player
disagrees, everyone plays their disagreement action. In a coalitional
game, any group of players can split off and negotiate an agreement, if this is beneficial to that group. I discuss this issue in \Cref{fairness-and-coalitional}.

Next, I give an example bargaining problem (\Cref{subsec:example-alice-bob}) and introduce the formal bargaining framework (\Cref{ecl-bargaining-problem-setup}). Afterwards, I discuss approaches to compromise that work via maximizing a sum of normalized utility functions (\Cref{subsec:Normalizing-utility-functions}). I argue against normalization and for explicitly picking out compromise outcomes. I then review some bargaining theory and discuss the Nash bargaining and the Kalai-Smorodinsky bargaining solutions, concluding that the former may serve as a good Schelling point (I review another solution by \textcite{Armstrong2013} in \Cref{appendix-armstrong-solution}). Finally, in \Cref{observations-bargaining}, I make some initial observations and discuss issues that arise in the bargaining model. I address the uniqueness of the actions corresponding to bargaining solutions  (\Cref{uniqueness}) and discuss how to think about gains from trade in the bargaining framework (\Cref{possible-gains-from-trade-bargaining}).

\subsection{Example}
\label{subsec:example-alice-bob}
Here, I give an example bargaining problem to motivate the following formal definitions, derived from a case by \textcite{armstrong2017double}.
\begin{example}
\label{exa:bargaining-case-1}
There are two players, Alice and Bob. Alice and
Bob care in an additive way about the things that both do. Say Alice
has 10 and Bob has 5 units of some resource, and \(A\) and \(B\) is the amount spent \emph{on Alice's utility function}, by Alice and Bob respectively. \(10-A\) and \(5-B\) are the respective amounts spent on Bob. Resources invested in
Alice's utility function produce linear utility for her, so her utility
is $A+B$. Bob's utility function,
on the other hand, has diminishing returns; the marginal cost of one
additional utilon equals the utilons that have already been produced. So to create
$x$ utilons for Bob, both Alice and Bob need to invest $\int_{0}^{x}y\mathrm{d}y=\frac{1}{2}x^{2}$
resources. Hence, Bob's utility function is $\sqrt{2(10-A)}+\sqrt{2(5-B)}$.
It is possible to plot both Alice's and Bob's actions (i.e., ways
to split up their resources between Alice and Bob) in a two-dimensional
plane in which the axes are the utility functions of Alice (\(x\)-Axis)
and Bob (\(y\)-Axis). Additionally, one can plot all combinations of actions
of Alice and Bob in one joint graph for both utilities (\Cref{fig:5}). The upper right boundary of the set of feasible utilities is the Pareto frontier---the set of utility vectors such that
no one can improve on their utility without making someone else worse
off.
\end{example}

\begin{figure}
\begin{centering}
\subfloat[The utilities produced by Alice's (blue) and Bob's (orange) actions
individually.]{\includegraphics[width=0.48\textwidth]{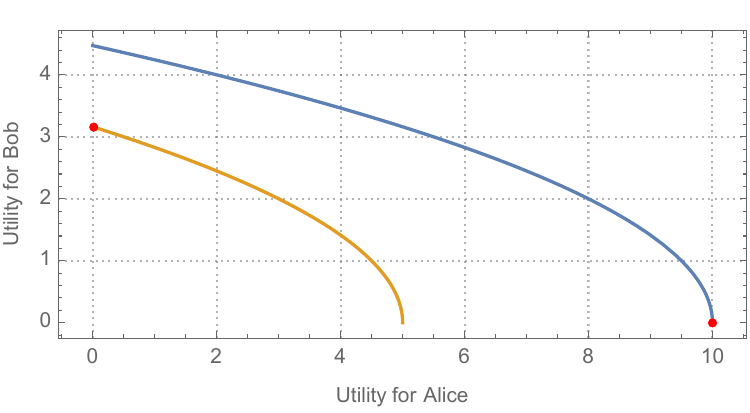}
}\hfill
\subfloat[The utilities produced by all possible combinations of Alice's and
Bob's actions.]{\includegraphics[width=0.48\textwidth]{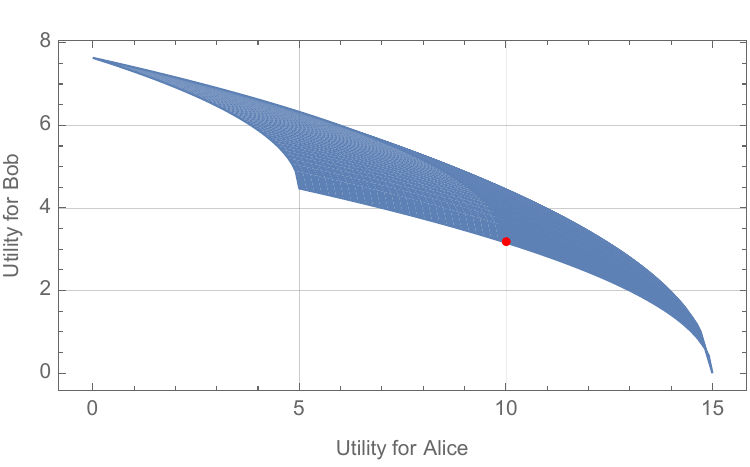}}
\end{centering}
\caption{Utilities in the bargaining problem from \Cref{exa:bargaining-case-1}. The point at which both Alice and Bob maximize their own utilities is indicated by a red dot.}
\label{fig:5}
\end{figure}

In this example, both agents maximizing their own utility function leads to a Pareto inferior outcome: the
point $(10,\sqrt{10})$, which does not lie on the Pareto frontier.
If, on the other hand, Bob and Alice are able to coordinate on a cooperative
combination of actions, this leaves both better off. There is the
question, though, \emph{which} point on the Pareto frontier they should
choose. In this section, I am considering this question in the case of ECL.

An interesting property of the Pareto frontier is that if Alice and Bob choose actions such that the slopes
of their individual Pareto frontiers---in this example, the slope
of the lines in \Cref{fig:5} (a)---at the point of their actions
are not the same, then the actions are not Pareto optimal. Regarding
the slope of the Pareto frontiers as marginal rates of substitution,
this is a well-known concept in economics. If the marginal rates of substitution for Alice and Bob are not the same, then both players can move in opposite directions on
their Pareto frontier to become jointly better off. One person can give
up some amount $x$ of utility for Alice to gain some amount $y$
of utility for Bob, and at the same time, the other person can give
up less than $y$ utility for Bob and gain more than $x$ for Alice,
such that jointly, the effect on both of their utilities is positive.

\subsection{Formal setup}
\label{ecl-bargaining-problem-setup}

\begin{defn}
A \emph{(complete information) bargaining game} is a 4-tuple \[B=(N,(A_{i})_{i\in N},(u_{i})_{i\in N},d),\]
where
\begin{itemize}
\item $N=\{1,\dots,n\}$ is the set of players;
\item $(A_{i})_{i\in N}$ is the tuple of finite sets of actions for
each player;
\item $(u_{i})_{i\in N},u_{i}\colon A\rightarrow\mathbb{R}$ is the tuple
of utility functions for all players, where $A=\prod_{i\in N}A_{i}$
is the set of outcomes, or the set of pure strategy profiles;
\item $d\in\mathbb{R}^{n}$ is the disagreement point.\label{def:ECL-bargaining-game}
\end{itemize}
\end{defn}

A bargaining game as defined above is a standard normal form game, with the addition of a disagreement point, which is needed to specify a default outcome that is realized when bargaining fails. 

In the following, I introduce some initial notation and definitions. As is standard, I write \(a_{-i}\in A_{-i}:=\prod_{j\in N, j\neq i}A_j\) and \((a_{-i},a_i)\) for the vector in which the \(i\)-th entry is \(a_i\in A_i\) and the remaining entries are given by \(a_{-i}\in A_{-i}\).

Given a bargaining game $B$, players
are able to randomize between actions. Let $\Sigma_i:=\Delta(A_{i})$ be the set of
probability distributions (identified with probability mass functions) over the actions in $A_{i}$. Then \(\sigma_i\in \Sigma_i\) is called a mixed strategy. Moreover, $\sigma \in\Sigma:=\prod_{i\in N}\Sigma_i$ is called a mixed strategy profile, and I write
\[
u_{i}(\sigma):=\sum_{a\in A}\left(\prod_{j\in N}\sigma_{j}(a_{j})\right)u_{i}(a)
\]
 for player $i$'s expected utility given mixed strategy profile
$\sigma\in\Sigma$.
I regard the mixed strategies as the options of the players. Note
that at this stage, strategies are always independent. Later I introduce
a different concept which involves possibly dependent joint distributions
over players' actions, which I call ``joint strategy distributions''.

Given a bargaining game \(B\), one can define
\[
F(B)=\{x\in\mathbb{R}^{n}\mid\exists\sigma\in\Sigma\colon\forall i\in N\colon x_{i}=u_{i}(\sigma)\}
\]
 as the \emph{feasible set}. This set is by construction a simplex, and hence
convex and compact. The feasible set contains, for all the possible
mixed strategy profiles that the players can choose, vectors that
specify the expected utilities for each player given that profile---i.e.,
the utility vectors that are feasible given the bargaining game $B$. I also
define
\[
H(B)=\{x\in F(B)\mid\forall y\in F(B):(\forall i\in N:y_{i}\geq x_{i})\Rightarrow y=x\}
\]
as the (strict) Pareto frontier of $F(B)$.

As an addendum to \Cref{def:ECL-bargaining-game}, I will assume in the following
that the disagreement point simply corresponds to one of the possible mixed strategies, and hence lies in the feasible set, i.e.,
$d\in F(B)$. Moreover, I assume that it is possible to achieve gains from trade for all players. I.e., there exists \(x\in F(B)\) such that \(x_i>d_i\) for all \(i\in N\). This simplifies matters and is not really a substantial restriction. To see this, note that if players cannot receive gains from trade, then it does not make sense for them to participate in ECL. Moreover, consider the set of players that can receive gains from trade without making other players worse off than their disagreement point. Then by convexity of \(F(B)\), there also exist outcomes that make all of these players better off simultaneously.

Utility functions are only specified up to positive affine transformation,
i.e., if there are utility functions $u,u'$ and there is $a\in\mathbb{R}_{>0}$,
$b\in\mathbb{R}$ such that $u=au'+b$, then these utility functions
imply exactly the same preference relation over $\Sigma$. I write
$u\sim u'$ to denote that two utility functions are equivalent in
this sense.

If one subtracts the disagreement point from a utility function,
then the resulting function is at least unique with respect
to the addition of a constant $b\in\mathbb{R}$. One can of course
still multiply the function by an arbitrary positive number.
The disagreement point may be just some outcome that would obtain
if everyone were to maximize their own utility function, but it could
also be a ``threat point'' where everyone takes the action which
produces the best subsequent bargaining solution. More on this in \Cref{subsec:The-question-of}, but for now I assume such a point is given.

In the ECL context,
we can assume that players live in completely causally disconnected universes. Hence, if all the players have
value systems that are additive in the universes, it does not make a difference to the
utility one player gets from another player what all the other players
do (this is false in general, e.g., if players were able to causally interact).
So it makes sense to give the following definition:
\begin{defn}[Additive separability] Utility functions are called \emph{additively
separable} if there are functions $u_{i,j}$ for $i,j\in N$ such that for all
\(i\in N\) and $a\in A$, we have
$u_{i}(a)=\sum_{j\in N}u_{i,j}(a_{j})$. A bargaining game \(B\) is called additively separable if the corresponding utility functions are additively separable.
\end{defn}

Unfortunately, this excludes some notable value systems. For instance, value systems which have diminishing
returns for some good across the multiverse, or value systems that care about the average happiness of all beings in the multiverse. However, it makes things
much easier. I will assume additive separability in many results.

In case utility functions are additively separable, it is possible
to write \[F(B)=\sum_{i\in N}F_{i}(B):=\left\{\sum_{i\in N}x_i\mid \forall j\in N\colon x_j\in F_j(B)\right\},\] where 
\[
F_{i}(B):=\{x\in\mathbb{R}^{n}\mid\exists\sigma_{i}\in\Sigma_{i}\colon\forall j\in N\colon x_{j}=u_{j,i}(\sigma_{i})\}
\]
 is the feasible set for player $i\in N$ and $u_{j,i}(\sigma_{i})=\sum_{a_{k}\in A_{i}}\sigma_{i}(a_{k})u_{j,i}(a_{k})$.
That is, for each player, there is an individual feasible set of utility vectors that this player can generate for all players, and the joint feasible
set consists of all the points $x$ that are sums of points in the
individual feasible sets. I define
\[
H_{i}(B)=\{x\in F_{i}(B)\mid\forall y\in F_{i}:(\forall j\in N:y_{j}\geq x_{j})\Rightarrow y=x\}
\]
as the strict Pareto frontier of $F_{i}(B)$.
This is the upper right boundary of
the utilities that the individual player $i\in N$ can contribute in
their part of the universe. Note that not every sum \(\sum_i x_i\) of points \(x_i\in H_i(B)\) on the individual Pareto frontiers will be Pareto optimal.

In \Cref{exa:bargaining-case-1}, utilities are additively
separable, and \Cref{fig:1} (a) depicts the two individual feasible sets, while (b) depicts the one combined feasible set. 
These feasible sets are not valid in the sense of the above definition,
since they are not convex and they are not simplices. However, we can relax the assumption that \(F(B)\) is a simplex by allowing for a bargaining problem to be directly defined as a tuple \(B=(N,F,d)\), where \(N\) is the set of players, \(F\) the feasible set, and \(d\in F\) the disagreement point. In this case, \(F\) still has to be compact and convex, but it need not be a simplex (e.g., if the underlying set of actions is continuous, as in \Cref{exa:bargaining-case-1}). Convexity and compactness are required so that we can apply bargaining theory (e.g., to ensure that strictly convex functions have unique minima).

Assuming additive separability, it is practical to just identify
the space of actions of player $i\in N$ with their feasible set
$F_{i}(B)$. In that case, we can define a bargaining problem as a tuple $B=(N,(F_{i})_{i\in N},d)$ of a set of players \(N\), individual feasible sets \(F_i\) for each player, and a disagreement point \(d\). Then we have $F(B)=\sum_{i\in N}F_{i}$. Here, the $F_{i}$ have to be compact, convex sets, but need not be simplices. Particularly, there are feasible sets \(F_i\) such that the
$H_i$ are smooth, \(n-1\)-dimensional manifolds, which we will assume in some results below.

Lastly, it is useful to define \[\mathcal{F}^{N}=\{F(B)\mid B\text{ is a bargaining game with set of players N\}}\]
as the set of all possible sets of feasible utilities for the set
of players $N$ and $\mathcal{F}=\bigcup_{N\in\mathcal{P}}\mathcal{F}^{N}$
as the set of all possible feasible sets, where $\mathcal{P}=\{\{1,\dotsc,n\}\mid n\in\mathbb{N}\}$
is the set of all finite sets of agents. Moreover,
\[\Upsilon=\{(F,d)\mid N\in\mathcal{P},F\in\mathcal{F}^{N},d\in F\}.\]
With these definitions, a bargaining solution is a function $\mu$
on $\Upsilon$ such that $\mu(F,d)\in F$. That is, it takes a feasible
set and a disagreement point, and outputs a unique point in the feasible
set as solution. 

\subsection{Normalizing utility functions\label{subsec:Normalizing-utility-functions}}

One possible approach to determining the actions of individual players
in the bargaining problem posed by ECL is maximizing some compromise utility function \parencite[Sec. 2.8]{Oesterheld2017-qg}. In particular, one may start by normalizing individual utility functions via shifting and scaling, and then maximize a weighted sum of them. Maximizing a sum picks out a specific point or affine subset of the Pareto frontier. Note that this correspondence also works the other way around---for every point on the Pareto frontier, we can derive weights such that the point maximizes the corresponding weighted sum. In this section, we will first argue why all players have to maximize the same sum to reach a Pareto optimal agreement. Second, we motivate the use of bargaining solutions that directly pick out points on the Pareto frontier, by arguing against an approach that starts by normalizing utility functions.

One motivation behind the idea of maximizing a weighted sum of utility functions is Harsanyi's utilitarian theorem \parencite{hammond1992harsanyi}. Assume that a player wants to maximize a compromise
utility function $u^{*}$ that also incorporates other players' preferences.
A very plausible axiom in this case is the following:
\begin{ax}
Let $\alpha,\beta\in\Sigma$ and $u_{i}(\alpha)\geq u_{i}(\beta)$
for all $i\in N$. Then $u^{*}(\alpha)\geq u^{*}(\beta).$\label{axm:Pareto-condition}
\end{ax}

This is a kind of Pareto optimality condition. If one mixed strategy
profile is at least as good for everyone as another mixed strategy
profile, then it should also be at least as good for the new utility
function $u^{*}$. According to a version of Harsanyi's utilitarian
theorem, it follows from this axiom that $u^{*}$ is just a weighted
sum of the utility functions of individual players:
\begin{thm}
\parencites{Resnik1983}{Fishburn1984-FISOHU} Let $u^{*}$ satisfy Axiom \ref{axm:Pareto-condition}.
Then there are weights $\lambda_{1},\dots,\lambda_{n}\in\mathbb{R}_{\geq0}$
such that
\begin{equation}
u^{*}\sim\sum\lambda_{i}u_{i}.
\end{equation}
\label{thm:aggregation-theorem}
\end{thm}

This result says that a player that wants to pursue a compromise and respect the Pareto axiom has to maximize some sum of utility functions. But it leaves open how
to choose the weights in this sum of utility functions.

Assuming additive
separability, we can also show that, to get a Pareto optimal outcome, different players have to maximize the same weighted sum of utility functions. This
follows from the fact that maximizing a weighted sum picks out the
point on the Pareto frontier where the slope of the frontier corresponds
exactly to the weights in the sum. But if two players choose points
on their frontiers with different slopes, there are gains from trade
left on the table. As mentioned in \Cref{subsec:example-alice-bob}, in a Pareto-optimal outcome,
the slopes of the frontiers, i.e., marginal rates of substitution,
have to be identical. Otherwise, both players could jointly move in
opposite directions on the frontier such that both gain more than
they lose.
\begin{thm}\label{thm-equal-weights}
Let $B=(N,(F_{i})_{i\in N},d)$ be an additively separable bargaining game. Assume that there are weight
vectors $\mu_{i}\in\mathbb{R}_{\geq0}^{n}$ for
$i\in N$ such that player $i\in N$ takes an action $x_{i}\in F_{i}$
that maximizes $\sum_{j\in N}\mu_{i,j}x_{i,j}$. Then \begin{enumerate}
\item[(i)] If $\mu_{1,i}=\dots=\mu_{n,i}>0$ 
for all $i\in N$, then $\sum_{i\in N}x_{i}$
is Pareto optimal.
\item[(ii)] If the boundaries $\partial F_{i}$ are smooth \(n-1\)-dimensional manifolds and there exist \(i,j\) such that \(\mu_i\neq \mu_j\), then \(\sum_ix_i\) is not Pareto optimal.
\end{enumerate}
\end{thm}

\begin{proof}
To begin, note that w.l.o.g. we can assume that for all \(i\in N\), we have \(\Vert\mu_i\Vert_2=1\). This is because we assume \(\mu_i\neq 0\) for both (i) and (ii), and we can rescale \(\mu_i\) to have norm \(1\) without changing the optimum \(x_i\).

Now, to prove (i), assume that $\mu_{1}=\dotsb=\mu_{n}$. We have
\[
\sum_{i\in N}\sum_{j\in N}\mu_{i,j}x_{i,j}=\sum_{i\in N}\max_{y_{i}\in F_{i}(B)}\sum_{j\in N}\mu_{i,j}y_{i,j}=\max_{y\in\prod_{i\in N}F_{i}(B)}\sum_{j\in N}\mu_{1,j}\sum_{i\in N}y_{i,j}=\max_{y\in F(B)}\sum_{j\in N}\mu_{1,j}y_{j}.
\]
If a point is a solution to a maximization problem $\max_{y\in F(B)}\sum_{j\in N}\mu_{1,j}y_{j}$
such that $\mu_{1,i}>0$ for all $i$, then we cannot improve the utilities for one of the players without making anyone else worse off. Hence, the point is Pareto optimal.

Next, we show (ii) via contraposition. Assume that $\sum_{i\in N}x_{i}$
is Pareto optimal. 
For any Pareto
optimal point, there is a weight vector $\nu\in\mathbb{R}_{\geq0}^{n}$,
$\Vert\nu\Vert=1$ such that $\sum_{i\in N}x_{i}\in\argmax_{y\in F}\nu^{\top}y$.
Moreover, since the boundaries \(\partial F_i\) are smooth, we can define smooth functions \(h_i\colon \mathbb{R}^n\rightarrow\mathbb{R}\) such that \(\partial F_i=\{x\mid h_i(x)=0\}\), i.e., the boundaries \(\partial F_i\) are the level sets \(h_i=0\), and such that for any \(x\in \partial F_i\), \(\nabla h_i(x)\) with \(\Vert \nabla h_i(x)\Vert =1\) is a normal vector to the boundary \(\partial F_i\) at \(x\).

Then we have $h_{i}(x_{i})=0$ for $i\in N$. Hence,
$x:=(x_{1},\dots,x_{n})$ is a solution to the problem of maximizing
$f\colon\prod_{i\in N}F_{i}\rightarrow\mathbb{R},y\mapsto\nu^{\top}\sum_{i\in N}y_{i}$
under the side-constraint that $\mathcal{H}(y)=0$ where $\mathcal{H}\colon\prod_{i\in N}F_{i}\rightarrow\mathbb{R}^{n}$
such that $\mathcal{H}_{i}\colon\prod_{j\in N}F_{j}\rightarrow\mathbb{R},y\mapsto h_{i}(y_{i})$.
According to the method of Lagrange multipliers, there hence are $\lambda_{j}\in\mathbb{R}$
for $j\in N$ such that
\begin{equation}
\partial_{i}f(x)=\sum_{j\in N}\lambda_{j}\partial_{i}\mathcal{H}_{j}(x),
\end{equation}
for all $i\in N$. Since $\partial_{i}\mathcal{H}_{j}(y)=\delta_{i,j}\nabla h_{i}(y_{i})$ (where \(\delta_{i,j}\) is the Kronecker delta),
it follows that 
\begin{equation}
\nu=\partial_{i}f(x)=\sum_{j\in N}\lambda_{j}\partial_{i}\mathcal{H}_{j}(x)=\lambda_{i}\nabla h_{i}(x_{i}).
\end{equation}
In particular, \(\lambda_i\neq 0\).

Moreover, by assumption, for all $i\in N$, $x_{i}$ maximizes $g_{i}\colon F_{i}\rightarrow\mathbb{R},y_{i}\mapsto\sum_{j\in N}\mu_{i,j}y_{i,j}$
under the side-constraint that $h_{i}(y_{i})=0$. Hence, it follows that
there is $\lambda'_{i}\in\mathbb{R}$ such that
\begin{equation}
\mu_{i}=\nabla g_{i}(x_{i})=\lambda'_{i}\nabla h_{i}(x_{i}).
\end{equation}

Putting everything together, it follows that
\[\mu_{i}=\lambda'_i\nabla_i(x_i)
=\frac{\lambda'_{i}}{\lambda_{i}}\nu\]
for all $i\in N$. Since $\Vert\mu_{i}\Vert=1=\Vert\mu_{j}\Vert$,
it is $\mu_{i}=\frac{\nu}{\Vert\nu\Vert}=\mu_{j}$. This shows the contrapositive.
\end{proof}
I believe the result carries over to some degree to a game with non-smooth feasible sets. If there are kinks
in the Pareto frontiers, then at these points, it will be possible to maximize
slightly different weights and still achieve a Pareto optimal outcome,
since several different maximized weighted sums or normal vectors of the
frontier will correspond to the same point.

Since there exist bargaining problems for which the boundaries \(\partial F_i\) are smooth \(n-1\)-dimensional manifolds (e.g., in the trivial case in which the \(F_i\) are \(n\)-dimensional balls), this result shows that there exist problems for which maximizing different weighted sums would result in Pareto suboptimal outcomes.
\begin{comment}
\begin{cor}
Assume that there are weight vectors $\lambda_{i}\in\mathbb{R}_{\geq0}^{n}$,
$\Vert\lambda_{i}\Vert_{1}=1$ for all $i\in N$, such that $\lambda_{i,k}\neq\lambda_{j,k}$
for some $i,j,k\in N$. Then there is an ECL bargaining game $B$
such that if all players $i\in N$ choose to play a mixed strategy
that corresponds to a point $x_{i}\in F_{i}(B)$ such that
\[
x_{i}\in\argmax_{y_{i}\in F_{i}(B)}\lambda_{i}^{\top}y_{j},
\]
it follows that $x=\sum_{i\in N}x_{i}$ is not Pareto optimal.
\end{cor}

\begin{proof}
Follows directly from Theorem 6.
\end{proof}
\end{comment}
Together with the utilitarian theorem, we can conclude that all superrationalists should maximize some common sum of utility
functions. This leaves open the question of \emph{which} weighted
sum to maximize.

One suggestion by \textcite[][sec.\ 2.8.5]{Oesterheld2017-qg} is to choose weights
that normalize utility functions according to their variance. Variance
normalization is also supported by \textcite{MacAskill2020-MACSNM-2}, who
set up a scenario in which players submit utility functions to cast
their vote on a social utility function. Using
relatively strong ignorance assumptions, they show that
normalizing the variance of utility functions leads all players to
have equal voting power; that is, they are all equally likely to change
the option that is best under the social utility function.

For my setting, I think this approach does not work well. This is because
under some circumstances, variance normalization can lead one player
to expect negative gains from trade, and I think that one important
requirement for a compromise is that everyone gets positive gains
from trade. This is true even if players that implement an updateless decision
theory \parencite{dai2009updateless} or have only very little prior knowledge about ECL. Players will have \emph{some}
(prior) beliefs to determine whether a trade will be positive. Given
these beliefs, the trade has to be positive. Otherwise, rational players
will decide not to engage in the compromise.

\begin{example}\label{example-rationality}
As an example where variance normalization does not work, take a game
with players $1,2$ and action sets $A_{1}=\{a_{1},b_{1}\}$ and $A_{2}=\{a_{2},b_{2}\}$, with utilities 
as depicted in Tables~\ref{tbl:2}--\ref{tbl:5-1}. Note that utility functions are additively separable.
\begin{table}
\begin{centering}
\subfloat[]{%
\begin{tabular}{|c|c|c|}
\hline 
 & $a_{1}$  & $b_{1}$\tabularnewline
\hline 
\hline 
$u_{1}$  & $0$  & $-1$\tabularnewline
\hline 
$u_{2}$  & $-3$  & $0$\tabularnewline
\hline 
\end{tabular}

}\subfloat[]{%
\begin{tabular}{|c|c|c|}
\hline 
 & $a_{2}$  & $b_{2}$\tabularnewline
\hline 
\hline 
$u_{1}$  & $0$  & $-3$\tabularnewline
\hline 
$u_{2}$  & $3$  & $2$\tabularnewline
\hline 
\end{tabular}

}
\par\end{centering}
\caption{Actions of player $1,2$ and additively separable utilities provided to either player in \Cref{example-rationality}.}
\label{tbl:2} 
\end{table}

\begin{table}
\centering{}%
\begin{tabular}{|c|c|c|}
\hline 
 & $u_{1}$  & $u_{2}$\tabularnewline
\hline 
\hline 
$a_{1},a_{2}$  & $0$  & $0$\tabularnewline
\hline 
$a_{1},b_{2}$  & $-3$  & $-1$\tabularnewline
\hline 
$b_{1},a_{2}$  & $-1$  & $3$\tabularnewline
\hline 
$b_{1,}b_{2}$  & $-4$  & $2$\tabularnewline
\hline 
\end{tabular}\caption{The utilities of all possible outcomes for both players in \Cref{example-rationality}.}
\label{tbl:5-1} 
\end{table}

\begin{table}
\centering{}%
\begin{tabular}{|c|c|c|}
\hline 
 & $u_{1}$  & $u_{2}$\tabularnewline
\hline 
\hline 
$a_{1},a_{2}$  & $0.2$  & $-0.1$\tabularnewline
\hline 
$a_{1},b_{2}$  & $-0.1$  & $-0.2$\tabularnewline
\hline 
$b_{1},a_{2}$  & $0.1$  & $0.2$\tabularnewline
\hline 
$b_{1,}b_{2}$  & $-0.2$  & $0.1$\tabularnewline
\hline 
\end{tabular}\caption{Normalized utilities}
\label{tbl:4} 
\end{table}

Here, the
dominant option for both players is $(a_{1},a_{2})$. To normalize according to variance, we have to determine a distribution over actions. Here, I assume  a uniform distribution. Then the mean
for player $1$ is $\mu_{1}=-2$, and for player $2$ it is $\mu_{2}=1$.
We subtract this mean from the utilities
of all the players, then divide the utilities by their variance.
The variance is $\sigma_{i}^{2}=\sum_{x\in A_{1}\times A_{2}}(u_{i}(x)-\mu_{i})^{2}$
for player $i$, which is $10$ for both players. The normalized
utilities are as depicted in \Cref{tbl:4}. 

Here, $b_{1},a_{2}$ maximizes the sum of normalized utility
functions. But this leaves player $1$ worse off than without a compromise.
\end{example}
Though I do not investigate this in more detail, I believe problems may arise with all methods that do not directly pick a point on the Pareto frontier as a compromise solution. It would still be interesting to investigate under which conditions variance normalization or other normalization methods give all agents positive gains from trade, but I will not pursue this approach here further.

In the following, I will consider solutions that directly pick out
a point in the feasible set. Once such a point is given, it is possible to derive weights for utility functions such that the point maximizes the corresponding
weighted sum. If the Pareto frontier is differentiable at this point,
it follows from the proof of \Cref{thm-equal-weights} that these weights are unique. I will not delve into the issue of translations between maximizing
weighted sums and points on the Pareto frontier further in this report. (Though I will address  the related issue of uniqueness of the individual mixed strategies maximizing a particular weighted sum in \Cref{uniqueness}.)

\subsection{Bargaining theory}
\label{bargaining-theory}
Here, I briefly review the existing literature on bargaining theory. Since there exists a large literature on bargaining, it seems likely to me that the most plausible and easy to find solutions to
bargaining problems have already been discovered. There are
two main approaches to bargaining problems:
\begin{enumerate}
\item The axiomatic or normative approach, which involves specifying plausible axioms for bargaining  solutions and proving that these
axioms are equivalent to some choice of a bargaining solution. 
\item The noncooperative or positive approach, which involves specifying a bargaining
game and analyzing the equilibria of the game.
\end{enumerate}
Both approaches are interesting from an ECL perspective. First, the axiomatic
approach is interesting because a solution that has any chance of giving an agent evidence
that others are pursuing the same solution must be parsimonious. This seems more likely if the bargaining solution depends on plausible axioms. Moreover, it is an argument for relying on the existing
literature, because solutions that have already been found by economists are ceteris
paribus also solutions that are more likely to be found by other superrationalists.\footnote{Note that this argument is informal, assuming dependencies of the sort ``if I look for plausible axioms
and find them, the other agent will do the same and find the same
axioms''. It is not backed up by some equilibrium or game-theoretic
analysis but a judgement of psychological plausibility.}

Second, modeling the situation using noncooperative game theory can provide
one with evidence in favor of a particular solution being more likely
to result from real-world bargaining situations. This has only been
done for causal bargaining, but hopefully acausal
bargaining theory would give similar results to the causal setting.
Some work points in the direction that such transfer may be possible \parencite{oesterheld2019robust}.

In the following, I will turn to the axiomatic approach and review some of the desiderata from the literature. There are several
plausible axioms for a bargaining solution:
\begin{ax}
\label{axm:1}Let $\mu$ be some bargaining solution, $B$ a 
bargaining game, $F(B)$ its feasible set with Pareto frontier $H$,
and $d\in F(G)$ its disagreement point.
\begin{enumerate}
\item[(1)] (Weak) Individual rationality. The solution should give everyone non-negative
gains from trade. So $\mu_{i}(F(B),d)\geq d_{i}$ for all $i\in N$.
\item[(2)] (Strong) Pareto optimality: $\mu(F(B),d)\in H(B)$.
\item[(3)] Invariance to affine transformations of utility functions. Let $\phi\colon\mathbb{R}^{n}\rightarrow\mathbb{R}^{n}$ such that $\phi(x)=[\lambda_{1}x_{1},\dots,\lambda_{n}x_{n}]^{\top}+y$
for some $\lambda_{i}\in\mathbb{R}_{>0},y\in\mathbb{R}^{n}$. Then
$\mu(\phi(F(B)),\phi(d))=\phi(\mu(F(B),d))$.
\item[(4)] Anonymity. For any permutation $\pi$ on $N$, define $\pi(x)=(x_{\pi(1)},\dots,x_{\pi(n)})$
for $x\in\mathbb{R}^{n}$. Then $\pi(\mu(F(B),d))=\mu(\pi(F(B)),\pi(d))$. 
\end{enumerate}
\end{ax}

I argued for (1) in the preceding section. (2) seems fairly plausible
on the grounds that ECL should not leave any possible gains from trade on the
table. (3) is plausible since the solution should not depend on which
representative we pick out of the equivalence class of utility functions
which give rise to the same cardinal ranking over mixed strategy profiles.

(4) tells us that the bargaining solution should be equivariant: the payoffs assigned to players should stay the same, even if we change their indices. While anonymity is plausible, this definition unfortunately ignores the individual feasible
sets $F_{i}(B)$ for each player $i$ that exist in the additively separable case. This means that players may have to be treated equally, even if their contributions \(F_i(B)\) to the overall payoffs differ. However, it seems that the relative size of the contributions
should make a difference for fairness. We will turn to this fairness point again 
in \Cref{fairness-and-coalitional}.

The axioms outlined above do not yet uniquely specify a bargaining solution. However, they do so after adding a fifth axiom. In the following sections, I will turn to two popular suggestions for fifth axioms, which correspond to two different bargaining solutions, the Nash bargaining solution (NBS) and the Kalai Smorodinsky bargaining solution (KSBS).  In \Cref{appendix-armstrong-solution}, I discuss an additional solution proposed by \textcite{Armstrong2013}. I do not focus on it here since it violates individual rationality. Since the main parts of this report were written in 2018, I do not consider more recent work such as \textcite{diffractor2022rose}.

\subsubsection{The Nash bargaining solution}
\label{nash-bargaining-solution}

\begin{defn}The Nash bargaining solution (NBS) \parencite{Nash1950-vg,Harsanyi1972,Lensberg1988-lf,Okada2010-ql,Anbarci2013-yd,Roth1979a}
is the point in $F(B)$ which maximizes the product of the players'
gains from trade, also called the Nash welfare.
\begin{equation}
\mu(F(B),d):=\argmax_{x\in F(B)^{\geq d}}\prod_{i\in N}(x_{i}-d_{i}),
\end{equation}
where $F(G)^{\geq d}:=\{x\in F(G)\mid\forall i\in N\colon x_{i}\geq d_{i}\}$.
\end{defn}
Since $F(B)^{\geq d}$ is compact and convex and there exists $x\in F(B)$
such that $x_i>d_i$ for all \(i\in N\) by assumption, this point exists and is unique. It is also called the symmetric NBS.

\begin{example}
Applying the NBS to \Cref{exa:bargaining-case-1}, using the point \((10,\sqrt{10})\) as a disagreement point,
we get the optimization problem
\[
\max_{A\in[0,10],B\in[0,5]}(A+B-10)(\sqrt{2(10-A)}+\sqrt{2(5-B)}-\sqrt{10}).
\]
This has a maximum at
$A\approx 8.15$, \(B\approx 3.15\), which I have plotted as a green dot in \Cref{fig:6}.

\begin{figure}
\begin{centering}
\includegraphics[width=0.5\textwidth]{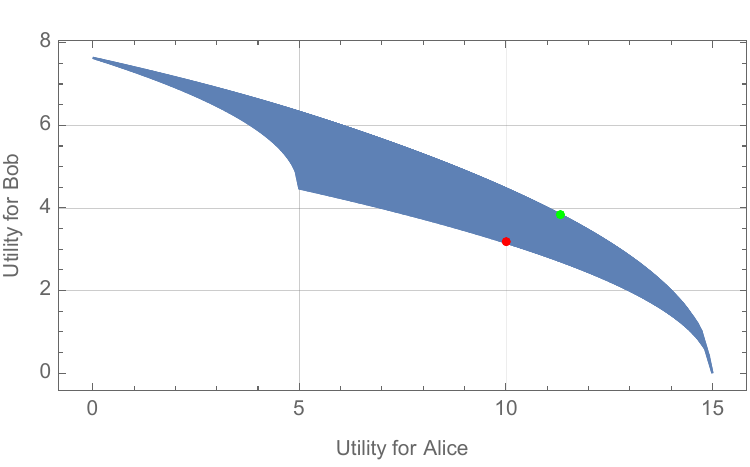}\caption{Disagreement point as a red dot and the NBS as a green dot, for the bargaining problem from \Cref{exa:bargaining-case-1}.}
\label{fig:6}
\par\end{centering}
\end{figure}
\end{example}

There are different axiomatizations of the NBS (i.e., choices for
the fifth axiom) which are equivalent. In the many-player case, an
axiom which I find plausible is due to \textcite{Lensberg1988-lf}.
It has an intuitive geometric interpretation but its mathematical formulation is quite technical. 

Let $P,Q\subseteq N,P\subseteq Q$. Let $x_{p}$ denote the projection
of $x\in\mathbb{R}^{Q}$ onto $\mathbb{R}^{P}$. Let $H_{P}^{x}=\{y\in\mathbb{R}^{Q}\mid y_{Q\setminus P}=x_{Q\setminus P}\}$.
Given $C\subseteq\mathbb{R}^{Q}$ and $x\in C$, denote $t_{P}^{x}(C)$ for
the projection of $H_{P}^{x}\cap C$ onto $\mathbb{R}^{P}$.
\begin{ax}\label{ax:stability}
Multilateral stability. If $P\subseteq N$, $\mu(F(G),d)=x$,
and $D=t_{P}^{x}(F(G))$, then \[{\mu(D,d_{P})=x_{P}}.\]
\end{ax}

\Cref{fig:1} shows an illustration from \textcite{Lensberg1988-lf}. Basically, the idea is that if one fixes the payoffs for a subset
$N\setminus P$ of the players and lets the remaining players $P$
renegotiate their solution on the new feasible set $D=t_{P}^{x}(F(G))$
that results if you fix the payoffs for all the players in $N\setminus P$
and project this feasible set onto $\mathbb{R}^{P}$, the result should be the same as the solution of the
entire problem, projected onto $\mathbb{R}^{P}$. I find this
axiom very appealing.
\begin{figure}
\includegraphics[width=1\textwidth]{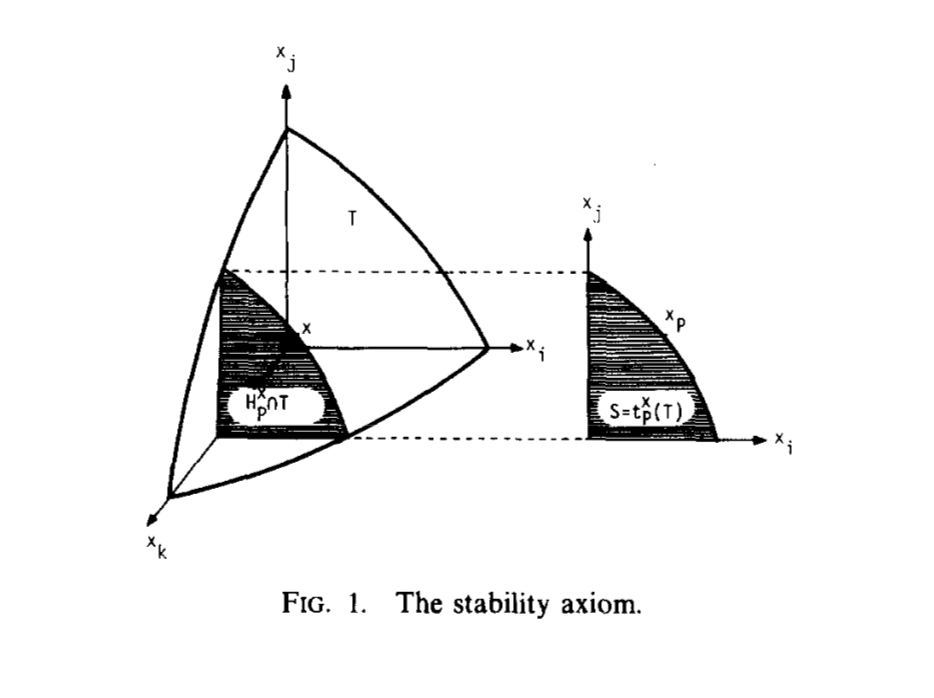}

\caption{Illustration of the stability axiom by \textcite{Lensberg1988-lf}.}

\label{fig:1} 
\end{figure}

\begin{thm}[\cite{Lensberg1988-lf}]
    \Cref{axm:1} (Pareto optimality, Invariance to affine transformations, Anonymity) and \Cref{ax:stability} (Multilateral stability) together are necessary and sufficient to specify the Nash bargaining solution.
\end{thm}

Assuming \Cref{axm:1}, multilateral stability is interchangeable with the Independence of irrelevant alternatives axiom:
\begin{ax}
[Independence of irrelevant
alternatives]\label{axm:Independence-of-irrelevant} Let $B,B'$ be two bargaining games such that $F(B)\subseteq F(B')$,
$\mu(F(B'),d)\in F(B)$. Then $\mu(F(B'),d)=\mu(F(B),d)$. 
\end{ax}
This also seems like an appealing desideratum. There are several further axiomatizations
of the NBS (in the 2 or \(n\)-player case).

There also exists an asymmetric version of the NBS \parencite{Kalai1977,Roth1979a}. Either Pareto optimality or strong Individual rationality
(i.e., $\mu_{i}(F(B),d)>d_{i}$ for all $i\in N$) in combination
with Invariance to affine transformations and Independence of irrelevant alternatives
are necessary and sufficient to characterize all functions 
\begin{equation}
\argmax_{x\in F(G)^{\geq d}}\prod_{i\in N}(x_{i}-d_{i})^{\alpha_{i}},
\end{equation}
where $\alpha_{i}>0$ and $\sum_{i\in N}\alpha_{i}=1$. (I am not sure
whether this would also work with Multilateral stability instead of
Independence of irrelevant alternatives.)

The NBS is also supported by several noncooperative bargaining models \parencites{Nash1953}{Binmore1986}{Anbarci2013-yd}{BRANGEWITZ2013224}{Okada2010-ql}.
The most common one is a version of
the alternating offers model by \textcite{Rubinstein1982-vw}. Here, players take turns in making offers, and the other
party (or, in the multilateral case, all other players) can reject
or accept the offer. Players are impatient: either there is a chance at each step that bargaining
breaks down, or the players discount their utilities over time. In this game, there is a unique subgame perfect equilibrium. The limit of this equilibrium as the probability of breakdown approaches zero is the NBS \parencite{Binmore1986}.

\subsubsection{The Kalai-Smorodinsky bargaining solution}
\label{subsec:The-Kalai-Smorodinsky}

The KSBS \parencite{Kalai1975-yv} is a more recent alternative to
the NBS. I think it is less suitable for ECL. In the KSBS, the utility
functions are normalized such that 0 is the disagreement point and
1 is the ideal point---the best possible payoff for the agent (among
all payoffs in $F(B)$). Then the bargaining solution is the point
where the line from the zero point to the point where everyone has
1 utility intersects with the Pareto frontier. This means that the
solution is the point on the Pareto frontier at which ratios between players' utilities are equal to the ratios between their ideal utilities. Formally, if $U_{i}(B)$ is the best
possible attainable point for $i$, then the solution is the point
$x\in H(B)$ such that 
\begin{equation}
\frac{U_{i}(B)-d_{i}}{U_{j}(B)-d_{j}}=\frac{x_{i}-d_{i}}{x_{j}-d_{j}}
\end{equation}
for all $i,j\in N$.

Apparently, there are some problems with generalizing the KSBS to
n-player games \parencite{Roth1979b}. One needs several axioms. One
possible way to axiomatize it in this case is via the following axioms
(in addition to \Cref{axm:1}) \parencite{Karos2018generalization}.
To make the definitions easier, assume for now that $d=0$
(if this is not the case, just subtract $d$ from all points in the
feasible set). Then a bargaining solution is just a function of the
feasible set, assuming $d=0$.
\begin{ax}[Individual monotonicity] $\mu_{i}(F(B))\le\mu_{i}(F(B'))$ for all
$i\in N$ and all problems $B$, $B'$ with $F(B)\subseteq F(B')$,
$U_{i}(B)\leq U_{i}(B')$ and $U_{j}(B)=U_{j}(B')$ for all $j\neq i$. 
\end{ax}

That is, if someone's ideal point is greater in $B$ than in $B'$, then,
all else equal, their bargaining solution should also be greater in
$B$ than in $B'$.
\begin{ax}[Homogeneous ideal independence of irrelevant alternatives] $\mu(F(B))=\mu(F(B'))$
for all bargaining problems $B,B'$ with $F(B)\subseteq F(B')$, $\mu(F(B))\in F(B')$,
and $U(B)=rU(B')$ for some $r\le1$. 
\end{ax}

This is a weakened version of independence of irrelevant alternatives
which requires the ratios of the ideal points to be equal for the
axiom to apply. 
\begin{ax}[Midpoint domination] For all bargaining problems \(B\) and any player \(i\), we have $\mu_i(F(B))\geq\frac{1}{n}U_i(B)$.
\end{ax}
This axiom is also known as Proportional fairness.\footnote{\url{https://en.wikipedia.org/wiki/Proportional_division}}

The KSBS is supposed to be fairer than the NBS, in the sense that if someone has better options (their ideal point
is better), then they should be left better off in bargaining. This
is not the case in the NBS but is the case for the KSBS. However, I disagree with this notion of fairness. To me, fairness in the two-player case is concerned with splitting
the gains from trade equally or according to differences in power.\footnote{In the Rubinstein bargaining model, one can derive bargaining power
from players' discount rates. If a player 
has a higher time discount, they have a weaker bargaining position.}
But splitting gains from trade equally only makes sense in a transferable utility game, i.e., a game in which there is a common currency of money or resources which has equal utility for both players. Since ECL deals with arbitrary utility functions, we cannot in general assess fairness in the same way here.

An important aspect of fairness is the idea
that there should not be one player or a group of players that only contribute very little to the ECL-compromise, while gaining a lot from it. This type of fairness can be ensured in a coalitional game by requiring that the solution is coalitionally stable. If a player contributes little, then a coalition
of players can split off such that all players in the coalition are
better off, making the solution unstable. I will discuss this in \Cref{fairness-and-coalitional} and conclude that the KSBS does not fare better than the NBS in this respect.

I think there is a problem with the KSBS that arises
if several agents have the same utility function. Consider a case
with players $1,2,3$ and utility functions $u_{1},u_{2},u_{3}$,
where players $1,2$ have the same utility function. Intially, players $1,2,3$
all have $1$ utility, so $d=(2,2,1)$ (since $1$ and $2$
benefit each other). Now they are trying to decide how to split a
surplus of $1$ utility that arises from cooperating. The best achievable utilities are $b_{1,2}=3$
and $b_{3}=2$. Hence, 
\[
\frac{b_{1}-d_{1}}{b_{2}-d_{2}}=\frac{b_{2}-d_{1}}{b_{3}-d_{2}}=\frac{b_{1}-d_{1}}{b_{2}-d_{2}}=1,
\]
so the ratios of utilities minus the default points in the chosen
outcome have to be equal. Hence, the KSBS chooses $(2.5,2.5,1.5)$.
But this seems wrong: Players \(1\) and \(2\) have only received half of the utility, even though there are two of them. If they had been two players with distinct goals, then they would have each gotten one third of the utility, giving \(1\) and \(2\) a total of \(2/3\).
The NBS, since it is maximizing a product, is instead skewed towards
$1$ and $2$ and chooses the point $\approx(2.7,2.7,2.3)$, effectively giving each player one third of the utility surplus.

Lastly, another reason to prefer the NBS over the KSBS is that it seems
to lack the widespread support via noncooperative models and plausible
axiomatizations that the NBS has. The axioms for the KSBS in the multilateral
case seem much more contrived than those for the NBS, which makes
it less plausible as a multiverse-wide Schelling point. 
Although there apparently are some noncooperative bargaining models
supporting the KSBS \parencite{Anbarci1997}, the support for the
NBS seems greater to me. Relatedly, Google scholar searches for
word combinations such as ``Nash bargaining noncooperative'' and
``Kalai smorodinsky bargaining noncooperative'', or for the names of the solutions, consistently turn up more than
ten times as many papers for the NBS. Admittedly, there may be
some path dependency or founder's effect here. Nash is a more prominent
name and the NBS was the first published solution. Still,
it seems reasonable that a solution like the NBS---simply maximizing
the product of gains from trade---will be discovered first and thus considered a Schelling point in many parts of the multiverse.

\subsection{Observations}\label{observations-bargaining}
Here, I make some initial observations about the bargaining model and discuss potential issues. First, I discuss the question of whether the actions corresponding to a point on the Pareto frontier are unique (\Cref{uniqueness}). If not, this could lead to a coordination problem. I give an example where the decomposition is not unique, show that it is unique if utilities are additively separable and the feasible set is strictly convex, and argue that this should not be an issue in practice. Second, I make some basic observations about gains from trade given additively separable utilities, based on marginal rates of substitutions between different utility functions on individual Pareto frontiers (\Cref{possible-gains-from-trade-bargaining}). I conclude with remarks on trade between more than two players and continuity of the NBS (\Cref{further-observations}).

\subsubsection{Uniqueness of the actions corresponding to bargaining solutions}
\label{uniqueness}

The NBS provides players with a unique compromise point $x\in H(B)$.
The question arises whether this leaves all players with
a clear instruction on which action to take. There may be
several mixed strategy profiles in $\Sigma$ which correspond to $x$.
Then, if players cannot coordinate, the actually chosen outcome may
differ from the NBS. In principle, this outcome could even be worse than the disagreement point
for some.
\begin{example}\label{example-non-uniqueness}
As an example, take the game with individual Pareto frontiers $H_{1}=H_{2}=\{x\in\mathbb{R}^{2}\mid x_{1}+x_{2}=1\}$
and $d=0$ (Figure \ref{fig:13}). Apparently the overall Pareto frontier is $H=\{x\in\mathbb{R}^{2}\mid x_{1}+x_{2}=2\}$
and the point $(1,1)\in H$ is the NBS. Then
given any action combination $a,b\in H_{1},H_{2}$ such that $a_{1}=b_{2}$
and $a_{2}=b_{1}$, it is $a_{1}+b_{1}=a_{2}+b_{2}=1$. Hence, any
such combination corresponds to $(1,1)$ and maximizes the product
of gains.

\begin{figure}
\begin{centering}
\includegraphics[width=0.5\textwidth]{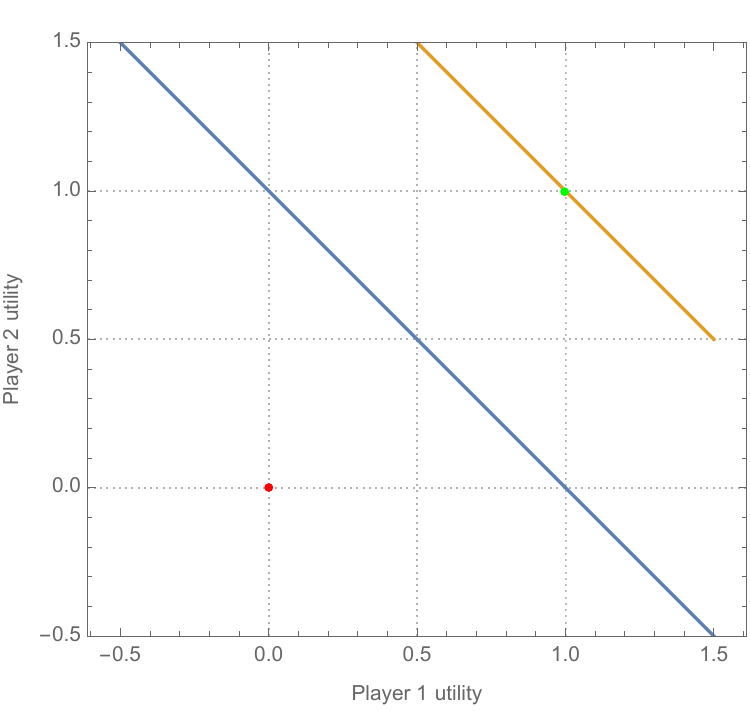}\caption{Pareto frontier \(H\) (in orange), individual Pareto frontiers $H_{1}=H_{2}$ (blue), the disagreement point (red dot), and the NBS (green dot), for \Cref{example-non-uniqueness}.}\label{fig:13}
\par\end{centering}

\end{figure}

Now, if player $1$ chooses $(2,-1)$, and the other player chooses
$(0.5,0.5)$, which are both individually possible choices if they
were combined with a suitable choice by the respective other player,
then one of the players is worse off than their disagreement point.
Hence, choosing a compromise outcome leads to a coordination problem.
\end{example}

The problem would be even worse without separability of utility
functions. In this case, the coordination problem may be
severe and wrong combinations of action may even be Pareto suboptimal.
Given additive separability, the problem is not as severe. It is not a problem
if, for a player $i\in N$, there are several $\sigma_{i}\in\Sigma_{i}$
with the same utilities for all players. Hence, it suffices to analyze
a game directly on the basis of the individual feasible sets $(F_{i})_{i\in N}$.
Here, as the next result shows, the outcomes will at least always be Pareto optimal.
\begin{prop}
\label{thm:2}Let $B=(N,(F_{i})_{i\in N},d)$ be an additively separable bargaining
game. Let $x\in H$. Let \[X_{i}=\{y_{i}\in H_{i}\mid\exists y_{-i}\in H_{-i}:=\sum_{j\in N\setminus\{i\}}H_{j}:y_{i}+y_{-i}=x\}\]
be the set of points $y_{i}$ that player $i$ can choose from to realize $x$. Then $\sum_{i\in N}X_{i}\subseteq H$. That is, any combination of such points chosen independently by different players is Pareto optimal.
\end{prop}

\begin{proof}
Let $\mu\in\mathbb{R}^{n}$ such that $\mu^{\top}x=\max_{y\in F}\mu^{\top}y$
(since $x$ is Pareto optimal, such a weight vector exists). Apparently,
for any $x_{i}\in X_{i}$, it is $\mu^{\top}x_{i}=\max_{y\in F_{i}}\mu^{\top}y$.
Hence, $\mu^{\top}y=\max_{z\in F}\mu^{\top}z$ for any $y\in\sum_{i\in N}X_{i}$.
But this means that $y\in H$.
\end{proof}
Under which conditions could $\sum_{i\in N}X_{i}$ contain more
than one vector? At least if $F$ is strictly convex, this cannot
happen.
\begin{prop}
Same assumptions as \Cref{thm:2}. Moreover, assume that $F(B)$
is strictly convex. 
Then $\sum_{i\in N}X_{i}=\{x\}$.
\end{prop}

\begin{proof}
Assume that for $i\in N$ there are two points $x\neq x'\in H_{i}$
and points $y,y'\in H_{-i}:=\sum_{j\in N\setminus\{i\}}H_{j}$ such
that $x+y=x'+y'=h\in H$. Let $\lambda=\frac{1}{2}$. It is $h'=\lambda x+(1-\lambda)x'+\lambda y+(1-\lambda)y'=\lambda(x+y)+(1-\lambda)(x'+y')=h\in H$.
Moreover, from \Cref{thm:2}, it follows that $\tilde{h}:=x'+y\in H$
and $\lambda h+(1-\lambda)\tilde{h}=\lambda x+(1-\lambda)x'+y\in H$.
Since $h\neq\tilde{h}\in\partial F(B)$ and $\lambda h+(1-\lambda)\tilde{h}\in\partial F(B)$,
$F(B)$ is not strictly convex, which is a contradiction. Hence, $x=x'$.
\end{proof}

Overall, I believe that the kind of non-uniqueness discussed here is unlikely to be a big problem. First, even if the decomposition is not unique in principle, there may still be unique points that are somehow more parsimonious and can thus serve as a Schelling points. E.g., in \Cref{example-non-uniqueness}, this could be the symmetric point \((0.5,0.5)\). Second, I think it is very unlikely that a situation in which
$\sum_{i\in N}X_{i}$ contains more than one point occurs in practice. I have not formalized this, but intuitively, the reason is that Pareto optimal points are points at which the normal vector to the individual Pareto frontiers for all players are colinear. It is unlikely that two players have Pareto frontiers that have a part that is affine and thus not strictly convex, for which their normal vectors are also exactly colinear. This is because there can only be countably many such affine parts.

\subsubsection{Possible gains from trade}
\label{possible-gains-from-trade-bargaining}
We can assess possible gains from trade by looking at the individual Pareto frontiers. Assume that
the whole surface of $F_{i}$ is a smooth manifold, for each player \(i\) (recall that \(F_i\) is the set of expected utility vectors that player \(i\) can choose from, assuming additive separability, i.e., that the total expected utility for each player is a sum of the individual contributions from each player). For instance, one could justify this with the fact that there exists a continuum of possible actions in the real world.

Then there exists a unique normal vector to this surface at each point on the Pareto frontier \(H_i\subseteq\partial F_i\). As mentioned in \Cref{exa:bargaining-case-1}, in the \(2\)-D case, the slope of the Pareto frontier at a point corresponds to the marginal rate of substitution between the two utility functions. Pareto optimal points are points at which those slopes are equal for both players, and the normal vectors colinear. Gains from trade
are possible whenever the marginal rates of substitution between the
different utility functions on the Pareto frontier are not equal for
all players. In particular,
if a player was
previously optimizing for their own goals, then giving utility
to other players costs them nothing on the margin (see \Cref{fig:4}). This idea was introduced as ``marginal charity'' by \textcite{hanson2012marginal}.
\begin{figure}
\includegraphics[width=1\textwidth]{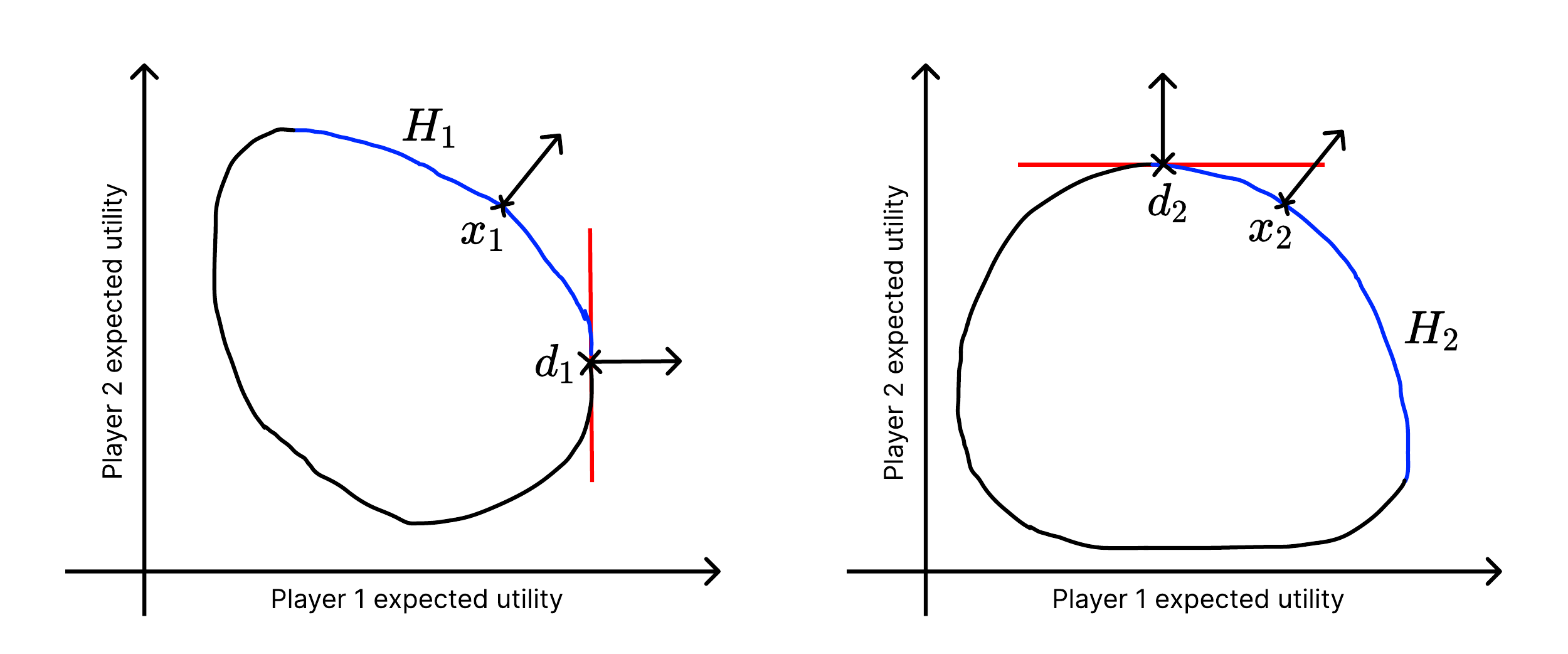}
\caption{Example additively separable bargaining game. Boundary of the feasible set for player 1 (left) and player 2 (right), with Pareto frontiers \(H_1,H_2\) (in blue) and disagreement points \(d_1,d_2\) with normal vectors and slopes (in red). At the disagreement point \(d=d_1+d_2\), players only optimize for their own values. A possible Pareto optimal compromise is \(x=x_1+x_2\).}
\label{fig:4}
\end{figure}

The amount of trade that can happen depends on the specific shape of the Pareto frontiers. If the Pareto frontiers are curved strongly at the disagreement point, such that Pareto optimal trades are very close to this point, then barely any trade is possible
(see \Cref{fig:14}). I will return to this analysis of possible gains from trade using different toy models for the Pareto frontiers in \Cref{observations-final}.
\begin{figure}
\includegraphics[width=1\textwidth]{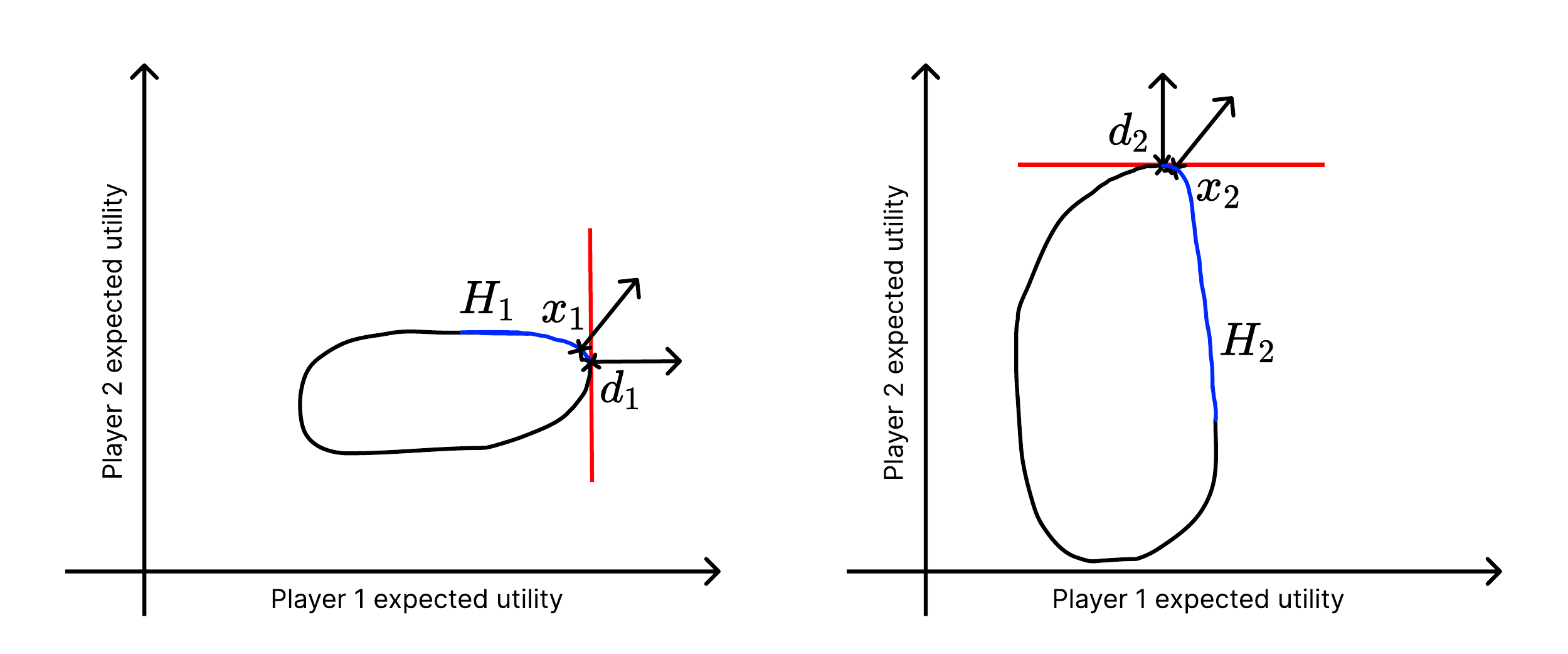}\caption{Bargaining game as in \Cref{fig:4}, with disagreement point \(d=d_1+d_2\) and potential Pareto optimal compromise \(x=x_1+x_2\). Due to the specific shape of the Pareto frontiers, there are almost no gains from trade.}
\label{fig:14}
\end{figure}

\subsubsection{Further observations}
\label{further-observations}
Another observation concerns trades between more than
two players, which can exhibit a more complicated graph structure. For instance, if there are three players $1,2,3$,
it is possible that $1$ benefits $2$, $2$ benefits $3$, and $3$
benefits $1$. Not everyone has to receive gains from trade from everyone
else. This property allows for higher gains from trades, but it also means that there can be players involved that don't benefit anyone else, which can be problematic (see \Cref{fairness-and-coalitional}).

Lastly, it is worth noting that the NBS is continuous in the
feasible set and the disagreement point \parencite{Lensberg1988-lf}.
This means that the NBS is in some sense robust to slight changes or uncertainties about the right specification of the bargaining game.\footnote{I believe this is also true of the KSBS, though I have not investigated this.}

\section{Bayesian game model}
\label{sec:ECL-as-a-Bayesian-Game}
In this section, I introduce a Bayesian game formalism to model uncertainty about the values and empirical situations of agents in the multiverse, using the type space formalism by \textcite{Harsanyi1967},
adapted to ECL. In a Bayesian game, players have \emph{incomplete information}, meaning they are uncertain about the utility function of other players. I will build on the formalism introduced in Sections~\ref{bayesian-game-formalism} and \ref{pure-and-mixed} in \Cref{sec:Bargaining-with-incomplete} to define incomplete information bargaining games. 

In \Cref{bayesian-game-formalism}, I introduce the basic formalism and notation, and in \Cref{pure-and-mixed}, I define pure and mixed strategies and their expected utilities. Next, I introduce joint distributions over strategies in \Cref{joint-strategy-distributions}, which I use in \Cref{subsec:Equilibrium-concepts} to define dependency equilibria, alongside standard Bayesian Nash equilibria. Dependency equilibria assume optimal conditional expectations of actions under joint distributions over strategies and can thus incorporate the evidential reasoning required by ECL.

I then show several equilibrium results, including a generalization of \textcite{Spohn2007-fp}'s folk theorem for dependency equilibria, which says that all strategy profiles that Pareto dominate a Nash equilibrium are dependency equilibria (\Cref{subsec:ObservationsEquil}). This result shows that dependency equilibria alone won't be useful in constraining beliefs over the players' strategies in ECL further. I also derive simple conditions for when a strategy profile leads to positive gains from trade and is thus a dependency equilibrium.

Finally, in \Cref{sec:Uncertainty-about-similarity}, I discuss a possible extension of the formalism to include uncertainty about decision procedures and similarity to other agents.

\subsection{Formal setup}
\label{bayesian-game-formalism}
\begin{defn}
An ECL Bayesian game is a tuple $G=(N,A,T,p,(u_{i})_{i\in N})$, where
\begin{itemize}
\item $N=\{1,\dots,n\}$ is the set of players;
\item $A$ is a generic, finite set of actions;
\item $T=\{1,\dots,m\}$ is a generic set of types, specifying the private information available to each player, i.e., the values and empirical situation in each universe;
\item $p\colon T^{n}\rightarrow[0,1]$ is a prior probability distribution
over the players' types such that all types have positive probability for all players, i.e., \({\sum_{t_{-i}\in T_{-i}}p(t_{-i},t_i)>0}\) for all players \(i\in N\) and types \(t_i\in T\);
\item $(u_{i})_{i\in N}$ is the tuple of utility functions for each player,
where $u_{i}\colon A^{n}\times T^{n}\rightarrow\mathbb{R}$.
\end{itemize}
\end{defn}

Each player's type gets randomly chosen according to the joint distribution \(p\). A type specifies a player's private information, i.e., whatever information about the player
is not common knowledge. In an ECL Bayesian game, I understand each player as a causally separate universe, inhabited by some intelligent civilization that is able to engage in ECL. The player's type then specifies this civilization's values, as well as their options in furthering any of the other types' values. Players know how many universes and thus causally disconnected civilizations there are, but they are uncertain about everyone's types.

I assume that this formal framework is common knowledge. In particular, everyone knows the common prior
over types. I believe this is a good starting point to analyze the situation, but I am unsure to what degree ECL breaks down as we relax the assumption. One possible generalization for future work would be to allow for individual
probability distributions over types that don't stem from a common
prior over types \parencite[see][]{Harsanyi1967}, or analyze relaxations to common knowledge such as common
\(p\)-belief \parencite{Monderer1989-pj}.

My formalization is different from a standard Bayesian game \parencite[e.g.][p.\ 347f.]{maschler2020game} since the set of types \(T\) is the same for each player, and there is only one set of actions, independent of the player and type. Both of these are simply notational simplifications without loss of generality. First, all the information about the actions is encoded in the utility functions, which can depend on players and types (if there are too many actions for some types, we can simply map several actions onto the same utilities). Second, we can still distinguish the different players' type distributions by choosing an appropriate prior distribution \(p\). The only restricting condition here is the assumption that each type for each player has strictly positive probability. However, this assumption could easily be relaxed without changing anything substantial; it merely serves to avoid cumbersome case distinctions based on whether a type has zero probability.

My simplification makes particular sense in ECL, where players are causally disconnected universes. Here, we can regard players simply as vessels that can be inhabited by any of the types, such that really only the types matter. I do not think it would be useful at this point to try to distinguish the different universes.
I formalize the idea that we cannot distinguish between the players as \emph{anonymity} below, alongside the additive separability assumption from the previous section. I will use this assumption a lot in the following.
\begin{defn}[Additively separable and anonymous \(u\)] Assume that there
are utility functions $u_{t,t'}$ for all $t,t'\in T$ such that for
$a=(a_{1},t_{1},\dots,a_{n},t_{n})\in A$, we have
\[
u_{i}(a)=\sum_{j=1}^{n}u_{t_j,t_i}(a_{j})
\]
for each player \(i\in N\). Then the utility functions are called additively separable  and anonymous.
\end{defn}
This definition says that we can express the utility function of any player as a sum of contributions from every player (additive separability), where the received utility only depends on their type, as well as the type of the other player (anonymity). The term \(u_{t,t'}(a)\) thus expresses the utility that any player of type \(t'\) gets from any player of type \(t\), when that player chooses action \(a\in A\).
\begin{defn}[Anonymous $p$]The prior distribution
$p$ is called anonymous if, for all permutations on players
$\pi$ and type vectors $t\in T^{n}$, we have $p(t_{1},\dots,t_{n})=p(t_{\pi(1)},\dots,t_{\pi(n)})$.
\end{defn}
This says that also the distribution over types is anonymous, i.e., symmetric in the players. Note that this does not mean that players' types have to be independent. One could still incorporate a belief under which, for instance, players always believe that other players are likely of the same type as they are.

Lastly, I define the same properties for Bayesian games.
\begin{defn}[Additively separable and anonymous \(G\)] I say that an ECL Bayesian game $G=(N,A,T,p,(u_{i})_{i\in N})$ is additively separable and anonymous if \(u\) is additively separable and anonymous and if \(p\) is anonymous.
\end{defn}

\subsection{Pure and mixed strategies}
\label{pure-and-mixed}
Now I turn to the strategies of players in an ECL Bayesian game, as well as associated expected utilities. I start with pure, i.e., deterministic strategies. I then turn to mixed strategies.

\begin{defn}[Pure strategy profile]
A pure strategy $\alpha_{i}\in A^{m}$ is a mapping from the possible
types of player $i$ to their actions. We denote a pure strategy profile
as $\alpha\in A^{n,m}$. 
\end{defn}

To introduce expected utilities, we need some additional notation for the distribution over types. In a slight abuse of notation, I denote the prior probability of player $i$
having type $t_{i}$ as 
\[
p(t_{i}):=\sum_{t_{-i}\in T_{-i}}p(t_{-i},t_{i}).
\]
Note that, if $p$ is anonymous,
$p(t_{i})$ does not depend on the player.
Player $i$ of type $t_{i}$ has a conditional belief over $t_{-i}\in T_{-i}$, which is given by 
\[
p(t_{-i}\mid t_{i}):=\frac{p(t_{-i},t_{i})}{p(t_{i})}.
\]
Now the expected utility of $\alpha$ for player $i$ of type $t_{i}$
is 
\[
EU_{i}(\alpha; t_{i}):=\sum_{t_{-i}\in T_{-i}}p(t_{-i}\mid t_{i})u(\alpha_{1,t_{1}},t_{1},\dots,\alpha_{n,t_{n}},t_{n}).
\]
This is an \emph{ex interim} expected utility, i.e., after updating on the player's own type, but before having seen anyone else's type. I will focus on ex interim expected utilities in this report since they allow for modeling players with different beliefs, which is an important aspect of ECL in my view.\footnote{For more discussion on the question of whether players should update on their own type in principle, see \textcite{benya2014sin,treutlein2018udt}.}

Given two players $i\neq j\in N$, the joint prior $p$ and types
$t_{i},t_{j}\in T$, we can define 
\[
p(t_{i}\mid t_{j}):=\sum_{t'_{-j}\in T_{-j}\text{ s.t. }t'_{i}=t_{i}}p(t_{-j}'\mid t_{j}).
\]
If $p$ is anonymous, then $p(t_{i}\mid t_{j})$
 depends only on the two types and not the players. We can thus write \(p(t'\mid t)\) for the probability that any player of type \(t\) assigns to any other player having type \(t'\).
Given additive separability and anonymity of $u$, one can use this to simplify
the expected utility of \(\alpha\) as
\[
EU_{i}(\alpha; t_{i})=u_{t_{i},t_{i}}(\alpha_{i,t_{i}})+\sum_{j\in N\setminus\{i\}}\sum_{t_{j}\in T}p(t_{j}\mid t_{i})u_{t_{j},t_{i}}(\alpha_{i,t_{j}}).
\]
If $\alpha_{1t}=\dots=\alpha_{nt}$ for all $t\in T$, I say that
$\alpha$ is anonymous and we can write $\alpha\in T^{m}$. If, in
addition to additive separability/anonymity of $u$, $\alpha$
and $p$ are anonymous, we can write
\begin{equation}
EU_{t}(\alpha)=u_{t,t}(\alpha_{t})+(n-1)\sum_{t'\in T}p(t'\mid t)u_{t',t}(\alpha_{t'})\label{eq:1}
\end{equation}
for the expected utility for any player of type \(t\), if the anonymous strategy \(\alpha\in T^m\) is played. Here, the first term stands for the utility that the player produces for themself, while the term with the factor \((n-1)\) stands for the expected utility provided by all other \(n-1\) players, where the expectation is over possible types for any of the other players (and this term is the same for every player due to anonymity).

Next, I consider mixed strategy profiles. In a Bayesian game, we have to specify a distribution over actions for each tuple \((i,t_i)\in N\times T\) of a player and associated type.
\begin{defn}[Mixed strategy profile]
A mixed strategy \(\sigma_i\in \Sigma_i:=\Delta(A)^T\) for player \(i\) specifies for each possible type \(t_i\), a probability distribution over actions, denoted via \(\sigma_i(\cdot \mid t_i)\). A mixed strategy profile is a vector \(\sigma\in \prod_{i\in N}\Sigma_i\) of mixed strategies for each player.
\end{defn}

As with actions, we can denote a mixed strategy profile specifying only distributions over actions for players \(N\setminus \{i\}\) as \(\sigma_{-i}\in \Sigma_{-i}\).
The expected utility for player \(i\) of action \(a_i\), given mixed strategy profile \(\sigma_{-i}\in \Sigma_{-i}\), is defined as
\[EU_i(\sigma_{-i},a_i;t_i):=\sum_{t_{-i}\in T_{-i}}p(t_{-i}\mid t_i)\sum_{a_{-i}}\left(\prod_{j\in N\setminus \{i\}}\sigma_j(a_j\mid t_j)\right)u_i(a_1,t_1,\dotsc,a_n,t_n).\]
Similarly, we can define
\[EU_i(\sigma;t_i)
:=\sum_{a_i\in A}\sigma_i(a_i)\left[EU_i(\sigma_{-i},a_i;t_i)\right].\]
Similarly to the above for pure strategy profiles, we could simplify this expression for additively separable and anonymous games. I will skip this here since it won't be needed in the following.

\subsection{Joint distributions over strategies}
\label{joint-strategy-distributions}
Mixed strategy profiles specify independent distributions over actions for each player. I will use them below to define Bayesian Nash equilibria \parencite[][p.\ 354]{maschler2020game}. However, in the case of ECL, it is important to consider dependencies between the actions of different players. In this section, I will thus define \emph{joint strategy distributions}, which allow for different players' actions to be dependent. I will use them to introduce a Bayesian game generalization of \emph{dependency equilibria} \parencites{Spohn2007-fp,Spohn2010Depen-13626,Spohn2003-gi,Spohn2005-hi}, which explicitly take such dependencies into account. 

\begin{defn}[Joint strategy distribution] Let $S=\{s\colon T^n\rightarrow \Delta(A^n),t\mapsto s(\cdot \mid t)\}$
be the set of conditional joint probability distributions over the
actions of all players given their types. Then $s\in S$ is called a joint strategy
distribution.
\end{defn}

Unlike the mixed strategy profiles in bargaining problems, I interpret the distributions over strategies here as subjective credences, rather than as options that could be implemented by the players, e.g., via a randomization device. If players were able to randomize, then this would naturally lead to independent distributions (absent a randomization device that is correlated across the multiverse). Instead, ECL is based on beliefs over actions that imply that agents' actions are dependent, due to the similarity of their decision procedures. I use joint strategy distributions to formalize such beliefs.\footnote{The idea that distributions over actions describe beliefs rather than randomization is also common in traditional game theory. E.g.,
\textcite{Aumann1987} writes:
``An important feature of our approach is that it does not require
explicit randomization on the part of the players. Each player always
chooses a definite pure strategy, with no attempt to randomize; the
probabilistic nature of the strategies reflects the uncertainties
of other players about his choice.''}

Joint strategy distributions can also be anonymous, i.e., symmetric in the player number.
\begin{defn}[Strategy anonymity] A joint strategy profile $s\in S$ is called anonymous if for any player permutation
$\pi\colon N\rightarrow N$, action vector \(a\in A^{n}\), and type vector \(t\in T^{n}\), we have
\[s(a\mid t)=s(a_{\pi(1)},\dots,a_{\pi(n)}\mid t_{\pi(1)},\dots,t_{\pi(n)}).\]
\end{defn}

Joint strategy distributions are equivalent to standard mixed strategy profiles in the special case in which the marginals over the different players' actions are independent. To define this formally, we denote the probability for player $i\in N$ of playing $a_{i}$ given
type vector $t\in T^{n}$ by
\[
s(a_{i}\mid t):=\sum_{a_{-i}\in A_{-i}}s(a_{-i},a_{i}\mid t).
\]
Moreover, the prior probability for player $i\in N$ of type $t_{i}$ of playing
$a_{i}$ is 
\[
s(a_{i}\mid t_{i}):=\frac{\sum_{t_{-i}\in T_{-i}}s(a_{i}\mid t_{-i},t_{i})p(t_{-i},t_{i})}{p(t_{i})}.
\]
If $s$ is anonymous, these probabilities don't depend on $i$. This
justifies defining $s(a\mid t)$ for any $a\in A$ and $t\in T$ in
this case. 
\begin{defn}[Uncorrelated joint strategy distribution]
A joint strategy distribution \(s\) is said to be \emph{uncorrelated}, if
\begin{enumerate}
    \item \(s(a_i\mid t_i)=s(a_i\mid t)\) for any player \(i\in N\), type \(t\in T^n\) and action \(a_i\in A\);
    \item
    \(s\) factorizes into a product of its marginals, i.e., if for any \(t\in T^n\) and \(a\in A^n\), we have
\[s(a\mid t)=\prod_{i\in N}s(a_i\mid t_i).\]
\end{enumerate}
\end{defn}
Note that the term ``uncorrelated'' is imprecise, since the definition actually requires independence. However, I am using the term for simplicity.

Now I turn to \emph{conditional expected utilities} of actions. For player \(i\in N\), the conditional probability of other players' actions $a_{-i}\in A_{-i}$
given player \(i\)'s action $a_{i}$, type vector $t\in T^{n}$, and joint strategy profile $s\in S$ is
\[
s(a_{-i}\mid a_{i},t):=\frac{s(a_{-i},a_{i}\mid t)}{s(a_{i}\mid t)}.
\]
If the players' action distributions under \(s\) are dependent, then this probability might differ between different actions \(a_i\). It takes dependencies into account, instead of simply marginalizing over all possible actions for player \(i\) to arrive at an unconditioned probability.

Next, for $i,j\in N$, the probability of $a\in A^{n}$ given $t_{i},t_{j}\in T$
is 
\[
s(a\mid t_{i},t_{j}):=\frac{\sum_{t'\in T^{n}\text{ s.t. }t'_{i}=t_{i},t'_{j}=t_{j}}s(a\mid t')p(t')}{\sum_{t'\in T^{n}\text{ s.t. }{t'}_{i}={t}_{i},{t'}_{j}={t}_{j}}p(t')}.
\]
In another slight abuse of notation, I regard $\alpha$
as an $A^{n,m}$-valued random variable and denote the probability
that player $i$ of type $t_{i}$ plays $a_{i}$ given that player
$j$ of type $t_{j}$ plays $a_{j}$ via
\[
s(\alpha_{i,t_{i}}=a_{i}\mid\alpha_{j,t_{j}}=a_{j},t_{i},t_{j}):=\frac{\sum_{a'\in A^{n}\text{ s.t. }a'_{i}=a_{i},a'_{j}=a_{j}}s(a'\mid t_{i},t_{j})}{\sum_{a'\in A^{n}\text{ s.t. }a'_{j}=a_{j}}s(a'\mid t_{i},t_{j})}.
\]
Given anonymous \(s\) and \(p\), if $i\neq j\in N$, this does not depend on
the players. Lastly, I define
\[
s(a_{i},t_{i}\mid a_{j},t_{j}):=s(\alpha_{i,t_{i}}=a_{i}\mid\alpha_{j,t_{j}}=a_{j},t_{i},t_{j})p(t_{i}\mid t_{j}).
\]
Apparently, given anonymity, $s(a_{i},t_{i}\mid a_{j},t_{j})$ only
depends on the types and actions, but not on either $i$ or $j$ (as
long as $i\neq j$). 

With these notations at hand, we can proceed and define conditional expected
utilities.

\begin{defn}[Conditional expected utility] The \emph{conditional expected
utility} of strategy \(s\in S\), given action $a_{i}\in A$ and type \(t_i\) for player $i$ is defined as
\[
EU_{i}(s; a_i,t_{i}):=\sum_{t_{-i}\in T_{-i}}p(t_{-i}\mid t_i)\sum_{a_{-i}\in A_{-i}}s(a_{-i}\mid a_{i},t_{-i},t_{i})u_{i}(a_{-i},a_{i},t_{-i},t_{i}).
\]
Moreover, assuming anonymity of \(s\) and \(p\) and additive separability and anonymity of \(u\), we define
\begin{equation}
EU_{t}(s; a):=u_{t,t}(a)+(n-1)\sum_{t'\in T}\sum_{a'\in A}s(a',t'\mid a,t)u_{t',t}(a')\label{eq:2-1}
\end{equation}
as the conditional expected utility of \(s\in S\) for any player of type $t$ given action
$a\in A$.
\end{defn}
Note that here, we condition the distribution over the other players' actions on player \(i\)'s action. The conditional expected utility of different actions hence differs not only due to the different causal effects of the actions, but also due to potential dependencies between different players' actions under the distribution \(s\). For instance, in a prisoner's dilemma, one could define a distribution \(s\) under which either all players cooperate or all players defect. Then the conditional expected utility of cooperating would be higher, since it would take into account the correlations between the players' actions.

The following lemma justifies above definition of \(EU_t(a\mid s)\) in the case of anonymity and additive separability.
\begin{lem}\label{lemma-anonymous}
Assume \(s\in S\) and \(p\) are anonymous and \(u\) is additively separable and anonymous. Then we have 
\[EU_i(s;a_i,t_i)=EU_{t_i}( s;a_i)\]
for any player \(i\in N\), action \(a_i\in A\), joint strategy profile \(s\in S\), and type \(t_i\in T\).
\end{lem}
\begin{proof}We have
\begin{align}
EU_{i}(s;a_i,t_{i})
&=\sum_{a_{-i}\in A_{-i}}\sum_{t_{-i}\in T_{-i}}s(a_{-i}\mid a_{i},t_{-i},t_{i})p(t_{-i}\mid t_{i})u_{i}(a_{-i},a_{i},t_{-i},t_{i})\\
&=u_{t_{i},t_{i}}(a_{i})+\sum_{a'\in A_{-i}}\sum_{t'\in T_{-i}}s(a'\mid a_{i},t',t_{i})p(t'\mid t_{i})\sum_{k\in N\setminus\{i\}}u_{t'_{k},t_{i}}(a'_{k})\\
&=
u_{t_{i},t_{i}}(a_{i})+\sum_{k\in N\setminus\{i\}}\sum_{t''\in T}\sum_{a''\in A}\sum_{a'\in A_{-i}\text{ s.t. }a'_{k}=a''}\sum_{t'\in T_{-i}\text{ s.t. }t'_{k}=t''}s(a'\mid a_{i},t',t_{i})p(t'\mid t_{i})u_{t'_{k},t_{i}}(a'_{k})\\
&=u_{t_{i},t_{i}}(a_{i})+(n-1)\sum_{t'\in T}\sum_{a'\in A}s(a',t'\mid a_i,t_i)u_{t',t_i}(a')
\\
&=EU_{t_i}( s;a_i).
\end{align}
\end{proof}

Before turning to equilibrium concepts, we briefly consider the case in which strategies are uncorrelated. In this case, conditional expected utilities correspond to standard expected utilities given a mixed strategy profile and an action. 
\begin{prop}
\label{equivalence-uncorrelated}
    Assume \(s\in S\) is uncorrelated and define \(\sigma\) via \(\sigma_i(a_i\mid t_i):=s(a_i\mid t_i)\) for any player \(i\in N\), action \(a_i\in A\), and type \(t_i\in T\). Then we have
    \[EU_i(s; a_i,t_i)=EU_i(\sigma_{-i},a_i;t_i)\]
    for all players \(i\in N\), actions \(a_i\in A\), and types \(t_i\in T\).
\end{prop}

\begin{proof}
    We have\begin{align}EU_i(s;a_i,t_i)&=\sum_{t_{-i}\in T_{-i}}p(t_{-i}\mid t_i)\sum_{a_{-i}\in A_{-i}}s(a_{-i}\mid a_{i},t_{-i},t_{i})u_{i}(a_{-i},a_{i},t_{-i},t_{i})
    \\
    &=
    \sum_{t_{-i}\in T_{-i}}p(t_{-i}\mid t_i)\sum_{a_{-i}\in A_{-i}}\frac{s(a_{-i},a_{i}\mid t_{-i},t_i)}{s(a_{i}\mid t_{-i},t_i)}u_{i}(a_{-i},a_{i},t_{-i},t_{i})
        \\
    &=
    \sum_{t_{-i}\in T_{-i}}p(t_{-i}\mid t_i)\sum_{a_{-i}\in A_{-i}}\frac{\prod_{j\in N}s(a_j\mid t_j)}{s(a_{i}\mid t_i)}u_{i}(a_{-i},a_{i},t_{-i},t_{i})
            \\
    &=
    \sum_{t_{-i}\in T_{-i}}p(t_{-i}\mid t_i)\sum_{a_{-i}\in A_{-i}}\prod_{j\in N\setminus\{i\}}s(a_j\mid t_j)u_{i}(a_{-i},a_{i},t_{-i},t_{i})
                \\
        &=
    \sum_{t_{-i}\in T_{-i}}p(t_{-i}\mid t_i)\sum_{a_{-i}\in A_{-i}}\prod_{j\in N\setminus\{i\}}\sigma_j(a_j\mid t_j)u_{i}(a_{-i},a_{i},t_{-i},t_{i})
    \\
    &=EU_i(\sigma_{-i},a_i;t_i).
    \end{align}
\end{proof}

\subsection{Equilibrium concepts}
\label{subsec:Equilibrium-concepts}
To analyze the equilibria of ECL Bayesian games, I first define a \emph{Bayesian Nash equilibrium}, which is a standard solution concept for Bayesian games and which assume mixed strategy profiles, or equivalently, uncorrelated joint strategy distributions. Afterwards, I will introduce \emph{dependency equilibria} \parencite{Spohn2007-fp,Spohn2010Depen-13626,Spohn2003-gi,Spohn2005-hi}, which are based on conditional expected utility of potentially dependent strategy distributions and thus more suitable for ECL. Both equilibrium concepts can be motivated descriptively, to analyze how agents in the multiverse might behave, as well as normatively, to ask how rational agents should behave. In addition to the assumption of common knowledge in rationality, both equilibrium concepts are based on the assumption that all players share the same belief over the actions of the other players, conditional on their types. This assumption is too restrictive, absent a mechanism that could force such a common belief, such as repeated interactions or mutual simulation. However, as with other modeling assumptions, we will use this as a starting point for our analysis. In the case of dependency equilibria, our assumptions don't constrain the space of equilibria much: there exists a result similar to the folk theorems for iterated games \parencite[see][]{fudenberg1986folk}, saying that any Pareto improvement over a Bayesian Nash equilibrium is a dependency equilibrium (see \Cref{subsec:ObservationsEquil}).

\begin{defn}[Bayesian Nash equilibrium]
A mixed strategy profile \(\sigma\) is a Bayesian Nash equilibrium
if for all players $i\in N$, types $t_{i}\in T$, actions \(a_i\in A\) such that $\sigma_i(a_{i}\mid t_{i})>0$, we have
\[
EU_{i}(\sigma_{-i},a_i;t_{i})\geq EU_{i}(\sigma_{-i},a'_{i};t_{i})\quad \forall a'_i\in A.
\]
An \emph{uncorrelated} joint strategy distribution \(\sigma\) is a Bayesian Nash equilibrium if
\[
EU_{i}(s;a_i,t_{i})\geq EU_{i}(s;a'_{i};t_{i})\quad \forall a'_i\in A.
\]
for all actions $a_{i}\in A$ such that $s(a_{i}\mid t_{i})>0$.
\end{defn}

\begin{remark}\label{remark-bne}
    Note that for a Bayesian Nash equilibrium \(\sigma\), we have 
    \[EU_i(\sigma;t_i)
    =\sum_{a_i\in A}\sigma_i(a_i)EU_i(\sigma_{-i},a_i;t_i)
    \geq \sum_{a_i\in A}\sigma_i(a_i)EU_i(\sigma_{-i},a'_i;t_i)
    =EU_i(\sigma_{-i},a'_i;t_i)\]
    for any player \(i\in N\), action \(a'_i\in A\), and type \(t_i\in T\). Similarly, one can show that if \(EU_i(\sigma;t_i)\geq EU_i(\sigma_{-i},a'_i;t_i)\) holds for all actions \(a'_i\), then \(\sigma\) is a Bayesian Nash equilibrium.
\end{remark}

A Bayesian Nash equilibrium is a generalization of a Nash equilibrium for Bayesian games, where the expected utility of a strategy is replaced with the ex interim expected utility. The condition for Nash equilibria is simply \(EU_i(\sigma_{-i},a_i)\geq EU_i(\sigma_{-i},a'_i)\) for all players \(i\in N\) and actions \(a_i,a_i'\in A\) where \(\sigma_i(a_i)>0\).

In a Bayesian Nash equilibrium, we assume that players respond optimally to the distributions over other players' actions. We assume that these distributions are independent, and taking an action does not provide any evidence about the actions of other players. Hence, the notion of best response here takes into account only causal effects of an action, by influencing \(u\) directly, rather than by influencing the distribution over actions. As a result, Bayesian Nash equilibria cannot capture the type of reasoning that is required for ECL.

There is another standard solution concept that does assume potential correlations between players' actions, the \emph{correlated equilibrium} \parencite[][ch.~8]{maschler2020game}. However, this equilibrium concept also fails to capture ECL-type reasoning. Even though players' actions can be correlated, the notion of best response still requires that a player cannot improve their payoff by unilaterally deviating from the joint distribution, without taking into account the evidence such deviations would provide about other players' actions. Hence, I will not delve further into correlated equilibria here.

Instead, I will turn to dependency equilibria, which incorporate evidential reasoning by considering potentially correlated joint distributions and evaluating only \emph{conditional} expected utilities of actions. There are several other concepts achieving a similar purpose that one could look at in future work \parencite{al2015evidential,daley2017magical,halpern2018game}, but I will focus on dependency equilibria in the following. The following definition of a dependency equilibrium is a Bayesian game generalization of the definition in \textcite{Spohn2007-fp}.\footnote{For more discussions on dependencies between agents in games, see \textcite{Spohn2007-fp,Spohn2010Depen-13626}. Spohn sees prior
causal interactions as a common cause between agents' actions, leading to a dependency \parencite[cf.][]{sep-physics-Rpcc}. ECL involves dependencies despite no prior causal interaction. Instead, the dependency is caused by the similarity of decision
algorithms and decision situations of agents in ECL. It could be considered a logical dependency, for which there does not need to exist a common cause. Alternatively, the decision situation and decision algorithm similarity could be considered as an abstract common cause \parencite[cf.][]{Yudkowsky2010-ur}.}

\begin{defn}[Dependency equilibrium]
\label{defn:dependency-equilibrium}
A joint strategy distribution $s\in S$ is a dependency equilibrium if
there exists a sequence of distributions $(s_{r})_{r\in\mathbb{N}}$
such that $\lim_{r\rightarrow\infty}s_{r}=s$, and ${s_{r}(a_{i}\mid t_{i})>0}$
for all players $i\in N,$ actions $a_{i}\in A$, types $t_{i}\in T$
and $r\in\mathbb{N}$, and if for all $i\in N,$ $t_{i}\in T$ and
$a_{i}\in A$ with $s(a_{i}\mid t_{i})>0$, it is 
\[
\lim_{r\rightarrow\infty}EU_{i}(s_{r};a_{i},t_{i})\geq\lim_{r\rightarrow\infty}EU_{i}(s_{r};a'_{i},t_{i})\quad\forall a'_{i}\in A.
\]
\end{defn}

The requirement of rationality here is that any action with nonzero probability (in the limit) has to have
greater or equal conditional expected utility for the player performing that
action than any other action. This is similar to a Bayesian Nash equilibrium, only that players' actions are potentially dependent, and we take such dependencies into account when calculating conditional expected utilities. 
The construction with limits is required since conditional credences \(s(a_{-i}\mid a_i,t)\) can only be computed for actions \(a_i\) that have positive probability. Hence, to be able to compute all possible conditional credences, we represent a dependency equilibrium \(s\) as a limit of distributions \(s_r\) for which this is the case.

\begin{example}
\label{exa:Bayesian_Prisoners_Dilemma}As an example, consider a Bayesian
version of a prisoners' dilemma with additively separable and anonymous
utilities. There are two players, $1,2$, and two types $1,2$. Assume that there is a simple ignorance prior $p$ which gives each
combination of types equal probability. In particular, $p$ is anonymous.
Table \ref{tbl:7} shows the utilities that players of the two types produce
with either of two actions $1,2$. 
\begin{table}
\begin{centering}
\subfloat[Utilities \(u_{1,t}(a_1)\) produced by type $1$'s actions for players of type \(t\).]{\begin{centering}
\begin{tabular}{c|c|c|}
 Action& $u_{1,1}$ & $u_{1,2}$\tabularnewline
\hline 
$1$ & $2$ & 2\tabularnewline
\hline 
$2$ & $3$ & 0\tabularnewline
\hline 
\end{tabular}
\par\end{centering}
}\quad
\subfloat[Utilities \(u_{2,t}(a_2)\) produced for type 2's actions for players of type \(t\).]{\begin{centering}
\begin{tabular}{c|c|c|}
 Action& $u_{2,1}$ & $u_{2,2}$\tabularnewline
\hline 
$1$ & $2$ & $2$\tabularnewline
\hline 
$2$ & $0$ & $3$\tabularnewline
\hline 
\end{tabular}
\par\end{centering}
}
\par\end{centering}
\caption{Additively separable payoffs in a Bayesian version of the prisoners' dilemma.}
\label{tbl:7}
\end{table}
For any of the two players, given an anonymous strategy profile $s$ and action \(a\), using \Cref{eq:2-1}, we get
\[
EU_{t}(s;a) =u_{t,t}(a)+ \sum_{t'\in \{1,2\}}\sum_{a'\in \{1,2\}}u_{t',t}(a)\cdot s(a',t'\mid a,t).
\]
Given an uncorrelated strategy profile we have \(s(a',t'\mid a,t)=s(a',t'\mid \hat{a},t)\) for any two actions \(a,\hat{a}\). Hence, the only term differing between different actions is the term \(u_{t,t}(a)\). It follows that the only possible optimal choice for either type is \(a=2\), leading to an expected utility of \(EU_t(s;a)=3 + \frac{1}{2}\cdot 3=4.5\), consisting of \(3\) utility produced by a player for themself and \(3\) utility provided by the other player in the \(50\%\) of cases in which the other player has the same type, and \(0\) utility provided otherwise. This is the only Bayesian
Nash equilibrium.

What about
dependency equilibria? For simplicity, I restrict myself to joint
strategies that have players of the same type always performing
the same action. This leaves us with $4$ probabilities to be determined
(Table \ref{tbl:10}) and the following payoffs for the two types:

\begin{align}
EU_{1}( s;1)= & 2+\frac{1}{2}\cdot 2+\frac{1}{2}\cdot 2\cdot \frac{a}{a+b} = 3 +  \frac{a}{a+b}\\
EU_{1}(s;2)= & 3+\frac{1}{2}\cdot 3 + \frac{1}{2}\cdot 2\cdot \frac{c}{c+d} = 4.5 + \frac{c}{c+d} \\
EU_{1}( s;1)= & 2+\frac{1}{2}\cdot 2+\frac{1}{2}\cdot 2\cdot \frac{a}{a+c}=3+ \frac{a}{a+c}\\
EU_{1}(s;2)= & 3+\frac{1}{2}\cdot 3 + \frac{1}{2}\cdot 2\cdot \frac{b}{b+d}=4.5 +\frac{b}{b+d} .\\
\end{align}
Here, the first term is the utility produced by a player for themself, the second term the utility produced by the other player given that they have the same type (which happens with probability \(\frac{1}{2}\)), and the third term is the utility produced by the other player, assuming they have the opposite type. The term \(\frac{a}{a+b}\), for instance, stands for the probability that a player of type \(2\) plays actions \(1\), assuming that the other player of type \(1\) also plays that action.

\begin{table}
\begin{centering}
\begin{tabular}{|c|c|c|}
\hline 
Type \(1\)\textbackslash Type \(2\)& $1$ & $2$\tabularnewline
\hline 
$1$ & $a$ & $b$\tabularnewline
\hline 
$2$ & $c$ & $d$\tabularnewline
\hline 
\end{tabular}
\par\end{centering}
\caption{Probability matrix for the different joint actions of the two players, given that one player has type \(1\) and the other player has type \(2\).}
\label{tbl:10}
\end{table}
In this case, there can be no dependency equilibrium in which action \(1\)
gets any probability, since $1$ is worse than $2$, regardless of the chosen probabilities. In the best case, we have \(a=d=\frac{1}{2}\), in which case \(EU_t(s;1)=4\) and \(EU_t(s;2)=4.5\).
This changes if there are many players. Suppose there
are $10$ players, and the other properties of the game stay
the same. Then we have the following payoffs:
\begin{align}
EU_{1}(s;1)= & 12+10\frac{a}{a+b}\\
EU_{1}(s;2)= & 18+10\frac{c}{c+d}\\
EU_{2}(s;1)= & 12+10\text{\ensuremath{\frac{a}{a+c}}}\\
EU_{2}(s;2)= & 18+10\text{\ensuremath{\frac{b}{b+d}}}.
\end{align}

Here, everyone playing action $1$
can be a dependency equilibrium. Given any distribution that puts
only weight on $a$ and $d$, action \(1\) is always better for either type. Hence, we can define \(s_r\) via \(a=\frac{r-1}{r}\) and \(d=\frac{1}{r}\). Then, for any \(r\in\mathbb{N}\), \(a\) is the optimal action under distribution \(s_r\), and \(s:=\lim_{r\rightarrow\infty}\) is the distribution in which all players play action \(1\). To find all the mixed joint strategy
dependency equilibria, we would have to solve for \(s\) such that $EU_{t}( s;1)=EU_{t}(s;2)$. I leave this as an exercise.

\end{example}

For further examples and to become more familiar with the concept,
see \textcite{Spohn2007-fp}. 

As in above example, both equilibrium concepts can again
be adopted to an anonymous and additively separable setting. For instance, for Bayesian Nash equilibria, given an anonymous and additively separable game \(G\) and anonymous and uncorrelated \(s\in S\), we get the requirement 
\[
EU_{t}(s;a)\geq EU_{t}(s;a')\quad\forall a'\in A
\]
for all types $t\in T$ and actions $a\in A$ such that $s(a\mid t)>0$. 

\subsection{Observations}
\label{subsec:ObservationsEquil}
In this section, I show basic results about equilibria in ECL Bayesian games. First, I show that in an additively separable and anonymous game, there is essentially only one unique Bayesian Nash equilibrium---the strategy profile in which each type simply optimizes for their own values in their own universe, disregarding what everyone else is doing. 

\begin{prop}
\label{BNE-unique-anonymous}
Let \(s\in S\) be a Bayesian Nash equilibrium of an additively separable and anonymous ECL Bayesian game \(G\). Then for any player \(i\in N\), action \(a_i\in A\) and type \(t_i\in T\), we have \(s(a_i\mid t_i)>0\) if and only if \(u_{t_i,t_i}(a_i)=\max_{a'\in A}u_{t_i,t_i}(a')\). In particular, if the maximizer of \(u_{t,t}\) is unique for any type \(t\), then \(s\) is anonymous and it corresponds to a unique anonymous pure strategy profile \(\alpha\in A^m\). Moreover, an anonymous pure strategy Bayesian Nash equilibrium always exists in an anonymous and additively separable game.
\end{prop}
\begin{proof}
First, since \(s\) is a Bayesian Nash equilibrium, it is uncorrelated, so by \Cref{equivalence-uncorrelated}, there exists \(\sigma\in \Sigma\) such that \(EU_i(s;a_i,t_i)=EU_i(\sigma_{-i},a_i;t_i)\) for all players \(i\), actions \(a_i\), and types \(t_i\). Now let \(a_i\in A,t_i\in T\) arbitrary. Then we have
\begin{align}
EU_{i}(s; a_i,t_{i})&=EU_{i}(\sigma_{-i},a_i;t_i)\\
&=\sum_{t_{-i}\in T_{-i}}p(t_{-i}\mid t_i)\sum_{a_{-i}\in A_{-i}}\prod_{j\in N\setminus\{i\}}\sigma_j(a_j\mid t_j)u_{i}(a_{-i},a_{i},t_{-i},t_{i})
\\
&=\sum_{t_{-i}\in T_{-i}}p(t_{-i}\mid t_i)\sum_{a_{-i}\in A_{-i}}\prod_{j\in N\setminus\{i\}}\sigma_j(a_j\mid t_j)\left(u_{t_i,t_i}(a_i)+\sum_{j\in N\setminus\{i\}}u_{t_j,t_i}(a_j)\right)
\\
&=\sum_{t_{-i}\in T_{-i}}p(t_{-i}\mid t_i)\sum_{a_{-i}\in A_{-i}}\prod_{j\in N\setminus\{i\}}\sigma_j(a_j\mid t_j)u_{t_i,t_i}(a_i)\\
&\phantom{=}+\sum_{a_{-i}\in A_{-i}}\sum_{t_{-i}\in T_{-i}}p(t_{-i}\mid t_i)\prod_{j\in N\setminus\{i\}}\sigma_j(a_j\mid t_j)\sum_{j\in N\setminus\{i\}}u_{t_j,t_i}(a_j)
\\
&=u_{t_i,t_i}(a_i)+\sum_{t_{-i}\in T_{-i}}p(t_{-i}\mid t_i)\sum_{a_{-i}\in A_{-i}}\prod_{j\in N\setminus\{i\}}\sigma_j(a_j\mid t_j)\sum_{j\in N\setminus\{i\}}u_{t_j,t_i}(a_j).
\end{align}
Note that the second term does not depend on \(a_i\). Hence, by the definition of a Bayesian Nash equilibrium, for any action \(a_i\in A\) such that \(s(a_i\mid t_i)>0\), and any alternative action \(a'_i\in A\), we have
\[
0\leq EU_i(s;a_i,t_i)-EU_i(s;a'_i,t_i)
=u_{t_i,t_i}(a_i)-u_{t_i,t_i}(a'_i).
\]
This shows that 
\(u_{t_i,t_i}(a_i)\geq u_{t_i,t_i}(a'_i)\) for all \(a'_i\in A\), so
\(u_{t_i,t_i}(a_i)=\max_{a'\in A}u_{t_i,t_i}(a')\).

For the ``in particular'' part, note that if the maximizer is unique, it follows for any \(i\in N\) and \(t_i\in T\) that \(s(a_i\mid t_i)=1\) for \(a_i=\argmax_{a'\in A}u_{t_i,t_i}(a')\). Hence, \(s\) corresponds to the unique anonymous pure strategy profile \(\alpha\in A^m\), defined via \(\alpha_{t}:=\argmax_{a'\in A}u_{t,t}(a')\) for any \(t\in T\).
Since this does not depend on the player, we have
\[s(a\mid t)
=\mathbbm{1}_{a=(\alpha_{t_i})_{i\in N}}
=\mathbbm{1}_{(a_{\pi(i)})_{i\in N}=(\alpha_{t_{\pi(i)}})_{i\in N}}
=s(a_{\pi(1)},\dotsc,a_{\pi(n)}\mid t_{\pi(1)},\dotsc,t_{\pi(n)})\]
for any permutation \(\pi\colon N\rightarrow N\).

Lastly, it is clear from the above that we can always define an anonymous pure strategy Bayesian Nash equilibrium \(\alpha\in A^m\) given an additively separable and anonymous game by choosing \(\alpha_t\in \argmax_{a'\in A}u_{t,t}(a')\) arbitrarily for each \(t\).
\end{proof}

Next, I turn to dependency equilibria. First, as mentioned by \textcite{Spohn2007-fp}, every Bayesian Nash equilibrium is also a dependency equilibrium.
\begin{prop}
Every Bayesian Nash equilibrium is a dependency equilibrium.
\end{prop}
\begin{proof}
Let \(s\in S\) be a Bayesian Nash equilibrium. Define \(s_r=s\) for any \(r\in\mathbb{N}\). Then, by the definition, we have
\[\lim_{r\rightarrow \infty}EU_i(s_r;a_i,t_i)
=EU_i(s;a_i,t_i)
\geq EU_i(s;a'_i,t_i)
=\lim_{r\rightarrow\infty}EU_i(s_r;a'_i,t_i)\]
for any player \(i\in N\), type \(t_i\in T\), and actions \(a_i,a'_i\in A\) such that \(s(a_i\mid t_i)>0\). This shows that \(s\) is also a dependency equilibrium.
\end{proof}

Second, any pure strategy profile that is at least as good as some mixed strategy profile for every player and given any action is a dependency equilibrium. It follows as a corollary that a profile is a dependency equilibrium if it is a (weak) Pareto improvement over a Bayesian Nash equilibrium. The latter was proven in \textcite[][Observation~5]{Spohn2007-fp} for two-player normal form games.
\begin{thm}
\label{thm:spohn5}Let $\sigma\in \Sigma$ be any mixed strategy profile. Let $\alpha\in A^{n,m}$ be a pure strategy profile such
that $EU_{i}(\alpha\mid t_i)\geq EU_i(\sigma_{-i},a_i;t_{i})$ for all players
$i\in N$, $t_{i}\in T$, and actions $a_{i}\in A$.
Define
$q\in S$ such that for all $t\in T^{n}$, $q(\alpha_{1,t_{1}},\dots,\alpha_{n,t_{n}}\mid t)=1$
and $q(a\mid t)=0$ for \(a\in A^n\) such that \(a\neq(\alpha_{1,t_{1}},\dots,\alpha_{n,t_{n}})\).
Then $q$ is a dependency equilibrium.
\end{thm}
\begin{proof}We construct distributions \(q_r\) that converge to \(q\) and such that, conditional on any player taking an action other than the one specified by \(\alpha\), the remaining players play \(\sigma\). It then follows from the assumption that the actions in \(\alpha\) have highest conditional expected utility.

To begin, note that for any player \(i\in N\) and type \(t_i\in T\), we have
\begin{multline}\label{eq:24}
EU_{i}(q;\alpha_{i,t_{i}},t_{i})=\sum_{t_{-i}\in T_{-i}}p(t_{-i}|t_i)\sum_{a_{-i}\in A_{-i}}s(a_{-i}\mid a_{i},t_{-i},t_i)u_{i}(a_{-i},a_{i},t_{-i},t_{i})\\
=\sum_{t_{-i}\in T_{-i}}p(t_{-i}\mid t_{i})u_i(\alpha_{1,t_{1}},t_{1},\dots,\alpha_{n,t_{n}},t_{n})=EU_{i}(\alpha;t_{i}),
\end{multline}
so conditional on taking actions specified by \(\alpha\), \(q\) is equivalent to \(\alpha\).

Now let \(s\) be the (uncorrelated) joint strategy distribution corresponding to \(\sigma\), i.e., such that \(EU_i(s;a_i,t_i)=EU_i(\sigma_{-i},a_i;t_i)\) for all \(i\in N\), \(t_i\in T\), and \(a_i\in A\) such that \(\sigma_i(a_i\mid t_i)=s(a_i\mid t_i)>0\). We distinguish two cases.

First, we assume $s$ is strictly positive, i.e., \(s(a_i\mid t_i)>0\) for any \(i\in N\), \(a_i\in A\) and \(t_i\in T\).
Define $q_{r}:=\frac{r-1}{r}q+\frac{1}{r}s$
(for $r>0$). Then for any $i\in N$, $t\in T^n$, \(a_{-i}\in A_{-i}\), and
$a_{i}\in A$ such that $a_{i}\neq\alpha_{i,t_{i}}$, 
we have
\begin{align}q_r(a_{-i}\mid a_i, t)
&=\frac{\frac{r-1}{r}q(a_{-i},a_i\mid t)+\frac{1}{r}s(a_{-i},a_i\mid t)}{\sum_{a'_{-i}\in A_{-i}}\frac{r-1}{r}q(a'_{-i},a_i\mid t)+\frac{1}{r}s(a'_{-i},a_i\mid t)}\\
&=\frac{\frac{1}{r}s(a_{-i},a_i\mid t)}{\sum_{a'_{-i}\in A_{-i}}\frac{1}{r}s(a'_{-i},a_i\mid t)}
\\
&=s(a_{-i}\mid a_i,t).\end{align}
That is, conditional on taking action \(a_i\), \(a_{-i}\) is distributed according to \(s\). Hence, using the assumption on \(\alpha\) and \(\sigma\), it follows
\[EU_i( q_{r}; a_{i},t_{i})=EU_i(s; a_{i},t_{i})=EU_i(\sigma_{-i},a_i;t_i)\leq EU_{i}(\alpha; t_{i})\]
for any $r\in\mathbb{N}_{>0}$. Since the expected utility is continuous in the joint strategy distribution, it follows that
\[\lim_{r\rightarrow\infty}EU_i(q_r;a_i,t_i)\leq EU_i(\alpha;t_i)\underset{\text{(\ref{eq:24})}}{=}
EU_i(q;\alpha_{i,t_i},t_i)=EU_i(\lim_{r\rightarrow\infty}q_r;\alpha_{i,t_i},t_i)
=\lim_{r\rightarrow\infty}EU_i(q_r;\alpha_{i,t_i},t_i).\]
Hence, since \(\alpha_{i,t_i}\) is the only action that player \(i\) of type \(t_i\) plays under \(q\), it follows that $q$ is a dependency equilibrium.

Now consider the case where $s$ is not strictly positive. Then we need to modify $q_{r}$ to put weight on all actions, to satisfy the definition of a dependency equilibrium. \textcite{Spohn2007-fp} does not specify exactly how this
would be done in their setup,
so I am providing a more detailed proof here.

To this end, I define
a joint distribution $s'$. Let $s'(a\mid t)=0$ for $a\in A^{n},t\in T^{n}$,
unless specified otherwise. Take any $i\in N,a_{i}\in A,t\in T^{m}$
such that $s(a_{i}\mid t)=0$. We want to define \(s'\) in a way such that \(s'(a_i\mid t)>0\). To do so, for some yet to be determined constant \(c>0\), we let $s'(a_{-i},a_{i}\mid t):=c\cdot \sigma_{-i}(a_{-i}\mid t)$
for all $a_{-i}\in A_{-i}$. Now $s'(a_{i}\mid t)>0$ and
\[s'(a_{-i}\mid a_{i},t)=\frac{c\sigma_{-i}(a_{-i}\mid t)}{\sum_{a'_{-i}\in A_{-i}}c\sigma_{-i}(a'_{-i}\mid t)}=\frac{c\sigma_{-i}(a_{-i}\mid t)}{c}=\sigma_{-i}(a_{-i}\mid t)\]
for all $a_{-i}\in A_{-i}$. Moreover, assume for some $j\neq i$, that there
is $s(a_{j}\mid t)=0$ for some $a_{j}\in A$. Then $s'(a_{j}\mid t)=0$
still after these definitions, and by defining $s'(a_{-j},a_{j}\mid t)$ for player \(j\), I don't change
the previously defined $s'(a_{-i},a_{i}\mid t)$ for player \(i\). Apply the same procedure
to all actions and players.

Now choose $c$ such that $\sum_{a\in A^{n}}s'(a\mid t)=1$.
Proceed in the same manner with all $t\in T^{n}$. Define $q_{r}(a\mid t):=\frac{r-1}{r}q(a\mid t)+\frac{r-1}{r^{2}}s(a\mid t)+\frac{1}{r^{2}}s'(a\mid t)$
for $a\in A^{n},t\in T^{n}$. Then for any player $i\in N$, action
$a_{i}$, $t\in T^{n}$ such that $q_{r}(a_{i}\mid t)$ was previously
zero, we now have $q_{r}(a_{-i}\mid a_{i},t)=\sigma_{-i}(a_{-i}\mid t)$ and hence
$\lim_{r\rightarrow\infty}EU_i(q_{r};a_i,t_i)=EU_i(\sigma_{-i},a_{i};t_i)$.
Moreover, for any of the actions $a_{i}$ that receive positive probability
by $s$ under some $t\in T^{n}$ (but not by $q$), we have 
\begin{multline*}
\lim_{r\rightarrow\infty}q_{r}(a_{-i}\mid a_{i},t)=\lim_{r\rightarrow\infty}\left(\frac{\frac{r-1}{r^{2}}s(a_{-i},a_{i}\mid t)}{\frac{r-1}{r^{2}}s(a_{i}\mid t)+\frac{1}{r^{2}}s'(a_{i}\mid t)}+\frac{\frac{1}{r^{2}}s'(a_{-i},a_{i}\mid t)}{\frac{r-1}{r^{2}}s(a_{i}\mid t)+\frac{1}{r^{2}}s'(a_{i}\mid t)}\right)\\
=\lim_{r\rightarrow\infty}\left(\frac{s(a_{-i},a_{i}\mid t)-\frac{1}{r}s(a_{-i},a_{i}\mid t)}{s(a_{i}\mid t)-\frac{1}{r}s(a_{i}\mid t)+\frac{1}{r}s'(a_{i}\mid t)}+\frac{\frac{1}{r}s'(a_{-i},a_{i}\mid t)}{s(a_{i}\mid t)-\frac{1}{r}s(a_{i}\mid t)+\frac{1}{r}s'(a_{i}\mid t)}\right)\\
=\frac{s(a_{-i},a_{i}\mid t)}{s(a_{i}\mid t)}=s(a_{-i}\mid a_{i},t)
\end{multline*}
and hence also $\lim_{r\rightarrow\infty}EU_i(q_{r};a_{i},t_i)=EU_i(s;a_{i},t_i)=EU_i(\sigma_{-i},a_i;t_i)$.

It follows that $\lim_{r\rightarrow\infty}EU_i(q_{r};a_{i},t_{i})=EU_i(\sigma_{-i},a_{i};t_{i})$
for any player $i\in N$, type $t_{i}$, and action $a_{i}\in A$. From here on we can proceed as above, thus concluding the proof.
\end{proof}

\begin{cor}\label{thm:spohn5cor}
    Let \(\sigma\) be a Bayesian Nash equilibrium and \(\alpha\) a pure strategy profile such that \(EU_i(\alpha\mid t_i)\geq EU_i(\sigma\mid t_i)\) for any player \(i\in N\) and type \(t_i\in T\). Then \(q\), defined as in \Cref{thm:spohn5}, is a dependency equilibrium.
\end{cor}
\begin{proof}
    Using the assumption and \Cref{remark-bne}, we have
    \[EU_i(\alpha\mid t_i)\geq EU_i(\sigma; t_i)\geq EU_i(\sigma_{-i},a_i; t_i)\]
    for any \(i\in N\), \(a_i\in A\), and \(t_i\in T\). Hence, the result follows from \Cref{thm:spohn5}.
\end{proof}

\Cref{thm:spohn5cor} provides a sufficient criterion to see whether an anonymous pure strategy profile $\beta\in A^{m}$ might be a dependency equilibrium.

\begin{cor}\label{cor-pure-de}
Let \(G\) be an additively separable and anonymous ECL Bayesian game. Let $\beta\in A^{m}$ be an anonymous pure strategy Bayesian Nash equilibrium of the game. Then \(\alpha\) is a dependency equilibrium if
\begin{equation}
u_{t,t}(\alpha_t)-u_{t,t}(\beta_t)+(n-1)\sum_{t'\in T}p(t'\mid t)(u_{t',t}(\alpha_{t'})-u_{t',t}(\beta_{t'}))\geq0\label{eq:3}
\end{equation}
for all \(t\in T\).
\end{cor}
\begin{proof}
Since \(G\) is  additively separable and anonymous, an anonymous pure strategy Bayesian Nash equilibrium  \(\beta\in A^m\) exists by \Cref{BNE-unique-anonymous}. Hence, by \Cref{thm:spohn5cor} and \Cref{lemma-anonymous}, \(\alpha\) is a dependency equilibrium if 
\[EU_t(\alpha)\geq EU_t(\beta)\] for all \(t\in T\). Using \Cref{eq:2-1}, it follows
\begin{align}0\leq EU_t(\alpha)-EU_t(\beta)&=u_{t,t}(\alpha_t)
+(n-1)\sum_{t'\in T}p(t'\mid t)u_{t',t}(\alpha_{t'})\\
&\phantom{=}-(u_{t,t}(\beta_t)
+(n-1)\sum_{t'\in T}p(t'\mid t)u_{t',t}(\beta_{t'}))\\
&=
u_{t,t}(\alpha_t)-u_{t,t}(\beta_t)+(n-1)\sum_{t'\in T}p(t'\mid t)(u_{t',t}(\alpha_{t'})-u_{t',t}(\beta_{t'})).
\end{align}
\end{proof}
Intuitively, \(\alpha\) is a dependency equilibrium if \(t\)'s gains from other players choosing \(\alpha\)
are larger than their losses from adopting \(\alpha\) themself. If $n$ is sufficiently large, only
the losses that other players of the same type occur are relevant. 

Lastly, if the distributions \(s_r\) are all uncorrelated, then the dependency equilibrium must be a Bayesian Nash equilibrium. This tells us that if players' actions are independent, then only Bayesian Nash equilibria are relevant even for superrational players.
\begin{prop}\label{uncorrelated-de-is-bne}
    Assume \(s=\lim_{r\rightarrow\infty}s_r\) is a dependency equilibrium, with \((s_r)_r\) as in \Cref{defn:dependency-equilibrium}. Then if for every \(r\in\mathbb{N}\), \(s_r\) is uncorrelated, \(s\) is a Bayesian Nash equilibrium.
\end{prop}
\begin{proof}
    Let \(\sigma\) be the mixed strategy profile corresponding to \(s\), and \(\sigma^r\) the one corresponding to \(s_r\). Let \(i\in N\) arbitrary. It is easy to see that since \(\lim_{r\rightarrow\infty }s_r=s\), also \(\lim_{r\rightarrow\infty}\sigma_{-i}^r=\sigma_{-i}\). Moreover, the expected utility of player \(i\) is a continuous function in \(\sigma_{-i}\). Now let \(t_i\in T\), \(a_i\) such that \(\sigma_i(a_i\mid t_i)>0\), and \(a'_i\in A\). Then also \(s(a_i\mid t_i)=\sigma_i(a_i\mid t_i)>0\) and thus (i) \(\lim_{r\rightarrow\infty}EU_i(s_r;a_i,t_i)\geq\lim_{r\rightarrow\infty}EU_i(s_r;a'_i,t_i)\) by the definition of a dependency equilibrium. It follows that
    \begin{multline}EU_i(\sigma_{-i},a_i;t_i)
    =EU_i(\lim_{r\rightarrow\infty}\sigma_{-i}^r,a_i;t_i)
    =\lim_{r\rightarrow\infty}EU_i(\sigma_{-i}^r,a_i;t_i)
    =\lim_{r\rightarrow\infty}EU_i(s_r;a_i,t_i)
    \\
    \underset{\text{(i)}}{\geq}\lim_{r\rightarrow\infty}EU_i(s_r;a'_i,t_i)
    =\lim_{r\rightarrow\infty}EU_i(\sigma^r_{-i},a'_i;t_i)
    =EU_i(\lim_{r\rightarrow\infty}\sigma^r_{-i},a'_i;t_i)
    =EU_i(\sigma_{-i},a'_i;t_i).\end{multline}
    This shows that \(\sigma\) is a Bayesian Nash equilibrium and thus concludes the proof.
\end{proof}

\subsection{Uncertainty about decision procedures and similarity}
\label{sec:Uncertainty-about-similarity}

The Bayesian game model introduced here does not explicitly incorporate players with different decision procedures, or with different degrees of similarity. While these aspects could still be modeled implicitly, by defining a suitable joint distribution over actions \(s\in S\), it might be valuable to introduce explicit controllable parameters. Moreover, dependency equilibria are based on conditional expected utilities and thus effectively assume that all players act optimally under evidentialist or superrational reasoning. In this section, I will relax this assumption by extending the model to incorporate different decision procedures. My analysis is a generalization of the discussion in \textcite[][Sec.~2.9.4]{Oesterheld2017-qg}. As I will show below, this does not substantially increase the generality of my model. For this reason, I will not use the concepts introduced here in the rest of the report, so this subsection can be skipped.

One approach would be splitting types into different subtypes. Assume that there is finite index set \(\Omega\) for the different subtypes (e.g., specifying the types' decision procedures). Then we can define new subsets of types $T_{\omega}=\{(1,\omega),\dots,(m,\omega)\}$ for each \(\omega\in \Omega\),
and let the new set of types be $T=\bigcup_{\omega\in \Omega}T_{\omega}$. We also need to specify a new prior \(p\) over this bigger set of types \(T\). I assume that utility functions do not depend on the subtype \(\omega\). Having defined these types, we can then restrict the space of possible joint distributions
in $S$ in some way based on types.

I consider a simple binary approach, with two indices: $C$
for cooperators and $D$ for anyone else. 
The cooperators can be thought of as implementing the same or an equivalent decision procedure. Moreover, I assume that these agents maximize conditional expected utility in an ECL Bayesian game, such that we can apply
dependency equilibria to joint distributions over their actions. I assume that the actions of the players with subtype \(D\) are independent from those of the \(C\) players.\footnote{This makes the situation easier to analyze, though I think it is not entirely realistic. Even though the players in $T_D$ are not thought of as superrational cooperators, the players in $T_{C}$ may still have some conditional beliefs
about their actions. This could include the possibility of these agents
being seen as irrational in some way. For instance, players in $T_{C}$
may believe that it is more likely for a player of type $T_{D}$ to
choose an uncooperative action given that they choose a cooperative
action.} In this framework, one could
model gradual beliefs about similarity by being uncertain about
whether another player belongs to $T_{C}$ or not.\footnote{Compare the comment discussion on \textcite{treutlein2018request}, in particular \url{https://forum.effectivealtruism.org/posts/92wCvqF73Gzg5Jnrr/request-for-input-on-multiverse-wide-superrationality-msr?commentId=iXXvEremjJtedccwh}}

For simplicity, I assume an additively separable and anonymous setting. I write \(p(t',\omega'\mid t,\omega)\) for the probability that any player \(j\neq i\) has type \((t',\omega')\), given that player \(i\) has type \((t, \omega)\). Moreover, I define a joint strategy distribution
$s_{C}\in \Delta(A^{T_{C}})$ for all the types in \(T_C\), and a similar distribution \(s_D\) for the types in \(T_D\). Note that here, I take a distribution over actions given types as fundamental, rather than deriving such a distribution from an anonymous joint strategy profile as in \Cref{joint-strategy-distributions}. Given this distribution, one can derive all the relevant probabilities, though I will not explicate this here. I denote with \(s_C(\alpha_{t'}=a'\mid a,t)\) the belief of a player of type \((t,C)\) that any other player of type \((t',C)\) would play action \(a'\), given that the first player plays action \(a\). We also need the marginal probability \(s_D(\alpha_{t'}=a')\) that a player of type \((t',D)\) plays action \(a'\) (since that type's action is independent from the actions of a player in \(T_C\), we do not condition it on anything).

Given joint distributions \(s_C,s_D\) and a type \((t,C)\), we can then use \Cref{eq:2-1} to define expected utilities:
\begin{align}
EU_{t,C}(s_C,s_D; a)
:=u_{t,t}(a)
&+(n-1)\sum_{t'\in T}p(t',C\mid t,C)\sum_{a'\in A}s_C(\alpha_{t'}=a'\mid a,t)u_{t',t}(a')
\\
&+(n-1)\sum_{t'\in T}p(t',D\mid t,C)\sum_{a'\in A}s_D(\alpha_{t'}=a')u_{t',t}(a').
\end{align}
Since the actions of players in \(T_C\) and \(T_D\) are independent, the term for the utility from the \(D\) types does not depend on the action of a player of type \(C\). Hence, we get
\[EU_{t,C}(s_C,s_D;a)-EU_{t,C}(s_C,s_D;\hat{a})
\geq 0\]
if and only if
\[
u_{t,t}(a)-u_{t,t}(\hat{a})+(n-1)\sum_{t'\in T}p(t',C\mid t,C)\sum_{a'\in A}\left(s_C(\alpha_{t'}=a'\mid a,t)-s_C(\alpha_{t'}=a'\mid \hat{a},t)\right)u_{t',t}(a').
\]

We can use this to determine when an anonymous pure strategy profile \(\beta\) is a dependency equilibrium, as in \Cref{cor-pure-de}. To that end, let \(\alpha\) be the unique anonymous pure strategy profile Bayesian Nash equilibrium (this does not depend on the subtypes, since utilities do not depend on subtypes). I will not work this out formally here, but analogously to \Cref{cor-pure-de}, we get the condition
\[
u_{t,t}(\alpha_{t})-u_{t,t}(\beta)+(n-1)\sum_{t'\in T}p(t',C\mid t,C)\left(u_{t',t}(\alpha_{t'})-u_{t',t}(\beta_{t'})\right).
\]
To see what this means, in another abuse of notation, I write \(p(C\mid t',t,C)\) to denote the belief of a player of type \((t,C)\) that another player has subtype \(C\), given that they are of type \(t'\). Then \(p(t',C\mid t,C)=p(C\mid t',t,C)p(t'\mid t,C)\). Hence, we get the condition
\begin{align}
&u_{t,t}(\alpha_t)-u_{t,t}(\beta_t)+(n-1)\sum_{t'\in T}p(t'\mid t,C)p(C\mid t', t,C)\left(u_{t,t'}(\alpha_{t'})-u_{t,t'}(\beta_{t'})\right)\geq 0\label{eq:5}
\\
\Leftrightarrow\quad
&u_{t,t}(\alpha_t)-u_{t,t}(\beta_t)+(n-1)\sum_{t'\in T}p(t'\mid t,C)(1-p(D\mid t', t,C))\left(u_{t',t}(\alpha_{t'})-u_{t',t}(\beta_{t'})\right)\geq 0
\end{align}

We can see that it differs from \Cref{eq:3} in that the weight of each type $t$ is reduced based
on how likely such a player is of subtype \(D\) instead of \(C\). Given sufficiently large $n$, this
is not a problem per se---if $n$ goes to infinity, it does not matter
how much weight there is on $T_{D}$ in total---but it may shift
the relative conditional credences about types of the players. For
instance, a type $(t,C)$ may deem other players with type $t$ more
likely to have subtype $C$, but may be sceptical whether players of types
$t'\neq t$ are of the $C$ subcategory.

Such shifts in the relative weight of the types, due to different coefficients \(p(C\mid t',t,C)\), can be equivalently modeled by an anonymous prior \(p'\) over types, without any subtypes. To sketch an argument for this, note first that we can regard such a prior \(p'\) as a symmetric joint distribution \(p'\in \Delta(T\times T)\). Here, \(p'(t,t')=p'(t',t)\) is the probability that any two distinct players will have types \(t\) and \(t'\). Now we can let
\[p'(t,t'):=\delta^{-1}p((t,C),(t',C)),\]
where \(p\) is the original prior over types and subtypes, and \(\delta:=\sum_{t,t'\in T}p((t,C),(t',C))\) is a normalization constant. Then this is apparently a symmetric probability distribution, and we have
\[p(t',C\mid t,C)
=\frac{p((t',C),(t,C))}{p(t,C)}
=\frac{\delta p'(t',t)}{p(t,C)}
= \frac{p'(t'\mid t) p'(t) \delta}{p(t,C)}
=p'(t'\mid t)\delta'\]
for some constant \(\delta':=\frac{p'(t) \delta}{p(t,C)}\) that only depends on \(t\), but not on \(t'\). This shows that the relative weights of the different types \(t'\), after conditioning on being of type \((t,C)\), are preserved under \(p'\). So, in particular, we have
\begin{align}&\phantom{=}u_{t,t}(\alpha_t)-u_{t,t}(\beta_t)+(n-1)\sum_{t'\in T}p(t',C\mid t,C)\left(u_{t',t}(\alpha_{t'})-u_{t',t}(\beta_{t'})\right)
\\
&=u_{t,t}(\alpha_t)-u_{t,t}(\beta_t)+\delta' (n-1)\sum_{t'\in T}p'(t'\mid t)\left(u_{t',t}(\alpha_{t'})-u_{t',t}(\beta_{t'})\right).
\end{align}
This shows that at least the simplified model of subtypes considered here is not more general than the already introduced type space model. Assuming large \(n\), believing that another player is of subtype \(D\) is equivalent to just giving them less relative weight in the conditional distribution \(p'(t'\mid t)\). For this reason, I will continue without the model introduced in this section in the following.

While I will not pursue this in this report, there may be other, more interesting ways to extend the model introduced here in future work. For instance, one could specify a specific bargaining solution or point on the Pareto frontier for each subtype. The same could be done for other contingent parameters such as disagreement points. One could then analyze how possible gains from trade change with different assumptions about these parameters.

\section{ECL as a Bayesian bargaining problem}
\label{sec:Bargaining-with-incomplete}
In this section, I combine the models from the two previous sections, by defining a bargaining game on top of a Bayesian game (Sections~\ref{formal-setup-bayesian-bargaining} and \ref{strategies-bayesian-bargaining}). To simplify the formal setup and analysis, I will assume from the start that the underlying ECL Bayesian game is additively separable and anonymous. In \Cref{bargaining-theory-bayesian}, I introduce a version of the Nash bargaining solution for incomplete information bargaining games by \textcite{Harsanyi1972}. In \Cref{joint-distributions-equilibria-bayesian-bargaining}, I adapt Bayesian Nash equilibria and dependency equilibria to the Bayesian bargaining setup. To be able to apply dependency equilibria to bargaining problems, I generalize dependency equilibria to joint beliefs over continuous strategy spaces.

Finally, I conclude with several takeaways from the model (\Cref{observations-final}). First, I adapt the results about dependency equilibria from \Cref{subsec:ObservationsEquil}, including \textcite{Spohn2007-fp}'s folk theorem. I then discuss gains from trade given different beliefs in general (\Cref{general-takeaways-gains-from-trade}) and analyze several toy examples (Sections~\ref{different-prior-probs}--\ref{paretotopia}).

\subsection{Formal setup}
\label{formal-setup-bayesian-bargaining}
\begin{defn}[ECL Bayesian bargaining game]
An \emph{ECL Bayesian bargaining game} is a tuple $G=(N,T,(A_{t})_{t\in T},p,d)$,
where
\begin{itemize}
\item $N=\{1,\dots,n\}$ is the set of players;
\item $T=\{1,\dots,m\}$ is a generic set of types;
\item $\mathcal{A}_{t}\subseteq\mathbb{R}^{m}$ is the convex and compact set of actions
for type $t$;
\item $p\colon \Delta(T^{n})$ is an anonymous prior probability
distribution over type vectors, such that each type has positive prior probability (i.e., for any \(i\in N\), \(p(t_i)>0\) for all types \(t_i\in T\)).
\item $d\in\mathbb{R}^{m}$ is the disagreement point.
\end{itemize}
\end{defn}

The set of players and types are the same as in an ECL Bayesian game, but the actions are now different. In the previous model, there
was a finite set of actions $A$, and in the additively separable
and anonymous case, there were utility functions $u_{t,t'}$ for each tuple of types
$t,t'\in T$, specifying the utility that type \(t\) produces for type \(t'\) with their actions. Since we assume additive separability and anonymity from the start, we now directly define sets of actions $\mathcal{A}_{t}$ for each type \(t\in T\), such that each vector \(x_t\in \mathcal{A}_t\) specifies the utilities \(x_{t,t'}\) that that player can produce for any other player of type \(t'\in T\). One could regard $\mathcal{A}_{t}$ as the convex hull of the
image of $A$ under the function $u_{t}:=[u_{t,1},\dots,u_{t,m}]^{\top}$.
That is, if $\Sigma_{t}:=\Delta(A)$
is the set of $t$'s mixed strategies, then 
\[
\mathcal{A}_{t}=\left\{ \sum_{a\in A}\sigma_{t}(a)u_{t}(a)\mid\sigma_{t}\in\Sigma_{t}\right\}.
\]
This corresponds to what was in \Cref{sec:ECL-as-a-bargaining-problem} the feasible set \(F_i\) for an individual
player (though note that we will define separate feasible sets for the setting here later). In the anonymous incomplete information setup, utilities depend only on
the types, so it suffices to have one such set for each type, with
as many dimensions as there are types. As in \Cref{ecl-bargaining-problem-setup}, this set could
also be something other than a simplex---it only needs to be convex
and compact.

\subsection{Strategies and feasible sets}
\label{strategies-bayesian-bargaining}
Turning to pure strategies and expected utilities, I directly introduce strategies that are anonymous and only depend on the types. First,  players of the same type have the same information, so it seems plausible that they would all choose the same action. Second, since they have exactly the same set of actions with the same utilities, they likely have to choose the same option to produce Pareto optimal outcomes (as discussed in \Cref{uniqueness}).

\begin{defn}[Pure strategy profile]
Let $G=(N,T,(\mathcal{A}_{t})_{t\in T},p,d)$ be an ECL Bayesian bargaining
game. Let $\alpha\in \mathcal{A}:=\prod_{t\in T}\mathcal{A}_{t}$. Then $\alpha$ is called a pure
strategy profile.
\end{defn}

Using \Cref{eq:1}, we can define the expected utility of \(\alpha\in \mathcal{A}\) for type \(t\), after updating on observing their own type, as
\begin{equation}\label{expected-utility-anonymous-bargaining}EU_{t}(\alpha):=\alpha_{t,t}+(n-1)\sum_{t'\in T}p(t'\mid t)\alpha_{t',t},\end{equation}
where the first term is the utility provided by the player to themself, and the second term is the utility provided by the \(n-1\) other players in expectation.

Next, a feasible set is the set of vectors of expected utilities for all types that can be produced by pure strategy profiles.
\begin{defn}[Feasible set]
Let $G=(N,T,(\mathcal{A}_{t})_{t\in T},p,d)$ be an ECL Bayesian bargaining
game. Then
\[
F(G):=\left\{ x\in\mathbb{R}^{m}\mid\exists\alpha\in \mathcal{A}\colon\forall t\in T\colon EU_{t}(\alpha\mid t)=x_{t}\right\} 
\]
is the feasible set of $G$.
\end{defn}
As in \Cref{sec:ECL-as-a-bargaining-problem}, I assume that \(d\in F(G)\) and that at least one payoff \(x\in F(G)\) exists that is a strict Pareto improvement, i.e., \(x_i>d_i\) for all \(i\in N\).

Next, we turn to the individual feasible sets. These are sets of vectors of expected utilities for all types \(t'\) that can be produced by type \(t\) with their pure strategies. Here, we have to carefully scale the utilities in \(\mathcal{A}_t\) to satisfy \Cref{expected-utility-anonymous-bargaining}.

\begin{defn}[Individual feasible set]
Let $t\in T$. Define $f^{(t)}\colon\mathbb{R}^{T}\rightarrow\mathbb{R}^{T}$
via its component functions $f_{t'}^{(t)}(y):=(n-1)p(t\mid t')y_{t'}$
for $t'\in T\setminus\{t\}$ and $f_{t}^{(t)}(y)=y_{t}+(n-1)p(t\mid t)y_{t}$.
Then $t$'s individual feasible set is
\[
F_{t}(G):=f^{(t)}(\mathcal{A}_{t}).
\]
\end{defn}

Given this definition, it follows that $F(G)=\sum_{t\in T}F_{t}(G)$. That is, 
similarly to complete information bargaining, the feasible set of \(G\) is the set of sums of vectors from the individual feasible sets. Since the sets $\mathcal{A}_{t}$ for
each $t\in T$ are convex and compact, $F_{t}(G)$ is also convex
and compact, since it is just the image of $\mathcal{A}_{t}$ under the linear
mapping $f^{(t)}$. Hence, the sum $F(G)$ is also convex and compact.

\begin{example}
\label{exa:bargaining-case-2}Assume there are two types,
$1,2$. First, we have to specify the sets of actions. Suppose that there are
diminishing returns for both types' utility functions, such that the sets of actions are $\mathcal{A}_{1}=\mathcal{A}_{2}=\{x\in\mathbb{R}^{2}_{\geq 0}\mid x_{1}^{2}+x_{2}^{2}\leq1\}$
(\Cref{fig:7}). This could be motivated, for instance, by assuming that resources invested are quadratic in the utilities, and both types can allocate at most one unit of resources to both utility functions.

\begin{figure}
\centering{}\includegraphics[width=0.5\textwidth]{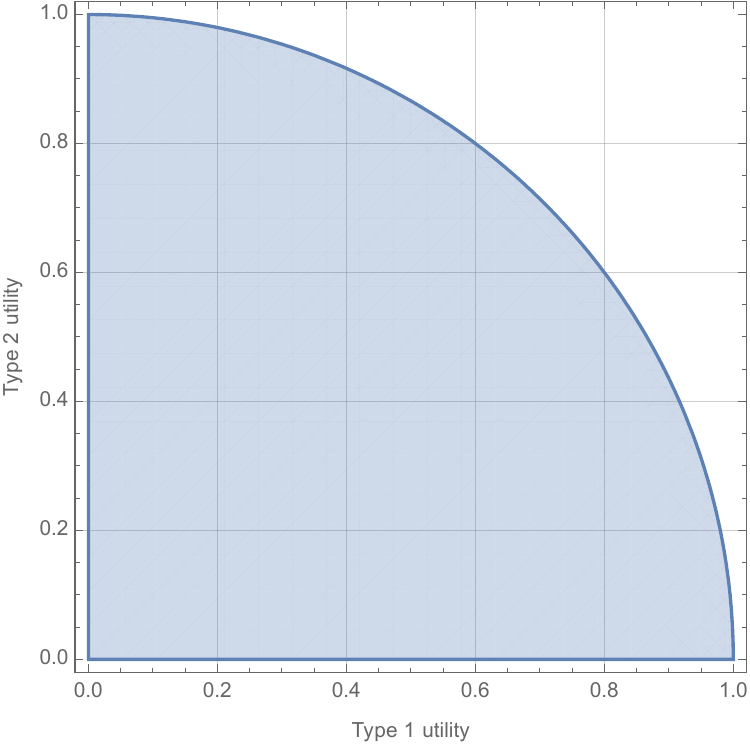}\caption{Set of actions $\mathcal{A}_{1}=\mathcal{A}_{2}$ of either type in \Cref{exa:bargaining-case-2}.}
\label{fig:7}
\end{figure}

Second, we compute the feasible sets $F_{t}(G)$ for both types.
 Say there are $3$ players in total with independent and uniform type distributions, such that $p(1\mid1)=p(1\mid2)=p(2\mid1)=p(2\mid2)=0.5$ (where \(p(t'\mid t)\) is the probability that any player of type \(t\) assigns to any other player having type \(t'\), as defined in \Cref{pure-and-mixed}).
In the feasible set for type $2$, the expected utilities
for type $1$ are lower than for type \(2\), because a player of type $1$ is
certain that they themselves have type $1$ and hence they believe
that there are in expectation two players of type $1$ and only one
type $2$ player. The same applies vice versa. The resulting feasible
sets are depicted in \Cref{fig:8} (a), where $F_{1}=f^{(1)}(\mathcal{A}_{1}),F_{2}=f^{(2)}(\mathcal{A}_{2})$ with \begin{align}
f_{1}^{(1)}(y)&=y_{1}+2\cdot\frac{1}{2}y_{1}=2y_{1}\\
f_{2}^{(2)}(y)&=y_{2}+2\cdot\frac{1}{2}y_{2}=2y_{2}\\
f_{2}^{(1)}&=2\cdot\frac{1}{2}y_{2}=y_{2}\\
f_{1}^{(2)}(y)&=2\cdot\frac{1}{2}y_{1}=y_{1}
.\end{align}
for $y\in \mathcal{A}_{1}$ or $y\in \mathcal{A}_{2}$. The feasible set $F(G)$ is then just the set with the points $x_{t}+x_{t'}$ for all
possible $x_{t}\in F_{t}(G),x_{t'}\in F_{t'}(G)$ (\Cref{fig:8} (b)).

\begin{figure}
\begin{centering}
\subfloat[Individual feasible set $F_{1}(G)$ of the first type (blue) and $F_{2}(G)$ of the second type (orange).]{\includegraphics[width=0.48\textwidth]{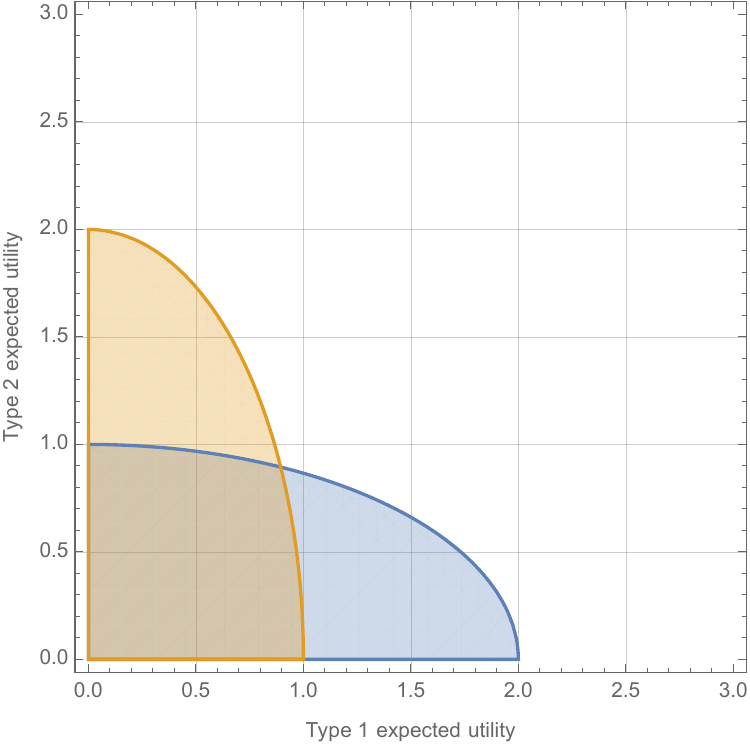}}
\hfill\subfloat[Feasible set $F(G)$.]{\includegraphics[width=0.48\textwidth]{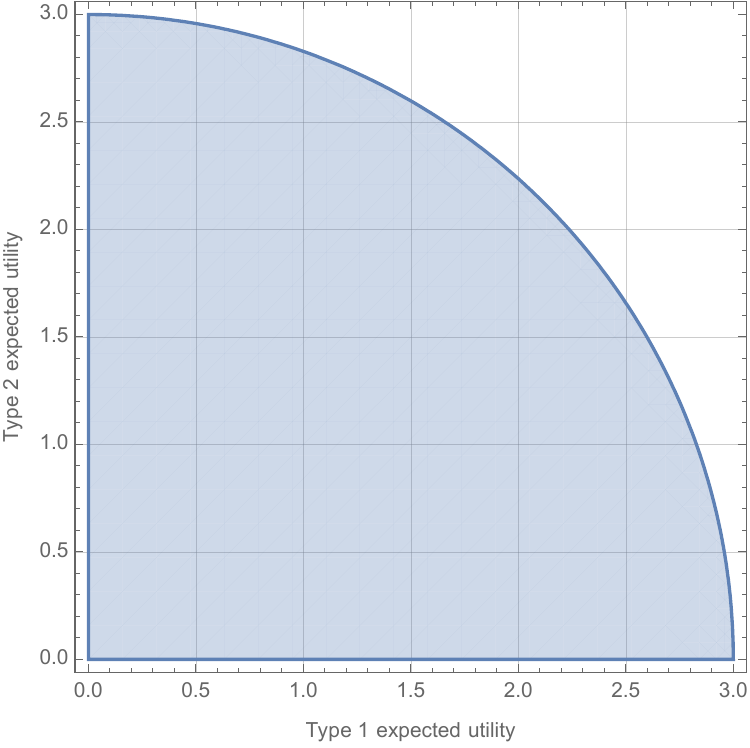}
}\caption{(a) Individual feasible sets and (b) combined feasible set from \Cref{exa:bargaining-case-2}.}\label{fig:8}
\par\end{centering}

\end{figure}

\end{example}

Above example has a type prior $p$ that assigns equal probability to
all combinations of types, but players end up with different feasible sets because they update on the type of their own player and there are only few players in the game. I will give examples with different priors, where \(n\) is very large, in the next section.

\subsection{The Nash bargaining solution for incomplete information games}
\label{bargaining-theory-bayesian}
Here, I define the NBS for an ECL Bayesian bargaining game $G$. I will use this bargaining solution below in my examples to compute cooperative outcomes. \textcite{Harsanyi1972} introduce an axiomatization of the NBS for a two-player incomplete information game, which I discuss in \Cref{appendix-harsanyi-nbs}. This axiomatization includes versions of all the axioms of the complete information NBS discussed in \Cref{nash-bargaining-solution}, and adds two additional axioms to deal with weights for the different types. Below definition is a generalization of \textcite{Harsanyi1972}'s definition for more than two players, adapted to my formal setup. 

\begin{defn}[NBS in an ECL Bayesian bargaining game] \label{defn-nbs-bayesian}Let \(G\) be an ECL Bayesian bargaining game, and let \((\nu_t)_{t\in T}\) be a set of weights for each type, such that $\nu_{t}\geq0$ for all \(t\in T\) and $\sum_{t\in T}\nu_{t}=1$. Then the Nash bargaining solution (NBS) for these weights is defined via the optimization problem
\begin{equation}
\argmax_{x\in F(G)^{\geq d}}\prod_{t\in T}(x_{t}-d_{t})^{\nu_{t}}.
\end{equation}
\end{defn}

\textcite{Harsanyi1972} derive the specific weighting \(\nu_t=p(t)\) for all \(t\in T\), suggesting that the utility of a type should be weighed by the prior probability of that type. This weighting ensures the desirable behavior that the bargaining solution does not change if a type is split into two types with identical actions and utilities, provided their combined probability equates to the probability of the original type. These weights also appear reasonable from an ex ante fairness perspective, given that a type with a higher prior likelihood would, in expectation, occupy a larger number of universes. However, when it comes to fairness, there are other criteria that could be important for determining weights (see \Cref{fairness-and-coalitional}).

\subsection{Joint distributions and equilibria}
\label{joint-distributions-equilibria-bayesian-bargaining}
Finally, I define joint distributions and equilibria in ECL Bayesian bargaining games. We do not need to introduce distributions to define Bayesian Nash equilibria. Action spaces are already convex, and Bayesian Nash equilibria are in any case trivial in the additively separable case.

\begin{defn}[Bayesian Nash equilibrium in an ECL Bayesian bargaining game]\label{defn-BNE-bargaining}
Let \(G\) be an ECL Bayesian bargaining game. Then a strategy profile \(\alpha\in \mathcal{A}\) is a Bayesian Nash equilibrium if for all types \(t\in T\) and actions \(\alpha'_t\in \mathcal{A}_t\), we have
\[EU_t(\alpha)\geq EU_t(\alpha_{-t},\alpha'_t).\]
\end{defn}

Note that the notion of best response here assumes that all players of the same type change their action simultaneously, rather than only a single player deviating. This assumes perfect correlations between players of the same type, which seems inappropriate for the uncorrelated notion of Bayesian Nash equilibria. However, I do not consider this an issue, since, as the next proposition shows, Bayesian Nash equilibria are trivial in the additively separable case.

\begin{prop}\label{bne-unique-bayesian-bargaining}
Let \(G\) be an ECL Bayesian bargaining game. Then the set of Bayesian Nash equilibria is given via \(A^*=\prod_{t\in T}A_t^*\), where
\[A^*_t:=\argmax_{\alpha_t\in \mathcal{A}_t}\alpha_{t,t}\]
for \(t\in T\).
\end{prop}
\begin{proof}
    This follows as a simple exercise from the definition of a Bayesian Nash equilibrium and \Cref{expected-utility-anonymous-bargaining}.
\end{proof}

Next, I introduce joint distributions to define dependency equilibria. In the present model, the sets \(\mathcal{A}_t\) of strategies are continuous, potentially containing all possible (independent) randomizations over a set of actions that a type could implement. This is necessary to enable bargaining. Joint strategy distributions are then separately defined as joint distributions over the space \(\mathcal{A}\). This allows us to express beliefs such as ``if I choose the NBS, other players do the same'', where the NBS is an arbitrary point in the continuous set \(\mathcal{A}_t\).

I define the set \(S\) of joint strategy distributions as the set of measures over the measure space \(\mathcal{A}\subseteq\mathbb{R}^{m,m}\), endowed with the default Borel \(\sigma\)-algebra. For a set \(A_t\subseteq \mathcal{A}_t\), I write \(s(A_t):={s(\{\alpha\in \mathcal{A}\mid \alpha_t\in A_t\})}\), and similarly I define conditionals \[s(A\mid A_t):=\frac{s(\{\alpha\in A\mid \alpha_t\in A_t\})}{s(A_t)}\] for \(A\subseteq\mathcal{A}\) and \(A_t\subseteq\mathcal{A}_t\) such that \(s(A_t)>0\).
For any set \(A_t\subseteq \mathcal{A}_t\) with \(s(A_t)>0\), the expected utility given \(A_t\) is defined as
\[EU_t(s;A_t):=\mathbb{E}_{\alpha\sim s}[EU_t(\alpha)\mid \alpha_t\in A_t]
=\int_{\alpha \in \mathcal{A}}EU_t(\alpha)ds(\cdot\mid A_t).\]
Now we can define dependency equilibria as a generalization of \Cref{defn:dependency-equilibrium}.

\begin{defn}[Dependency equilibrium in an ECL Bayesian bargaining game]\label{defn-DE-bargaining}
Let \(s\in S\) be a joint strategy distribution, and assume there exists a sequence of distributions \((s_r)_{r\in\mathbb{N}}\) that converges weakly to \(s\) such that for each \(r\in \mathbb{N}\), \(s_r\) has full support on \(\mathcal{A}_t\), i.e., such that \(s_r(A_t)>0\) for all nonempty open sets \(A_t\subseteq\mathcal{A}_t\). Then \(s\) is a dependency equilibrium if for all \(A_t\subseteq\mathcal{A}_t\) with \(s(A_t)>0\) and arbitrary nonempty open set \(A'_t\subseteq\mathcal{A}_t\), we have
\[\lim_{r\rightarrow\infty}EU_t(s_r;A_t)\geq \lim_{r\rightarrow\infty}EU_t(s_r;A'_t).\]
We say that \(\alpha\in \mathcal{A}\) is a dependency equilibrium if \(\delta_\alpha\), the Dirac measure with \(\delta_\alpha(A)=1\) if and only if \(\alpha\in A\), is a dependency equilibrium.
\end{defn}
Here, weak convergence means that for any continuous function \(f\colon\mathcal{A}\rightarrow\mathbb{R}\), we have \[\lim_{r\rightarrow\infty}\mathbb{E}_{\alpha\sim s_r}[f(\alpha)]=\mathbb{E}_{\alpha\sim s}[f(\alpha)].\]
I choose weak convergence as a generalization of the pointwise convergence we assumed in the case where \(s\) is a discrete distribution over a finite set of joint actions. Note that weak convergence does not require that the probability of each set converges; for instance, assume that \(s\) is the Dirac measure for some point \(\alpha\). Then we can define \(s_r\) via densities \(f_r\colon \alpha'\mapsto c_r\cdot \exp(-r\Vert \alpha'-\alpha\Vert)\), where \(c_r\) is some normalization constant. \(s_r\) becomes more and more concentrated on \(\alpha\) as \(r\rightarrow\infty\), and thus the integral with respect to \(s_r\) converges to the one with respect to \(\delta_\alpha\) for continuous functions. But \(s_r(\{\alpha\})=0\) for all \(r\in\mathbb{N}\) and \(s(\{\alpha\})=\delta_\alpha(\{\alpha\})=1\).

\subsection{Observations}
\label{observations-final}

In this section, I discuss several takeaways from the model introduced above. I begin by adapting results from \Cref{subsec:ObservationsEquil}. I prove a version of \textcite[][Observation 5]{Spohn2007-fp}, saying that any strategy profile that is a Pareto improvement over a Bayesian Nash equilibrium is a dependency equilibrium. In particular, the NBS with the Bayesian Nash equilibrium disagreement point is a dependency equilibrium.
In \Cref{general-takeaways-gains-from-trade}, I make some general remarks about gains from trade in this model. I show that, given large \(n\), only the beliefs over the types of other players matter. 

I then work through several toy examples with two types. I apply the NBS as a compromise solution and analyze how gains from trade are affected by different assumptions about beliefs and utility functions. I start with a model in which all types have the same posterior beliefs, but different prior weights (\Cref{different-prior-probs}). Next, I analyze the situation in which players' types are correlated, such that players have higher posterior weight for their own type. In this case, gains from trade diminish when players become more confident that other players have the same type. This happens roughly quadratically in a model where utilties are square roots of resource investments (\Cref{double-decrease}), reproducing the ``double decrease'' observed by \textcite{armstrong2017double}. However, given logarithmic returns, as in \textcite{drexler2019pareto}'s ``Paretotopia'' model, gains from trade go down more slowly (\Cref{paretotopia}).

\subsubsection{Dependency equilibria}

To begin, I show that if a strategy profile is at least as good for each type as some other strategy profile, for any possible action they could take, then it is a dependency equilibrium. The proof idea is the same as for \Cref{thm:spohn5}. As a corollary, it follows that Bayesian Nash equilibria and Pareto improvements over Bayesian Nash equilibria are dependency equilibria.

\begin{restatable}{thm}{spohnfivec}
\label{spohn5-continuous}
    Let \(\alpha,\beta\) be two strategy profiles such that for every \(t\in T\) and \(\beta'_t\in \mathcal{A}_t\), we have
    \[EU_t(\alpha)\geq EU_t(\beta_{-t},\beta'_t).\]
    Then \(\alpha\) is a dependency equilibrium.
\end{restatable}
\begin{proof} In \Cref{appendix-proof-of-spohn5-continuous}.
\end{proof}

\begin{cor}\label{spohn5-continuous-cor}
    Let \(\beta\) be a Bayesian Nash equilibrium and \(\alpha\) such that
    \[EU_t(\alpha)\geq EU_t(\beta)\]
    for all \(t\in T\). Then \(\alpha\) is a dependency equilibrium. In particular, any Bayesian Nash equilibrium \(\beta\) is a dependency equilibrium.
\end{cor}
\begin{proof}
    Since \(\beta\) is a Bayesian Nash equilibrium, we have
    \[EU_t(\alpha)\geq EU_t(\beta)\geq EU_t(\beta_{-t},\beta'_t)\]
    for all \(\beta'_t\in \mathcal{A}_t\). Hence, the result follows by \Cref{spohn5-continuous}.
\end{proof}

Lastly, I conclude that the NBS with the Bayesian Nash equilibrium disagreement point is a dependency equilibrium.
\begin{prop}\label{nbs-is-de}
    Let \(\alpha\) be the strategy profile corresponding to the NBS with Bayesian Nash equilibrium disagreement point \(d\). Then \(\alpha\) is a dependency equilibrium.
\end{prop}
\begin{proof}
    The NBS as defined in \Cref{bargaining-theory-bayesian} always chooses a point \(x\) such that \(x_t>d_t\) for all \(t\in T\). Hence, if \(\beta\) is the profile corresponding to the disagreement point \(d\), then \(EU_t(\alpha)\geq d_t= EU_t(\beta)\) for all \(t\in T\). By \Cref{spohn5-continuous-cor}, it follows that \(\alpha\) is a dependency equilibrium.
\end{proof}

\subsubsection{General takeaways about gains from trade}
\label{general-takeaways-gains-from-trade}
Here, I give some general takeaways from the incomplete information bargaining model outlined above. First, assuming additively separable utilities and anonymous strategy profiles allows us to greatly simplify expected utilities received by each type. Recalling \Cref{expected-utility-anonymous-bargaining}, we have
\[EU_t(\alpha)=\alpha_{t,t}+(n-1)\sum_{t'\in T}p(t'\mid t)\alpha_{t',t}.\]
Assuming large \(n\), this becomes
\[EU_t(\alpha)\approx (n-1)\sum_{t'\in T}p(t'\mid t)\alpha_{t',t}.\]
That is, only the expected utility provided by the other players matters. In the following, I will assume large \(n\) such that this is the case (unlike in \Cref{exa:bargaining-case-2}, where I assumed \(n=3\)).

Second, the contributions \(\alpha_{t',t}\) by other players of different types are weighted by type \(t\)'s posterior weight for that type, \(p(t'\mid t)\). The higher the posterior weight \(p(t\mid  t)\), the lower the weight of all other types, reducing the gains from trade from other types cooperating. If players of type \(t\) believe that type \(t'\) does not exist, then that type \(t'\) cannot benefit players of type \(t\). As observed in \Cref{sec:Uncertainty-about-similarity}, uncertainty about the decision procedures of other players or similarities between players similarly factor into expected utilities. It does not matter, for instance, whether a type just cannot benefit other types much, or whether other types believe that the type has low posterior weight.

Positive correlations between players' types reduce gains from trade. In the extreme case in which all players are always of the same type, for instance, no trade is possible. However, if there are no or only small correlations, then trade is possible even given uncertainty about types of other players.

Note that this analysis still depends on the common prior assumption. Relaxing it could lead to further reductions in gains from trade, or it could completely break the analysis.
Moreover, I have not addressed uncertainty about actions, such as which bargaining solution other players will choose. Lastly, it is unclear what happens if the number of types \(m\) is as large as the number of players \(n\). In that case, we cannot simply assume that only the other players matter, as the product \((n-1)p(t'\mid t)\) may stay roughly constant for any given types \(t,t'\). If the different types all have different values, then this could imply a situation more similar to the one with only few players and types.

\subsubsection{Different prior probabilities but equal posterior beliefs}
\label{different-prior-probs}
Here, I analyze a toy model in which two types have different prior weight, but players have the same beliefs over types, since players' type distributions are independent. I continue as in \Cref{exa:bargaining-case-2} with
two different types of players.
\begin{example}
Recall Example \ref{exa:bargaining-case-2}. Assume large \(n\), such that approximately only the expected utilities from the other players matter, and assume
the utility functions are rescaled such that if $p(t\mid t)=0.5$
for $t=1,2$, we have $F_{1}(G)=F_{2}(G)=\{x\in\mathbb{R}_{+}^{2}\mid(2x_{1})^{2}+(2x_{2})^{2}\leq1\}$. I assume the disagreement point is the Bayesian Nash equilibrium, which is the point \((1/2,1/2)\).
Since the sets and the disagreement point are symmetric, the NBS would pick the symmetric point on the Pareto frontier, 
\[2\cdot \left(\frac{\sqrt{2}}{4},\frac{\sqrt{2}}{4}\right)
=\left(\frac{\sqrt{2}}{2},\frac{\sqrt{2}}{2}\right)\in F(G).\]

Now consider a situation in which everyone has the same conditional
beliefs about other players (i.e., types of different players are independent), but one type has lower prior probability, $p(1\mid2)=p(1\mid1)=\frac{3}{4}$
and $p(2\mid2)=p(2\mid1)=\frac{1}{4}$. In this case, the individual feasible
sets get rescaled and we get
\begin{align}
    F_{1}(G)&=\left\{x\in\mathbb{R}_{+}^{2}\,\middle\vert\,\left(\frac{4}{3}x_{1}\right)^{2}+\left(\frac{4}{3}x_{2}\right )^{2}\leq1\right\}\\
    F_{2}(G)&=\left\{x\in\mathbb{R}_{+}^{2}\,\middle\vert\,\left(4x_{1}\right)^{2}+\left(4x_{2}\right)^{2}\leq1\right\},
\end{align}
as displayed in \Cref{fig:9}.
\begin{figure}
\subfloat[Individual feasible set $F_{1}(G)$ of the first type (blue) and $F_{2}(G)$ of the second type (orange).]{\includegraphics[width=0.48\textwidth]{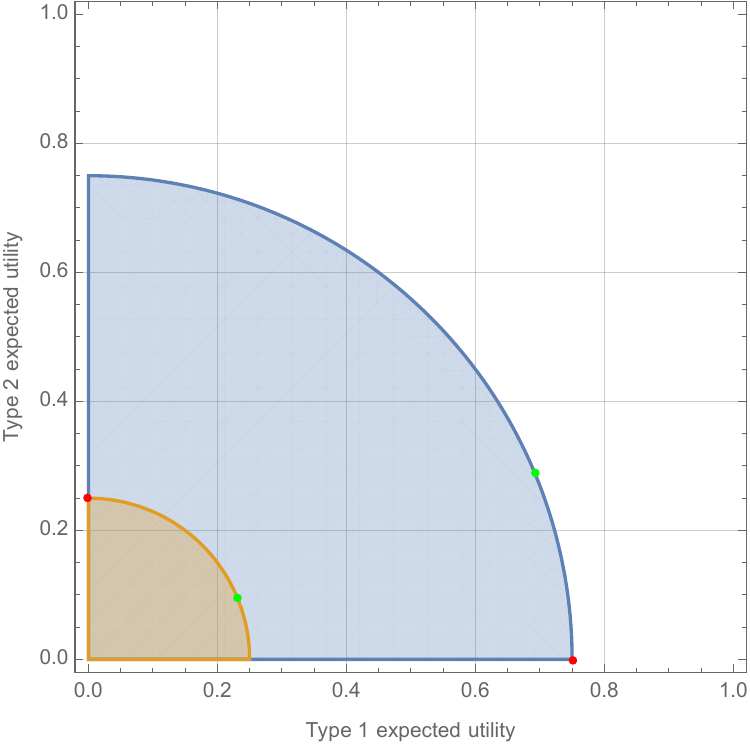}}
\hfill\subfloat[Feasible set $F(G)$.]{\includegraphics[width=0.48\textwidth]{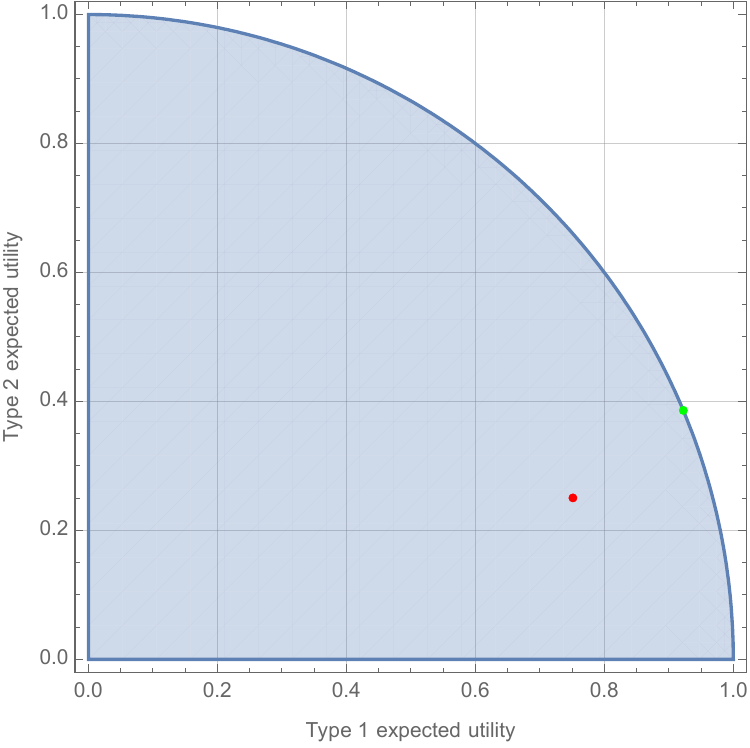}
}\caption{(a) Individual feasible sets of both types and (b) the combined feasible set from \Cref{exa:bargaining-case-3}, with Bayesian Nash equilibrium disagreement point (red dot) and NBS (green dot).}
\label{fig:9}
\end{figure}

 Fewer players have type \(2\) in expectation, so their actions produce less expected utility, both for themselves and for players of the other type. Since the shapes of both Pareto frontiers are the same, the NBS will still pick the same point on \(F_2(G)\) as \(F_1(G)\), only scaled down. However, the Bayesian Nash equilibrium is asymmetric, given by \(d_1=\frac{3}{4}\) and \(d_2=\frac{1}{4}\), and the prior weights of types are also asymmetric. Using the Bayesian Nash equilibrium as disagreement point and the prior probabilities as weights in the NBS, we hence get
 \[\argmax_{x\in F_1(
G)+F_2(G)}(x_1-d_1)^\frac{3}{4}(x_2-d_2)^{\frac{1}{4}}\approx (0.92,0.39).\]
The points in the individual feasible sets corresponding to the NBS and the disagreement point, as well as corresponding points in the overall feasible set, are plotted as green and red dots in \Cref{fig:9}.
\end{example}

\subsubsection{``Double decrease'' given different beliefs and square root utilities}
\label{double-decrease}
Next, I consider a case in which types have different conditional beliefs.
In this situation, if a type deems another type more likely, then
this increases the gains they can receive from trade. Conversely, if a type deems another
less likely, this decreases their potential gains. Since it is more plausible
that players consider their own type more likely than others, I only consider the latter case.

First, I investigate to what degree lower beliefs in the other type decrease gains from trade, given that utilities are square roots of invested resources. \textcite{armstrong2017double} has observed a ``double decrease'' in this case, which is the effect that gains from trade
quadratically decrease with the probability assigned to the other type.

\begin{example}
\label{exa:bargaining-case-3}Assume that types have equal prior weights, but beliefs $p(1\mid1)=p=p(2\mid2)$
and $p(1\mid2)=1-p=p(2\mid1)$. That is, conditional on observing their own type, players of either type believe other players have the same type with probability \(p\), and the other type with probability \(1-p\). For, $p=\frac{3}{4}$, we get the feasible
sets $F_{1}(G)=\{x\in\mathbb{R}_{+}^{2}\mid(\frac{4}{3}x_{1})^{2}+(4x_{2})^{2}\leq1\}$,
$F_{2}(G)=\{x\in\mathbb{R}_{+}^{2}\mid(4x_{1})^{2}+(\frac{4}{3}x_{2})^{2}\leq1\}$
(Figure \ref{fig:9-1}).

Due to the symmetry of the situation, it is easy to see that the NBS always picks points on the individual Pareto frontiers where the Pareto frontier has slope
$-1$ (i.e., the symmetric point on the overall Pareto frontier). Using this, we can compute the point 
\[
\left(\frac{p^{2}}{\sqrt{1-2(1-p)p}},\frac{(1-p)^{2}}{\sqrt{1-2(1-p)p}}\right)\in F_{1}(G)
\]
 for the first type, and the same point with swapped coordinates for the
second type.

Now we compute the share of expected utility received by the other type, as well as the percent gains from trade, at the NBS outcome. The expected utility received by the other player is \(\frac{(1-p)^{2}}{\sqrt{1-2(1-p)p}}\), while the sum of expected utilities received by both types is
\[\frac{p^{2}}{\sqrt{1-2(1-p)p}}+\frac{(1-p)^{2}}{\sqrt{1-2(1-p)p}}
=\sqrt{1 - 2 (1-p) p}.\]
Overall, we get a share of \(\frac{(1-p)^2}{1 - 2 (1-p) p)}\). This is approximately quadratic as \(1-p\rightarrow 0\). Next, turning to gains from trade, the total expected utility from compromise for either player is
\[\frac{p^{2}}{\sqrt{1-2(1-p)p}}+\frac{(1-p)^{2}}{\sqrt{1-2(1-p)p}}=\sqrt{1-2(1-p)p}.\]
The individually achievable expected utility is $p$. The gain from trade
in percentage of disagreement expected utility is hence $\frac{\sqrt{1-2(1-p)p}}{p}-1$.

We plot share of expected utilities received by the other type as well as percent gains from trade for \(\frac{1}{2}\leq p\leq 1\) in \Cref{fig:10} (where \(p=1-p(t'\mid t)\) for \(t'\neq t\)). Interestingly, gains from
trade as a percentage of individually attainable utility decline even faster than share of expected utility received by the other type. Overall, this confirms \textcite{armstrong2017double}'s observation of a ``double decrease'' in the square root utility model. 

\begin{figure}
\centering{}
\subfloat[Individual feasible set $F_{1}(G)$ of the first type (blue) and $F_{2}(G)$ of the second type (orange).]{\includegraphics[width=0.48\textwidth]{plots_2023/ex_46_individual.pdf}}
\hfill\subfloat[Feasible set $F(G)$.]{\includegraphics[width=0.48\textwidth]{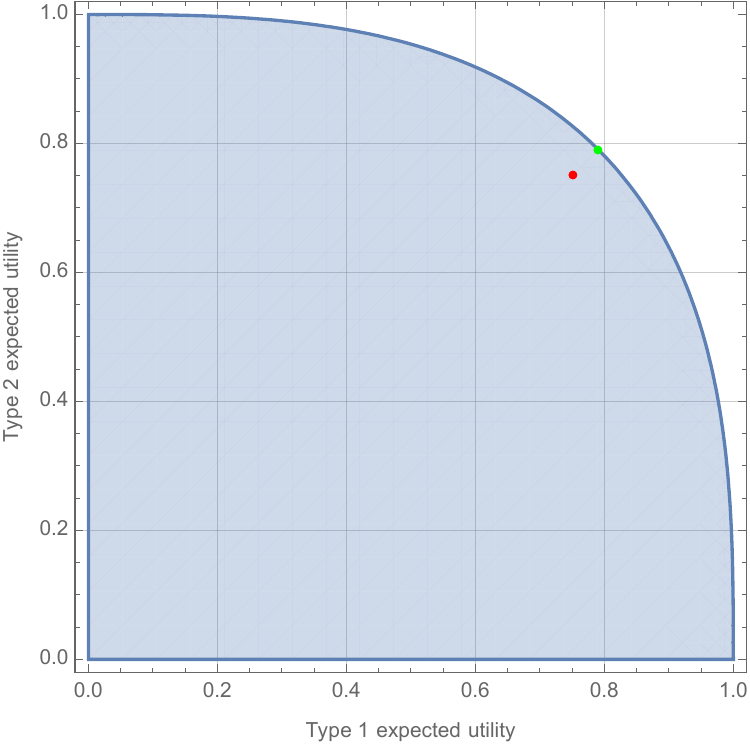}}
\caption{Feasible sets in the square root utility case (\Cref{exa:bargaining-case-3}), where $p(1\mid1)=p(2\mid 2)=\frac{3}{4}$
and $p(2\mid1)=p(1\mid 2)=\frac{1}{4}$. With Bayesian Nash equilibrium disagreement point (red dot) and NBS (green dot).}
\label{fig:9-1}
\end{figure}
\end{example}

\begin{figure}
\centering{}\includegraphics[width=0.96\textwidth]{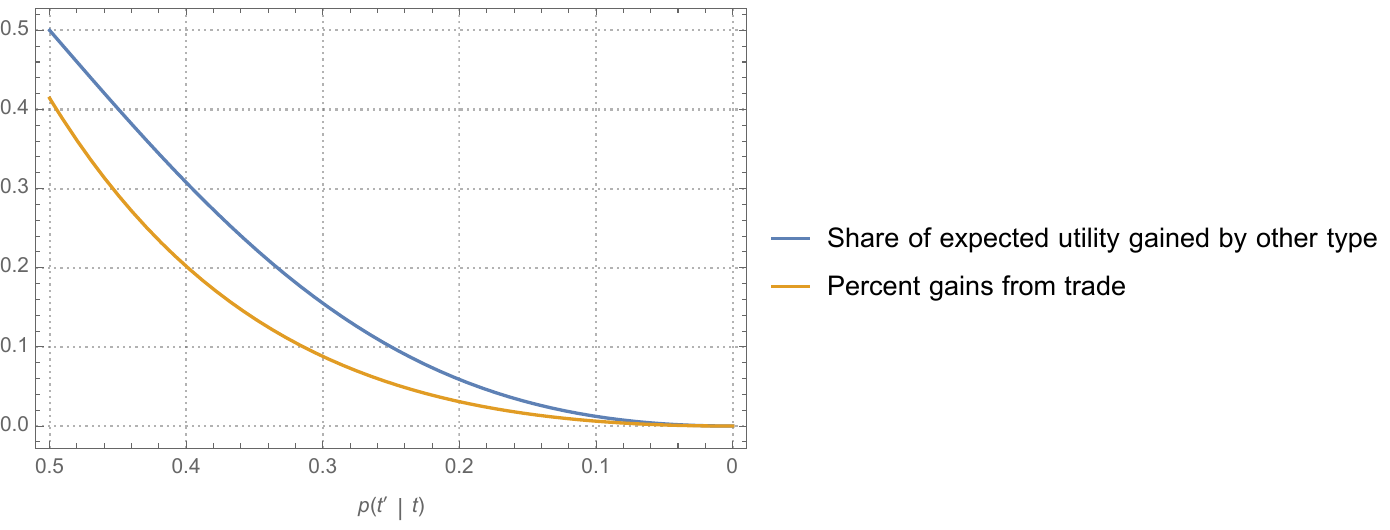}
\caption{Share of expected utility received by the other type (blue) and percent gains from trade (orange), at the NBS, for \Cref{exa:bargaining-case-3}. \(p(t'\mid t)\), for \(t'\neq t\), is the credence of a player of type \(t\) that any other player has type \(t'\). When \(p=\frac{1}{2}\), all players have the same uniform posterior distribution over types, leading to maximal gains from trade. When \(p(t'\mid t)=0\), players think that all other players are of the same type, so no trade is possible. As \(p(t'\mid t)\) goes to zero, gains from trade decrease approximately quadratically, exhibiting a ``double decrease'' \parencite{armstrong2017double}.}
\label{fig:10}
\end{figure}

\subsubsection{``Paretotopia'' given logarithmic utilities}
\label{paretotopia}
While the previous example demonstrates a ``double decrease'', this relies on the particular shape of Pareto frontiers in that example. In this section, I show a different result in the case of logarithmic utilities. Given logarithmic utilities, gains from trade can be very large, and Pareto frontiers are shaped in a way that makes compromise expected utilities and gains from trade change less as the belief in the other type goes down. This relates to \textcite{drexler2019pareto}'s idea of a ``Paretotopia'' in the case of logarithmic returns to resources, where reaping gains from trade at all is vastly more important to players than increasing their share of the compromise outcome.

\begin{example}
\label{bargaining-case-4}
As before, let $p(t\mid t)=p$ and $p(t'\mid t)=1-p$ for $t\neq t'\in\{1,2\}$. Assume feasible sets 
are given by
\begin{align}F_{1}(G)&=\left \{x\in\mathbb{R}_{+}^{2}\,\middle\vert\,\exp\left(\frac{1}{p}x_{1}\right)+\exp\left(\frac{1}{1-p}x_{2}\right)\leq r\right\}\\
F_{2}(G)&=\left\{x\in\mathbb{R}_{+}^{2}\,\middle\vert\,\exp\left(\frac{1}{1-p}x_{1}\right)+\exp\left(\frac{1}{p}x_{2}\right)\leq r\right\}
\end{align}
as illustrated in \Cref{fig:12}. We could interpret this as a case in which
resources produce logarithmic utility for either value system and where $r$ is the amount
of available resources. For symmetry reasons, the NBS is again
the point on the Pareto frontiers where the frontier has slope $-1$, as long as that point is in the feasible set. This is the point $(p\log(pr),(1-p)\log((1-p)r))$
for type $1$, with swapped coordinates for type $2$, for \(p\leq \frac{1}{r}\). For \(p>\frac{1}{r}\), no trade is possible, and players just optimize for their own values. $p\log(r)$ is the amount of utility either type can produce for themself.

Performing the same calculations as in Example \ref{exa:bargaining-case-3},
we get
\[\frac{(1-p)\log(\max\{(1-p)r,1\})}{p\log(pr)+(1-p)\log(\max\{(1-p)r,1\})}\]
as the share of expected utility received by the other type, and
\[\frac{p\log(pr)+(1-p)\log(\max\{(1-p)r,1\}}{p\log(r)}-1\]
percent gains from trade.
I plot both functions for $p\in [\frac{1}{2},1]$, for the cases $r=100$ and $r=10^{9}$, in \Cref{fig:11}.

\begin{figure}
\centering{}
\subfloat[Feasible set \(F_1(G)\) for type \(1\), for \(p(t'\mid t)=\frac{1}{2}\).]{\includegraphics[width=0.48\textwidth]{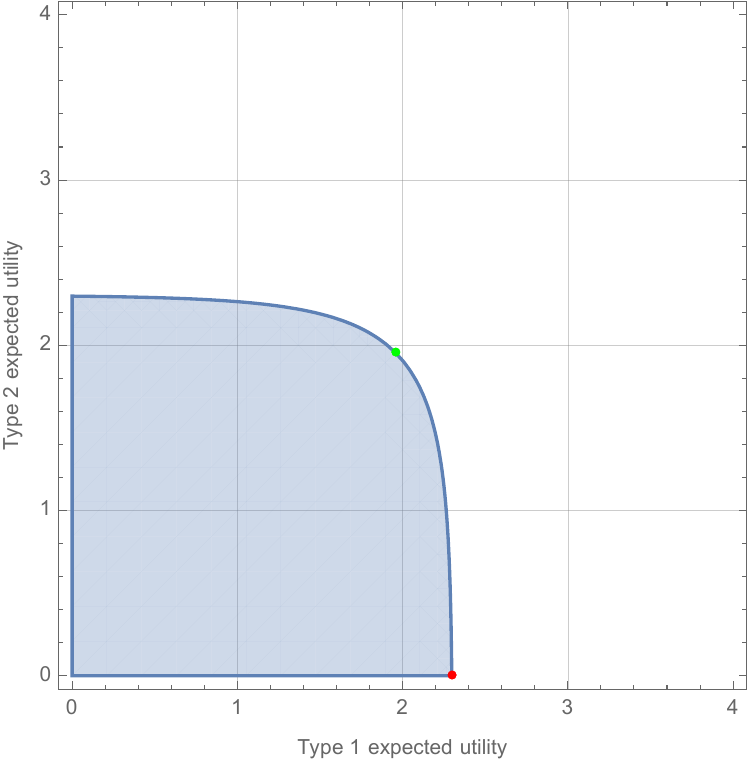}}
\hfill\subfloat[Feasible sets \(F_1(G)\) (blue) and \(F_2(G)\) (orange) for type \(1\) and \(2\), for \(p(t'\mid t)=\frac{1}{2}\).]{\includegraphics[width=0.48\textwidth]{plots_2023/ex_47_p34.pdf}}
\caption{Individual feasible sets in the logarithmic utility case (\Cref{bargaining-case-4}), with Bayesian Nash equilibrium disagreement point (red dot) and NBS (green dot), for \(p(t'\mid t)=\frac{1}{2}\) and \(p(t'\mid t)=\frac{1}{4}\).}
\label{fig:12}
\end{figure}

\begin{figure}
\subfloat[\(r=100\).]{\includegraphics[width=0.48\textwidth]{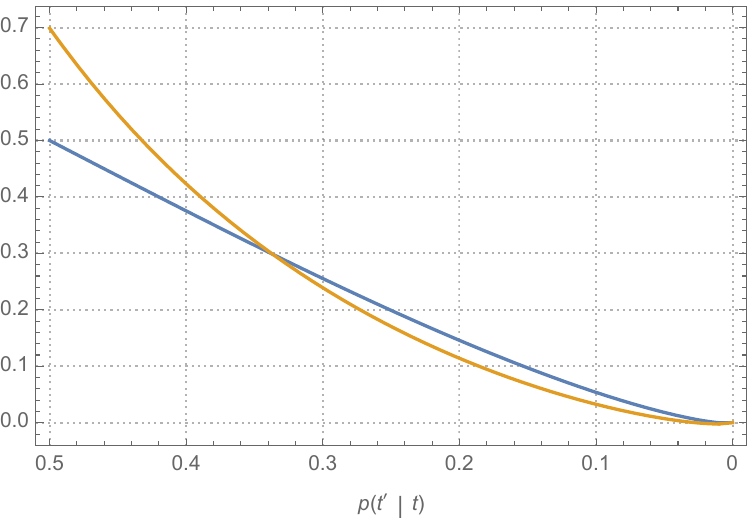}

}\subfloat[\(r=10^9\).]{\includegraphics[width=0.48\textwidth]{plots_2023/double_decrease_log_r10-9.pdf}

}\caption{Share of expected utility received by the other type (blue) and percent gains from trade (orange), given the feasible sets from \Cref{bargaining-case-4}, if beliefs in the respective other type \(p(t'\mid t)\) for \(t'\neq t\) range between \(\frac{1}{2}\) and \(0\). \(r\) is the total available amount of resources. Since utilities are logarithmic in resources, gains from trade go down more slowly than in \Cref{exa:bargaining-case-3}.}
\label{fig:11}
\end{figure}

\end{example}

Both share of expected utility received by the other player and percent gains from trade decrease much more slowly than in \Cref{exa:bargaining-case-3}, particularly for the case in which the amount of resources is large and thus gains from trade are vast. Given \(r=10^9\), both percent gains from trade and share of expected utility received by the other type appear to go down approximately \emph{linearly} as \(p\rightarrow 1\).
This shows that the shape of the Pareto frontier determines how gains from trade are affected by differing posterior beliefs. In future work, it would be interesting to extend this analysis, for instance, by investigating situations in which Pareto frontiers are asymmetric. %All of these factors will likely influence gains from trade.

\section{Discussion}
\label{sec:Discussion-of-this}
In this section, I discuss two issues that arise in my model.
First, I discuss the problem of choosing a disagreement point (\Cref{subsec:The-question-of}). I define the \emph{threat game} disagreement point, which is an equilibrium of a game in which players choose disagreement actions to improve their bargaining position. I discuss an axiomatization that supports the threat point, and show that the NBS with this disagreement point is a dependency equilibrium, even though it can be worse for some players than the Nash equilibrium. I also discuss reasons against its relevance to ECL.

Second, I discuss coalitional stability (\Cref{fairness-and-coalitional}). A compromise outcome is coalitionally stable and thus in the \emph{core} of a game if no subset of players can unilaterally guarantee its members higher payoffs. I argue that stability is a desirable property in the ECL context. Unfortunately, the NBS with the Nash equilibrium or threat game disagreement point is sometimes not stable. I investigate the existence of stable allocations and show that in an additively separable game, the core is always nonempty. However, I also show that sometimes all core allocations make some players worse off than a Nash equilibrium, providing a strong argument against the Nash equilibrium disagreement point. I conclude by suggesting an alternative disagreement point that guarantees stability.

\subsection{Disagreement points}
\label{subsec:The-question-of}

The bargaining model introduced in \Cref{sec:ECL-as-a-bargaining-problem} requires a disagreement
point, an outcome that is obtained if players do not reach an agreement.
For ECL, a plausible choice for a disagreement point is a Bayesian Nash equilibrium, which is unique in an anonymous and additively separable game, up to each type's choice of an action that optimizes their own values (Propositions~\ref{BNE-unique-anonymous} and \ref{bne-unique-bayesian-bargaining}). This is the outcome that players would plausibly choose absent any dependencies
between players. In particular, in the Bayesian game model from \Cref{sec:ECL-as-a-Bayesian-Game}, I showed that this is the only dependency equilibrium in this case (\Cref{uncorrelated-de-is-bne}). I also showed in the model from \Cref{sec:Bargaining-with-incomplete} that the NBS with the Bayesian Nash equilibrium disagreement point is a dependency equilibrium (\Cref{nbs-is-de}).

Unfortunately, I will show in \Cref{example-empty-ne-core} in \Cref{fairness-and-coalitional} that sometimes no point that is a weak Pareto improvement over the Bayesian Nash equilibrium is coalitionally stable (even if a stable point exists in principle, i.e., if the core of the game is nonempty). This strongly suggests that the Nash equilibrium may not be the right disagreement point.

Similar to bargaining solutions, one can also find a disagreement point by positing axioms that constrain the possible choices for disagreement points, or by setting up a noncooperative game and analyzing its equilibria.  As argued in \Cref{sec:ECL-as-a-bargaining-problem}, both approaches can provide relevant insights for ECL, even if ECL does not involve any actual bargaining.
\textcite{Nash1953} provides both an axiomatization and a noncooperative game that implies the ``threat game'' disagreement point. This point represents the equilibrium of a game where players choose disagreement actions and receive as payoffs the Nash bargaining solution computed with these disagreement actions. I define this point here for the setup from \Cref{sec:Bargaining-with-incomplete}.

For the following definition, I assume that \(\mu^\nu(F(G),d)\) is defined as the NBS with weights \(\nu\), computed for all the types \(t\) for which there exists \(x\in F(G)^{\geq d}\) such that \(x_t>d_t\). Note that since \(F(G)^{\geq d}\) is convex, if such a point \(x\) exists for all types \(t\in P\subseteq T\), then there also exists a point \(x'\in F(G)^{\geq d}\) such that \(x'_t>d_t\) for all \(t\in P\) simultaneously. Hence, we can define
\[\mu^\nu(F(G),d):=\argmax_{x\in F(G)^{\geq d}}\prod_{t\in P}(x_t-d_t)^{\nu_t}.\]
 \begin{defn}[Threat game disagreement point]
     Let \(G\) be an ECL Bayesian bargaining game. Then the \emph{threat game} disagreement point or \emph{threat point} is a point \(d\in F(G)\) such that there exists a strategy profile \(\alpha\in \mathcal{A}\) with \(d_t=EU_t(\alpha)\) for all types \(t\in T\), and for any \(t\in T\) and \(\alpha'_t\in\mathcal{A}_t\), letting \(d':=(EU_{t'}(\alpha_{-t},\alpha'_t))_{t'\in T}\), we have
     \[\mu_t^\nu(F(G),d)\geq \mu^\nu_t(F(G),d').\]
 \end{defn}
This definition says that the threat point is a point \(d\), corresponding to a strategy profile \(\alpha\in\mathcal{A}\), such that no type can improve their bargaining outcome by changing their action in \(\alpha\). \textcite{Nash1953} shows that the threat point exists and is unique in his two-player bargaining game. I believe Nash's proof translates to my setup at least with respect to existence, though uniqueness could be violated if there are more than two players.

In \textcite{Nash1953}'s axiomatization of the NBS with the threat point, there exists a feasible set
$F$ together with two sets $S_{1},S_{2}$ that contain the possible
disagreement strategies for players $1,2$. In addition to versions
of Axioms \ref{axm:1} and \ref{axm:Independence-of-irrelevant},
\textcite[][p.~137]{Nash1953} requires the following axioms:
\begin{ax}A restriction of the set of strategies available to a player cannot
increase the value to him of the game. That is, if $S_{1}'\subseteq S{}_{1}$,
then $\mu_{1}(S_{1}',S_{2},F)\leq\mu_{1}(S_{1},S_{2},F)$. The same
applies for the second player.
\end{ax}

\begin{ax}There is some way of restricting both players to single strategies
without increasing the value to player one of the game. That is,
there exist $s_{1}\in S_{1},s_{2}\in S_{2}$ such that $\mu_{1}(\{s_{1}\},\{s_{2}\},F)\leq\mu_{1}(S_{1},S_{2},F)$.
The same applies for the second player.
\end{ax}

It follows from those axioms that the bargaining solution $\mu$ will be the NBS with the threat game disagreement point. Note that the
axioms and Nash's proof require a separate set for disagreement strategies,
so this does not directly translate to my setting. However, it seems plausible that one may be able to extend the result.

Note that, even in the two-player case, the NBS with the threat point can be worse for a player than a Nash equilibrium.
\begin{example}\label{nbs-worse-than-ne}
Take the game with two players $1,2$ and actions \(a_1,a_2,a_3\) and \(b_1,b_2\), respectively, given by \Cref{tbl:6} (a). 
\begin{table}
\begin{centering}
\subfloat[]{\begin{centering}
\begin{tabular}{|c|c|c|c|}
\hline 
 & $a_{1}$  & $a_{2}$  & $a_{3}$\tabularnewline
\hline 
\hline 
$b_{1}$  & $4,4$  & $2,5$  & $-4,2$\tabularnewline
\hline 
$b_{2}$  & $5,2$  & $3,3$  & $-3,2$\tabularnewline
\hline 
\end{tabular}
\par\end{centering}
}\subfloat[]{\begin{centering}
\begin{tabular}{|c|c|c|c|}
\hline 
 & $a_{1}$  & $a_{2}$  & $a_{3}$\tabularnewline
\hline 
\hline 
$b_{1}$  & $7,2$  & $5,3$  & $-1,0$\tabularnewline
\hline 
$b_{2}$  & $8,0$  & $6,1$  & $0,0$\tabularnewline
\hline 
\end{tabular}
\par\end{centering}
}
\par\end{centering}
\caption{Two-player bargaining game with (a) the payoffs for both players given all pure strategy profiles, and (b) the payoffs normalized by the threat point.}
\label{tbl:6} 
\end{table}
Here, the threat game disagreement point would be $(-3,2)$, since given actions \((a_3,b_2)\) as disagreement point, none of the players can change their action to improve their bargaining outcome. Normalizing by this point leads to the
payoffs in \Cref{tbl:6} (b). The feasible set, alongside the relevant points, is illustrated in \Cref{fig:2}.
\begin{figure}
\begin{centering}
\includegraphics[width=0.5\textwidth]{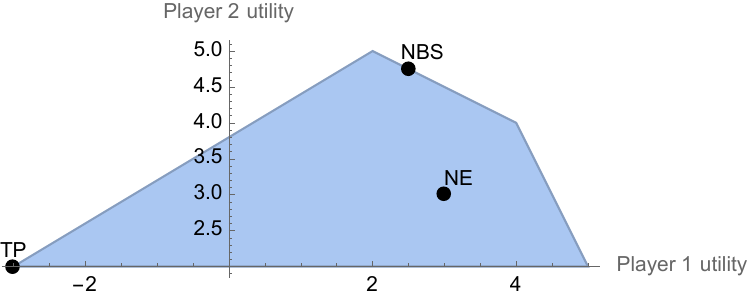}
\par\end{centering}
\caption{Feasible set for \Cref{nbs-worse-than-ne}, with threat game disagreement point (TP), Nash equilibrium (NE), and the Nash bargaining solution (NBS) based on the threat point.}
\label{fig:2} 
\end{figure}
One can calculate the NBS as the point $(5.5,2.75),$ which is worse for player
$1$ than the Nash equilibrium $(3,3)$. 
\end{example}

An important question when it comes to ECL is whether there is a dependency equilibrium supporting the NBS with the threat point. This gives at least some basic plausibility to joint beliefs that imply this compromise outcome. Despite it potentially being worse than a Nash equilibrium, this is the case.
\begin{prop}
    Let \(\alpha\in\mathcal{A}\) be a strategy profile corresponding to the NBS with the threat game disagreement point. Then \(\alpha\) is a dependency equilibrium.
\end{prop}
\begin{proof}
    Let \(\beta\) be the strategy corresponding to the threat point \(d\). We show that \(EU_t(\alpha)\geq EU_t(\beta_{-t},\beta_t')\) for all \(t\in T\) and \(\beta_t'\in\mathcal{A}_t\). Then the result follows using \Cref{spohn5-continuous}.

    Towards a contradiction, assume that there exists \(t\in T\) and \(\beta'_t\) with \(EU_t(\alpha)<EU_t(\beta_{-t},\beta'_t)\). Then, defining \(d':=(EU_{t'}(\beta_{-t},\beta'_t))_{t'\in T}\), we have \(\mu^\nu_t(F(G),d')\geq d'_t>EU_t(\alpha)=\mu^\nu_t(F(G),d)\). But this is a contradiction to the definition of a threat point \(d\).
\end{proof}

One problem with the threat point in conventional bargaining is that it supposes an ability to commit to a non-equilibrium action in case no agreement is reached. Insofar as humans cannot credibly commit to certain actions, this suggests that it may not be an appropriate solution concept for bargaining problems between humans. Another concern with the threat point is that it potentially leads to an agreement reached through coercion. It seems reasonable to assume that rational agents should refuse to give in to such coercion. Therefore, if the opponent commits to pursuing a threat in case no agreement is reached, one should not take this as a disagreement point for evaluating gains from trade.\footnote{Note that, in general, the distinction between extortion and a fair trade depends on some assignment
of a default outcome \parencite[][]{armstrong2016extortion}. In an additively separable game, the Nash equilibrium is a plausible non-threat default outcome, but it leads to coalitional instability. I will turn to defining an alternative non-threat default outcome in the next section.} It is unclear how these considerations apply to ECL, though it seems plausible that threats should be even less relevant to ECL than to conventional bargaining.

Overall, disagreement points are an important area of further study for ECL. Some recent work on threat-resistant bargaining may be particularly relevant \parencite{diffractor2022rose}. Moreover, it would be interesting to investigate acausal bargaining models to gain insights into the question \parencite[e.g.][]{diffractor2018cooperative,kosoy2015superrationality}.

\subsection{Coalitional stability}
\label{fairness-and-coalitional}

Another issue with ECL is coalitional
stability. In a \emph{coalitional game} \parencite[][Pt.~4]{osborne1994course}, players can choose to cooperate with a smaller coalition (subset of players), ignoring the remaining players. This is different from a bargaining game, where all players have to agree to a compromise. A bargaining solution is coalitionally stable if no coalition can unilaterally ensure higher payoffs for their members. In the ECL case, it seems possible for superrationalists
to choose to cooperate with a subset of players rather than with all players (the ``grand coalition''). Hence, coalitional stability is an important desideratum for a bargaining solution in the ECL case.\footnote{Issues with coalitional stability in ECL were also informally discussed by \textcite{gloor2018commenting2}.}

In the following, I will focus on a complete information bargaining model for simplicity. To formalize coalitional stability, let $P\subseteq N$ be an arbitrary coalition. Then we define $\nu(P)\subseteq\mathbb{R}^{n}$ as the set of payoffs $x\in\mathbb{R}^{n}$ for all players
such that the players in $P$ can achieve at least as much for themselves via a collective action. Depending on assumptions about the responses by the remaining players, this can be formalized in different ways, leading to different functions \(\nu\). In any case, we have \[\nu(N)=\{x\in\mathbb{R}^n\mid \exists y\in F(B)\colon \forall i\in N\colon x_i\leq y_i\}.\] 
Given a function \(\nu\), the \emph{core} \(C^\nu(B)\) is defined as the set of all payoffs \(x\in\nu(N)\) such that no coalition can guarantee their members strictly higher payoffs.
\begin{defn}[Core]
 The  \emph{core} of the bargaining game \(B\) with respect to \(\nu\) is the set
 \[C^\nu(B):=\{x\in \nu(N)\mid \forall P\subseteq N\colon \forall y\in \nu(P)\colon\exists i\in P\colon x_i\geq y_i\}.\]
 \end{defn}

A standard definition for \(\nu\) is the set of \(\alpha\)-effective vectors, which assumes worst-case responses by the remaining players. Formally, $x\in\nu^\alpha(P)$
if and only if there is $\sigma_{P}\in\prod_{i\in P}\Sigma_{i}$ such
that for all $\sigma_{-P}\in\prod_{j\in N\setminus P}\Sigma_{j}$,
 $u_{i}(\sigma_{P},\sigma_{-P})\geq x_{i}$ for all $i\in P$. The corresponding core \(C^\alpha(B)\) is called the \(\alpha\)-core.
This definition allows for threats to discourage formation of a coalition. In general, it is unclear whether unfriendly actions like
these should play a role (see \Cref{subsec:The-question-of}).

I also consider another way to define \(\nu\) that does not involve outright threats. For simplicity, I assume additively separable utility functions. Assume that there is some worst case payoff matrix \(A\in\mathbb{R}^{n,n}\), specifying for each player \(i\in N\) payoffs \(A_{i,j}\in \{x_{i,j}\mid x_i\in F_{i}(B)\}\) they may produce for a player \(j\in N\), if they are left out of the coalition. Then I define \(\nu^A\) as the set of \(A\)-effective vectors, via \(x\in \nu^A(P)\) if there exists \(y_j\in F_j(B)\) for all \(j\in P\) such that for each coalition member \(i\in P\), we have
\[x_i\leq \sum_{j\in P}y_{j,i}+\sum_{j'\in N\setminus P}A_{j',i}.\]
That is, we assume coalition members contribute payoffs \(y_{j,i}\) and the remaining players payoffs \(A_{j',i}\) to player \(i\). For instance, \(A\) may represent Nash equilibrium payoffs, i.e., \(A_i\in\argmax_{x_i\in F_i(B)}x_{i,i}\) for each \(i\in N\). The \emph{\(A\)-core} \(C^A(B)\) is defined analogously to the \(\alpha\)-core, but using the \(A\)-effective vectors \(\nu^A\). In the case where \(A\) represents the Nash equilibrium, I also write \(C^{NE}(B)\) for the Nash equilibrium core.

Note that by definition, the \(\alpha\)-core is the largest possible core.
\begin{prop}\label{d-core-alpha-core}
    Let \(B\) be a bargaining problem with additively separable utilities. Then for any \(A\in\mathbb{R}^{n,n}\) with \(A_{i,j}\in \{x_{i,j}\mid x_i\in F_{i}(B)\}\) for \(i,j\in N\), we have \(C^A(B)\subseteq C^\alpha(B)\).
\end{prop}
\begin{proof}
    Let \(x\in C^A(B)\) arbitrary. Let \(P\subseteq N\) and \(y\in \nu^\alpha(P)\). To show \(x\in C^\alpha(B)\), we have to show that there exists at least one \(i\in P\) such that \(x_i\geq y_i\). We know by definition of \(\nu^\alpha\) that there exists \(\sigma_{P}\in \Sigma_P\) such that \(y_i\leq u_i(\sigma_P,\sigma_{-P})\) for all \(i\in N\) and \(\sigma_{-P}\in \Sigma_{-P}\). In particular, let \(\sigma_{j}\) be the strategy corresponding to \(A_{j}\) for \(j\notin P\). Then we know
    that
    \(y_i\leq u_i(\sigma_P,\sigma_{-P})\) for all \(i\in N\). Letting \(\hat{x}_j\in F_j(B)\) corresponding to \(\sigma_j\) for \(j\in P\), it follows
   \[y_i\leq u_i(\sigma)=\sum_{j\in P}\hat{x}_{j,i} + \sum_{j'\in N\setminus P}A_{j',i}\]
   for all \(i\in N\). Hence, we have \(y\in \nu^A(P)\).
   By the definition of \(C^A(B)\), there thus exists \(i\in P\) such that \(x_i\geq y_i\). This concludes the proof.
\end{proof}

Now we turn to analyzing the existence of core allocations. First, we show that the NBS is not necessarily in the \(\alpha\)-core, even if the core is nonempty.

\begin{example}\label{example-nbs-notin-core}
Consider the game of three players $1,2,3$, where each player has
three options $1,2,3$. The (additively separable)
utilities generated by each player taking any of the options are specified
in \Cref{tbl:9}.

\begin{table}
\begin{centering}
\begin{tabular}{|c|c|c|c|}
\hline 
Player\textbackslash Options & $1$ & $2$ & $3$\tabularnewline
\hline 
\hline 
$1$ & $(3,0,0)$ & $(2.5,2.5,0)$ & $(2,2,2)$\tabularnewline
\hline 
$2$ & $(0,3,0)$ & $(2.5,2.5,0)$ & $(2,2,2)$\tabularnewline
\hline 
$3$ & $(0,0,3)$ & $(0.5,0.5,0)$ & $(0.5,0.5,0)$\tabularnewline
\hline 
\end{tabular}
\par\end{centering}
\caption{Utility vectors \(u_{i}\) generated by three options \(a_i=1,2,3\) for each player \(i=1,2,3\), corresponding to the vertices spanning the individual feasible sets \(F_i(B)\) in \Cref{example-nbs-notin-core}.}
\label{tbl:9}
\end{table}
The unique Nash equilibrium disagreement point is \(d=(3,3,3)\). This is also the threat game disagreement point---no matter the actions of the other players, a player can always increase their bargaining position by moving to action \(1\) and thus raising their own disagreement payoff and reducing that of the other players.

Now, if $1$ and $2$ form a coalition, they can both guarantee each other a payoff of \(5\) each, so
\[\nu^\alpha(\{1,2\})=\{x\in\mathbb{R}^3\mid x_1\leq 5, x_2\leq 5\}.\]
However, the NBS payoffs can be computed as \(x=(4.33,4.33,5.67)\). While both $1$ and $2$ can benefit $3$
well, $3$
cannot benefit $1$ or $2$ well, so they are left worse off by joining the grand coalition. Hence, \(x\notin \nu^\alpha(\{1,2\})\), and the NBS is unstable.

The same goes for the KSBS, since $3$'s ideal point is better than $1$ and $2$'s, so
the KSBS would grant $3$ the highest surplus of the players. In particular, the KSBS does not seem fairer than the NBS in this example, giving the highest surplus to a player that is not contributing much.

However, the \(\alpha\)-core (and thus also the \(A\)-core for any payoff matrix \(A\)) is nonempty. For instance, consider the payoff vector \((5,5,3)\). This is feasible via the strategy profile \((2,2,1)\), and one can show that no coalition could guarantee strictly higher payoffs for all of their members.
\end{example}

To address the instability issue, one could try to find a bargaining solution that always picks elements from the core. For instance, if one chooses the disagreement payoff in the core, then the NBS will always be in the core as well. While the \(\alpha\)-core can be empty in general games \parencite[][ch.~13.2]{osborne1994course}, if we assume additive separability, the \(\alpha\)-core is always nonempty. This follows as a corollary from a theorem by \textcite{scarf1967core}. Similar results have been shown in the literature \parencite{scarf1967core}, but I have not found this exact result upon a cursory search, so I am providing it here.

Since the \(\alpha\)-core includes possible threats, which I regard as undesirable, I show the result for a somewhat more strict notion of core. For a coalition \(P\subseteq N\), define \(\Sigma^H_P\subseteq \prod_{i\in P}\Sigma_i\) as the set of Pareto optimal strategies for the players in \(P\). That is, \(\sigma_P\in \Sigma_P^H\) if and only if for
\(x:=\sum_{i\in P}u_i(\sigma_i)\) and \(y:=\sum_{i\in P}u_i(\sigma'_i)\) for any \(\sigma'_P\in \Sigma_P\), if \(y_j\geq x_j\) for all \(j\in P\), then \(y_j=x_j\) for all \(j\in P\).
We then define \(A\) as the worst-case Pareto optimal payoffs in any coalition. That is,
\begin{equation}\label{a-payoffs}A_{i,j}:=\min_{P\subseteq N\text{ s.t. }i\in P}\min_{\sigma_P\in\Sigma_P^H}u_{i,j}(\sigma_i)\end{equation}
for \(i,j\in N\). If we assume that players outside of a coalition are allowed to form their own arbitrary coalitions and Pareto optimal compromises, but we do not allow any threats, this is the relevant notion of core.

The \(A\)-core as defined above, and thus by \Cref{d-core-alpha-core} also the \(\alpha\)-core, is nonempty in additively separable games.

\begin{restatable}{thm}{corenonempty}\label{core-nonempty}
Let \(B\) be a bargaining game with additively separable utility functions, as defined in \Cref{ecl-bargaining-problem-setup}. Let \(A\in\mathbb{R}^{n,n}\) be defined as in \Cref{a-payoffs}. Then the \(A\)-core \(C^A(B)\) is nonempty.
\end{restatable}
\begin{proof}
In \Cref{appendix-proof-of-theorem-core-nonempty}.
\end{proof}

One may ask whether the same would hold for the Nash equilibrium core. Unfortunately, the next example shows that the Nash equilibrium core can be empty, even given additive separability. The intuition behind this is that sometimes, if two players cooperate, this can lead to negative externalities for a third party. However, the Nash equilibrium point does not take this possibility of cooperation between two players into account. Hence, the two players are better off ignoring any agreement that gives everyone at least their Nash equilibrium payoffs.

\begin{example}\label{example-empty-ne-core}
Consider a bargaining game with three players, \(N=\{1,2,3\}\), with payoffs as in \Cref{tab:10}. There is a unique Nash equilibrium in which all players play action \(1\) and receive utilities \((5,5,15)\). However, players \(1\) and \(2\) can also coordinate on action \(2\), which serves as a compromise between the two and produces \(8\) utility for both. Intuitively, we can imagine that \(1\) and \(2\) share some common goal that they can choose maximize instead of their own idiosyncratic goals. However, player \(3\) benefits from the players optimizing their idiosyncratic goals, and if \(1\) and \(2\) cooperate, player \(3\) loses out.

I added a third strategy for the first two players to make sure an option \(x\in F(B)\) exists that strictly dominates the Nash equilibrium disagreement point, but this is inessential to the example. (Similarly, the fact that player \(3\) only has one option that is not Pareto dominated is inessential and can easily be relaxed.)

Now let \(A\) correspond to the Nash equilibrium strategies. Then
\[A=\begin{bmatrix}
    5&0&5\\
    0&5&5\\
    0&0&5
\end{bmatrix}.\]
The coalition \(P=\{1,2\}\) can guarantee its members a payoff of \(8\) each, so
\[\nu^{NE}(\{1,2\})={\{x\in\mathbb{R}^3\mid x_1,x_2\leq 8\}}.\] Moreover, we have
\[\nu^{NE}(\{3\})=\{x\in\mathbb{R}^3\mid x_3\leq 15\},\]
since \(A_{1,3}+A_{2,3}+A_{3,3}=15\), and \(3\) cannot improve upon this payoff by changing their action.

It follows from the above that any payoff vector \(x\in C^{NE}(B)\) has to satisfy \(x_3\geq 15\) and \(x_1,x_2\geq 8\). However, such a payoff vector \((8,8,15)\) is not in the feasible set and thus impossible to obtain. 
The only way to produce \(8\) utility for both \(1\) and \(2\) is for both players to play \(2\). But then player \(3\) can have at most \(5<15\) utility. Hence, \(C^{NE}(B)=\emptyset\).

\begin{table}
\begin{centering}
\begin{tabular}{|c|c|c|c|}
\hline 
Player\textbackslash Options & $1$ & $2$ & $3$ \tabularnewline
\hline 
\hline 
$1$ & $(5,0,5)$ & $(4,4,0)$ & $(3,3,6)$\tabularnewline
\hline 
$2$ & $(0,5,5)$ & $(4,4,0)$& $(3,3,6)$ \tabularnewline
\hline 
$3$ & $(0,0,5)$ & $(0,0,0)$& $(0,0,0)$ \tabularnewline
\hline 
\end{tabular}
\par\end{centering}
\caption{Utility vectors \(u_{i}\) generated by options \(a_i=1,2,3\) for each player \(i=1,2,3\), corresponding to the vertices spanning the individual feasible sets \(F_i(B)\) in \Cref{example-empty-ne-core}.}
\label{tab:10}
\end{table}
\end{example}

\begin{prop}
    There exists a bargaining game \(B\) with additively separable utilities such that the Nash equilibrium core \(C^{NE}(B)\) is empty.
\end{prop}
\begin{proof}
See \Cref{example-empty-ne-core}.
\end{proof}

Based on the above results, one possible way to define a disagreement point that leads to a stable bargaining solution and that does not involve threats would be via
\[d_j =\max\{x_j\mid x\in \nu^A(\{j\})\}=\max_{\sigma_j\in \Sigma_j}u_{j,j}(\sigma_j) + \sum_{i\in N\setminus \{j\}}A_{i,j}\]
for any \(j\in N\) and where \(A\) is defined as in \Cref{a-payoffs}. That is, we let \(d_j\) correspond to the best possible payoff that \(j\) can attain given worst-case Pareto optimal responses by the other players.

It would be valuable to investigate coalitional stability and stable solution concepts in future work, including a more thorough review of the relevant literature on nontransferable utility coalitional games \parencite[e.g.][]{shapley1967utility,Maschler1989,Maschler1992,hart1996,Harsanyi1963}. As in the case of disagreement points, the work by \textcite{diffractor2022rose} may also be relevant.

\section{Conclusion and future work}
\label{conclusion}
In this report, I developed a game-theoretic model of ECL, making it possible to formalize many important aspects and issues with ECL. This includes agents' uncertainty about other agents in the multiverse, the problem of selecting a multiverse-wide compromise outcome, and the question of which joint beliefs to adopt over the actions of agents. There are many interesting open problems and avenues for future work:
\begin{itemize}
    \item How to model agent's default options without ECL? The choice of a disagreement point (\Cref{subsec:The-question-of}) is a fundamental issue in ECL. In particular, there is the question whether threats should play a role in selecting a multiverse-wide compromise. It  may be valuable to review the ROSE value \parencite{diffractor2022rose} or consider acausal bargaining models \parencite[e.g.][]{diffractor2018cooperative,kosoy2015superrationality} in future work.
    \item Another fundamental issue is that of coalitional stability (\Cref{fairness-and-coalitional}), which is related to the problem that compromise between some parties can make other parties worse off, potentially preventing the formation a grand multiverse-wide coalition. While there always exist stable payoff allocations given additively separability, it is unclear what happens if some value systems violate this assumption. Additionally, it is an open question how to choose a stable bargaining solution. Here, it may be useful to review the literature on nontransferable utility coalitional games \parencite[e.g.][]{shapley1967utility,Maschler1989,Maschler1992,hart1996,Harsanyi1963}, as well as \textcite{diffractor2022rose}.
    \item How to assess possible dependencies between different agents, especially in the human case where no source code is available? What is the nature of these dependencies? What is the relevant reference class of agents for superrationality in ECL? Can one rigorously justify inferences such as ``if I choose the NBS, other players are likely to do so, too''? 
    \item How can acausal bargaining models inform ECL? Can we model the process of arriving at conditional beliefs about other agents' actions as some kind of bargaining procedure? If so, what is a plausible model, and how can it inform the problems discussed above?
    \item I make standard common knowledge and common prior assumptions (see \Cref{bayesian-game-formalism}), which are unrealistic in the ECL context, at least when it comes to ECL among humans. How to relax these assumptions? Assigning posterior beliefs to other players is important to assess their gains from trade. How to do this without a common prior? See \textcite{Harsanyi1967,Monderer1989-pj}.
    \item How do gains from trade diminish when agents have different models or choose different bargaining solutions? This would lead to wasted gains from trade, but it is unclear how much would get lost, and how much different value systems would be affected.\footnote{Thanks to Lukas Gloor for a comment on an earlier draft.} How robust are bargaining solutions in practice to different empirical assumptions and model parameters?
    \item An alternate approach to the one employed in this report would be to take joint distributions over actions as given, analyzing and classifying them based on the dependencies they imply. For example, a specific joint distribution could imply positive or negative correlations between more or less cooperative actions of players. One could then investigate which joint distributions enable ECL.\footnote{This approach was suggested to me by Philip Trammell.}
    \item How to deal with the infinities involved in ECL in an infinite universe, as well as the potential continuum of players and values, rather than the discrete set assumed in this report? Is there a relatively small number of discrete clusters of similar value systems, or are there as many different types as players?
    \item What is the distribution over values of superrational cooperators, and what are their beliefs? Can humans usefully make progress on this question, and if not, would superintelligent AI systems be able to do so?
\end{itemize}

The main purpose of this report is to contribute towards the development of a theory of ECL and to outline open technical and philosophical problems, rather than to introduce an applicable model. However, the Bayesian bargaining model from \Cref{sec:Bargaining-with-incomplete} could still be useful for preliminary simulations to investigate possible gains from trade. This might help estimate the potential value of ECL and inform prioritization decisions.

%In light of the potential development of superintelligent AI within the next decades, it seems plausible that the most important intervention for ECL is ensuring that such AI systems are aligned with humans' idealized philosophical views about decision theory and ECL. However, interventions to promote ECL could also backfire by exacerbating other risks from advanced AI \parencite[see][]{xu2021open}. For this reason, I am unsure which concrete interventions to recommend at this point.

\section*{Acknowledgements}
Part of the work on this report was carried out during a Summer Research Fellowship at the Centre for Effective Altruism (CEA). Special thanks go to Max Dalton, who was my supervisor at the CEA. I am grateful for support by the Center on Long-Term Risk, the Center on Long-Term Risk Fund, an Open Phil AI Fellowship, and an FLI PhD Fellowship. Moreover, I am indebted to Lennart Stern, Philip Trammell, Owen Cotton-Barratt, Caspar Oesterheld, Max Daniel, Sam Clarke, Daniel Kokotajlo, Lukas Gloor, Leon Lang, Abram Demski, and Stuart Armstrong for their invaluable discussions and feedback, as well as for their help with the mathematics and game theory in this report. Finally, I want to express my gratitude to commenters on an earlier post, where I requested input on this report \parencite{treutlein2018request}.
\printbibliography

\begin{appendix}

\section{\textcite{Armstrong2013}'s bargaining solution}
\label{appendix-armstrong-solution}
\textcite{Armstrong2013} has published a series of blog
posts on bargaining in which he develops a bargaining solution. In this appendix, I will discuss the solution and argue against using it to model ECL.

In Armstrong's solution, utility functions are normalized such that their
zero point is the disagreement point and \(1\) is their
ideal point, just as with the KSBS. But instead of then taking the
point on the Pareto frontier where everyone has the same utility given
this normalization (as the KSBS would), Armstrong suggests maximizing the
sum of the thus normalized utility functions.

Armstrong discusses two ideas to support his proposed solution. The
first one is the normalization according to the KSBS, which is supposed
to give credit to the fact that if a player can benefit another player
a lot, the other's ideal point will also be higher, and their utility
function will thus be scaled down in the normalization in comparison
to the utility function of the player. The second idea is that of maximizing a sum
instead of maximizing a product or just taking some point with a fixed
ratio of utilities, which is to give agents higher ex ante
expectations of utility.

I think Armstrong's solution is unsuitable for my
setting. First, his solution does not solve the issue with fairness
in a multilateral setting that I discuss in \Cref{fairness-and-coalitional}. Second, as argued in \Cref{subsec:Normalizing-utility-functions},
solutions should guarantee positive gains from trade for all participants.
Maximizing a sum of normalized utility functions does not generally
guarantee that, as I have shown in the case of variance normalization.
As has been pointed out in the comments to \textcite{Armstrong2013}, normalizing according
to disagreement and ideal point may also not guarantee positive gains
from trade.

Lastly, the fact that a bargaining solution maximizes the sum of utilities is not a reason to choose it over other Pareto optimal solutions. Even the KSBS or NBS will maximize \emph{some}
weighted sum of utility functions, since every point on the Pareto frontier corresponds to the maximizer of some weighted sum of 
coordinates. I currently don't see a reason why
choosing the weighting based on knowledge of the entire Pareto frontier is
at an (a priori) disadvantage over weightings which are chosen based
on other information.

Lastly, note that the NBS maximizing a product does not mean that an agent's
uncertainty cannot be taken into account well by the NBS. As outlined in \Cref{sec:Bargaining-with-incomplete}, the expectations
of agents over different possible games can be incorporated
into feasible sets and Pareto frontiers, so the NBS need not only be applied to games
with certainty. Hence, when it comes to expectations over different
games, the NBS chooses a point that is Pareto optimal as judged by agents' beliefs---as
opposed to, for instance, choosing a point which leads to certain gains from
trade but to a lower expectation
across games.

\section{\textcite{Harsanyi1972}'s axiomatization of the Nash bargaining solution in incomplete information games}
\label{appendix-harsanyi-nbs}
In this section, I outline \textcite{Harsanyi1972}'s axiomatization of the NBS in two-player incomplete information games. It is not directly applicable to my setup in \Cref{sec:Bargaining-with-incomplete}, and I did not find a more relevant result in the literature. I believe one should be able to translate the analysis to my setup, but I will not investigate this here.

\textcite{Harsanyi1972}'s axiomatization
includes versions of the axioms from \Cref{bargaining-theory}, namely Individual rationality, Pareto optimality,
Invariance to affine transformations, a version of Anonymity for both
players and all types, and the Independence of irrelevant alternatives
axiom. In addition, there are two new axioms which specifically address the types.

To define these new axioms as in \textcite{Harsanyi1972}, we first have to specify a slightly different version of a Bayesian
bargaining game.
\begin{defn}
A two-player Bayesian bargaining game is a tuple $G=(T_{1},T_{2},F,p)$
where
\begin{itemize}
\item $T_{1}=\{1,\dots,m\}$ and $T_2=\{m+1,\dots,l\}$ are the two sets of types
for either player;
\item $F\subseteq\mathbb{R}^{l}$ is the feasible set, which specifies the
ex interim expected utilities for each type;
\item $p$ is a joint distribution over types for both players.
\end{itemize}
\end{defn}

In this game, there are only two players, $1$ and $2$, and each
player has their own set of types. The feasible set $F$ is just what
would have been the set $F(G)$ in my case, only that the payoffs
depend on both types and players instead of just depending on types. If $x\in F$, then there
exists a mixed strategy profile such that $x_{i}$ specifies the utility
that type $i$ would expect given this mixed strategy profile and
their beliefs about which types the other player could have.

The set $F$ is assumed to be chosen such that the minimal element
in $F$ is the disagreement point. That is, there exists $d\in F$
such that $d_{i}\leq x_{i}$ for all $x\in F,i\in T_{1}\cup T_{2}$.
Moreover, it is assumed that there are positive gains from trade to
be had for everyone---i.e., there is an element $x\in F$ such that
$x_{i}>d_{i}$ for all $i\in T_{1}\cup T_{2}$.

To define one of the new axioms, we need to define the operation of ``splitting a type''.
\begin{defn}[Splitting a type]
We can define splitting a type for feasible payoff vectors as well as for games:
\begin{enumerate}\item
Let $j\in\{1,\dots,m\}$. $j$ is the type of player $1$ we want
to split (the definition is analogous for player $2$). We have two
new sets of types $T'_{1}=\{1,\dots,m+1\}$ and $T'_{2}=\{m+2,\dots,l+1\}$.
Define $F'$ such that it contains all $x'$ such that there is $x\in F$
such that $x'_{i}=x_{i}$ for $i\in\{1,\dots,j\}$, $x'_{j+1}=x_{j}$,
and $x'_{i}=x_{i-1}$ for $i\in\{j+2,\dots,l+1\}$. This is called
deriving $x'$ from $x$ by splitting type $j$ of player $1$ into
two types. 
\item Let $0<\nu<1$. Let $t\in T'_{2}$. We then define $p'$ such that
$p'(k,t)=p(k,t-1)$ for all $k=1,\dots,j-2$, $p'(j,t)=\nu p(j,t-1)$,
$p'(j+1,t)=(1-\nu)p(j,t-1)$, and $p'(k,t)=p(k-1,t-1)$ for $k\in\{j+1,\dots,m+1\}$.
The new game $G'=(T'_{1},T'_{2},F',p')$ with $F'$ as feasible set,
types $T'_{1},T'_{2}$, and $p'$ as distribution over types is derived
from splitting type $j$ of player $1$ into two types with probabilities
$\nu$ and $1-\nu$.
\end{enumerate}
\end{defn}

With these definitions, the two new axioms are as follows:
\begin{ax}
Splitting types. If $G'=(T'_{1},T'_{2},F',p')$ is derived from $G$
by splitting type $j$ of player $1$ into two types with probabilites
$\nu$ and $\nu-1$, then $x'=\mu(G')$ is derived from $x=\mu(G)$
by splitting type $j$ of player $1$ into two types.
\end{ax}

\begin{ax}
Mixing basic probability matrices. If $G=(T_{1},T_{2},F,p)$ and $G'=(T_{1},T_{2},F,p')$
have the same solution vector, then for every $G''=(T_{1},T_{2},F,p'')$
with $p''=\nu p+(1-\nu)p'$ where $\nu\in[0,1]$, it is $\mu(G'')=\mu(G')=\mu(G)$. 
\end{ax}

Given these two additional axioms, \textcite{Harsanyi1972} show that the solution
function must be
\begin{equation}
\mu(G)=\argmax_{x\in F}\prod_{t\in T_{1}\cup T_{2}}(x_{t}-w_{t})^{p(t)}.
\end{equation}
That is, an asymmetric version of the NBS where the weights are the
prior probabilities of the types.

\section{Proof of Theorem~\ref{spohn5-continuous}}
\label{appendix-proof-of-spohn5-continuous}

\spohnfivec*
\begin{proof}
Let \(q:=\delta_\alpha\) be the Dirac measure, defined via \(\delta_\alpha(A)=1\) if and only if \(\alpha\in A\). As in \Cref{thm:spohn5}, we now want to define a joint distribution \(q_r\) for each \(r\in\mathbb{N}\) that converges weakly to \(q\). To that end, define \(s\) as follows. For every \(t\in T\), let \(\mu_t\) be some probability measure on \(\mathcal{A}_t\) with full support such that \(\mu_t(\{\alpha
_t'\})=0\) for any \(\alpha'_t\in \mathcal{A}_t\) (assuming \(\mathcal{A}_t\) contains more than one point, and thus by convexity a continuum of points). For any set \(A\subseteq\mathcal{A}\), define
\[s(A):=m^{-1}\sum_{t\in T}\mu_t(\{\alpha'_t\mid (\alpha'_t,\beta_{-t})\in A\}).\]
To show that this is a probability measure,  note that \(s(\emptyset)=0\), \(s\) is always non-negative, and
\[s(\mathcal{A})
=m^{-1}\sum_{t\in T}\mu_t(\{\alpha_t\mid (\alpha'_t,\beta_{-t})\in \mathcal{A}\})
=m^{-1}\sum_{t\in T}\mu_t(\mathcal{A}_t)
=1.
\]
Moreover, for any countable collection of pairwise disjoint sets \(A^1,A^2,\dotsc\), we have
\begin{multline}s\left(\bigcup_{l\in\mathbb{N}}A^l\right)
=m^{-1}\sum_{t\in T}\mu_t\left(\alpha'_t\mid (\beta_{-t},\alpha'_t)\in \bigcup_{l\in\mathbb{N}}A^l\right)
=
m^{-1}\sum_{t\in T}\sum_{l\in\mathbb{N}}\mu_t(\{\alpha'_t\mid (\beta_{-t},\alpha'_t)\in A^l\})
\\=\sum_{l\in\mathbb{N}}m^{-1}\sum_{t\in T}\mu_t(\{\alpha'_t\mid (\beta_{-t},\alpha'_t)\in A^l\})
=\sum_{l\in\mathbb{N}}s(A^l).
\end{multline}
This shows that \(s\) is a probability measure.

Moreover, for any open, nonempty \(A_t\subseteq\mathcal{A}_t\), we have 
\[s(A_t)
=
m^{-1}\sum_{t'\in T}\mu_t(\{\alpha'_{t'}\mid (\alpha'_{t'},\beta_{-t'})\in \mathcal{A}_{-t}\times A_t\})
\geq m^{-1}\mu_t(A_t)>0,\]
so this measure satisfies the full support condition that is required to define \(q_r\).

Now we define \(q_r:=\frac{r-1}{r}q+\frac{r-1}{r^2}\delta_\beta+\frac{1}{r^2}s\). Since this is a convex combination of probability measures, it is still a probability measure.  It remains to show that this measure satisfies our requirements. First, clearly, this weakly converges to \(q\) as \(r\rightarrow\infty\). Second, since \(\frac{1}{r^2}>0\) for all \(r\in\mathbb{N}\), it is \(q_r(A_t)\geq \frac{1}{r^2}s(A_t)>0\) for any \(t\in T\) and nonempty open set \(A_t\subseteq\mathcal{A}_t\).

Now we turn to the condition on expected utilities. Let \(t\in T\) and \(A_t\subseteq\mathcal{A}_t\) with \(q(A_t)>0\) arbitrary but fixed in the following. Then it follows that \(\alpha_t\in A_t\), and thus
\(
q(\{\alpha\}\mid A_t)=\frac{q(\{\alpha\})}{q(A_t)}=1
\). Hence, for measurable \(A\subseteq \mathcal{A}\), it follows
\[\lim_{r\rightarrow\infty}q_r(A\mid A_t)
=\lim_{r\rightarrow\infty}\frac{\frac{r-1}{r}\delta_\alpha(A)+\frac{r-1}{r^2}\delta_\beta(A)+\frac{1}{r^2}s(\{\alpha'\in A\mid \alpha'_t \in A_t\})}{\frac{r-1}{r}+\frac{r-1}{r^2}\delta_\beta(A)+\frac{1}{r^2}s(A_t)}
=\delta_\alpha(A)=q(A\mid A_t)\]
and thus
\[
\lim_{r\rightarrow\infty}EU_t(q_r;A_t)=EU_t(q;A_t)=EU_t(\alpha).
\]

Next, let \(B_t\subseteq\mathcal{A}_t\) an arbitrary nonempty open set, representing any other set of actions type \(t\) could condition on. We have to show that \(\lim_{r\rightarrow\infty}EU_t(q_r;A_t)\geq \lim_{r\rightarrow\infty}EU_t(q_r;B_t)\). To that end, note that if \(q(B_t)>0\), it follows from the above that
\[\lim_{r\rightarrow\infty}EU_t(q_r;B_t)=EU_t(\alpha)=\lim_{r\rightarrow\infty}EU_t(q_r;A_t),\]
and we are done.

Now consider the case \(q(B_t)=0\). First, assume \(\beta_t\in B_t\). In this case, for measurable \(A\subseteq \mathcal{A}\), we have
\begin{align}
\lim_{r\rightarrow\infty}q_r(A\mid B_t)
&=\lim_{r\rightarrow\infty}
\frac{\frac{r-1}{r}\delta_{\alpha}(A\cap(\mathcal{A}_{-t}\times B_t)) +\frac{r-1}{r^2}\delta_\beta(A\cap(\mathcal{A}_{-t}\times B_t))+\frac{1}{r^2}s(A\cap(\mathcal{A}_{-t}\times B_t))}{\frac{r-1}{r}q(B_t)+\frac{r-1}{r^2}\delta_{\beta_t}(B_t)+\frac{1}{r^2}s(B_t)}
\\
&=\lim_{r\rightarrow\infty}
\frac{\frac{r-1}{r^2}\delta_\beta(A\cap(\mathcal{A}_{-t}\times B_t))+\frac{1}{r^2}s(A\cap(\mathcal{A}_{-t}\times B_t))}{\frac{r-1}{r^2}\delta_{\beta_t}(B_t)+\frac{1}{r^2}s(B_t)}
\\
&=\lim_{r\rightarrow\infty}
\frac{\frac{r-1}{r}\delta_\beta(A\cap(\mathcal{A}_{-t}\times B_t))+\frac{1}{r}s(A\cap(\mathcal{A}_{-t}\times B_t))}{\frac{r-1}{r}\delta_{\beta_t}(B_t)+\frac{1}{r}s(B_t)}
\\
&=
\frac{\delta_\beta(A\cap(\mathcal{A}_{-t}\times B_t))}{\delta_{\beta_t}(B_t)}
=\delta_\beta(A)
.
\end{align}
Hence, it follows that
\[\lim_{t\rightarrow\infty}
EU_t(q_r;B_t)=
\lim_{t\rightarrow\infty}
\mathbb{E}_{\alpha'\sim q_r}[EU_t(\alpha')\mid \alpha'_t\in B_t]
=
\lim_{t\rightarrow\infty}
\mathbb{E}_{\alpha'\sim \delta_\beta}[EU_t(\alpha')]
=EU_t(\beta).\]
Using the assumption on \(\alpha\) and \(\beta\), we can conclude that
\[\lim_{t\rightarrow\infty}
EU_t(q_r;B_t)=EU_t(\beta)\leq EU_t(\alpha)=\lim_{r\rightarrow\infty}EU_t(q_r;A_t),\]
and we are done.

Second, consider the case \(\beta_t\notin B_t\). Then for any \(r\in\mathbb{N}\), we have
\begin{multline}q_r(A\mid B_t)\label{qr-equal-s}
=
\frac{\frac{r-1}{r}\delta_{\alpha}(A\cap(\mathcal{A}_{-t}\times B_t)) +\frac{r-1}{r^2}\delta_\beta(A\cap(\mathcal{A}_{-t}\times B_t))+\frac{1}{r^2}s(A\cap(\mathcal{A}_{-t}\times B_t))}{\frac{r-1}{r}q(B_t)+\frac{r-1}{r^2}\delta_{\beta_t}(B_t)+\frac{1}{r^2}s(B_t)}\\
=
\frac{\frac{1}{r^2}s(A\cap(\mathcal{A}_{-t}\times B_t))}{\frac{1}{r^2}s(B_t)}
=
\frac{s(A\cap(\mathcal{A}_{-t}\times B_t))}{s(B_t)}
=s(A\mid B_t).
\end{multline}
It follows that \(EU_t(q_r;B_t)=EU_t(s;B_t)\).

Now define \(A_t^{\beta_{-t}}:=\{\alpha_t'\mid \alpha'\in A\colon \alpha'_{-t}=\beta_{-t}\}\). Then
\begin{multline}s(A\cap (\mathcal{A}_{-t}\times B_t))=
m^{-1}\sum_{t'\in T}\mu_{t'}(\{\alpha'_{t'}\mid (\alpha'_{t'},\beta_{-{t'}})\in A\cap (\mathcal{A}_{-t}\times B_t)\})
\\
=
m^{-1}\mu_{t}(\{\alpha'_{t}\mid (\alpha'_{t},\beta_{-{t}})\in A, \alpha'_{t}\in B_t\})
=m^{-1}\mu_t(B_t\cap A_t^{\beta_{-t}}).\end{multline}
It follows that \(s(A\mid B_t)=0\) if \(\beta_{-t}\notin A_{-t}\), so
for a random variable \(\alpha'\sim s\), we have \[s(\alpha'_{-t}=\beta_{-t}\mid B_t)=1.\]
It follows for any \(r\in\mathbb{N}\) that
\begin{multline}EU_t(q_r;B_t)\underset{\text{(\ref{qr-equal-s})}}{=}EU_t(s;B_t)=\mathbb{E}_{\alpha'\sim s}[EU_t(\alpha')\mid \alpha'_t\in B_t]=\mathbb{E}_{\alpha'\sim s}[EU_t(\beta_{-t},\alpha'_t)\mid \alpha_t'\in B_t]
\\
\underset{\text{(i)}}{\leq}
\mathbb{E}_{\alpha'\sim s}[EU_t(\alpha)\mid \alpha_t'\in B_t]=EU_t(\alpha),\end{multline}
where we have used the assumption on \(\alpha,\beta\) in (i).
Hence, also
\[\lim_{r\rightarrow\infty}EU_t(q_r;B_t)\leq EU_t(\alpha)=\lim_{r\rightarrow\infty}EU_t(q_r;A_t),\]
which concludes the proof.
\end{proof}

\section{Proof of Theorem~\ref{core-nonempty}}
\label{appendix-proof-of-theorem-core-nonempty}
I begin by introducing some additional notation, in order to be able to state the result used to prove \Cref{core-nonempty}. 
The following definitions and conditions are adapted from \textcite{kannai1992core}. I assume a set \(N\) of players is given.

\begin{defn}[Characteristic function]A function \(\nu\colon \mathcal{P}(N)\rightarrow\mathcal{P}(\mathbb{R}^N)\) is called \emph{characteristic function} if it satisfies the following criteria:
\begin{enumerate}
    \item[(i)] \(\nu(\emptyset)=\emptyset\);
    \item[(ii)] for all \(S\subseteq N\), \(S\neq \emptyset\), \(\nu(S)\) is a nonempty closed subset of \(\mathbb{R}^N\);
    \item[(iii)] if \(x\in \nu(S)\) and \(y_i\leq x_i\) for all \(i\in S\), then \(y\in \nu(S)\);
    \item[(iv)] there exists a closed set \(F\subseteq\mathbb{R}^N\) such that
    \[\nu(N)=\{x\in\mathbb{R}^N\mid \exists y\in F\colon \forall i\in N\colon x_i\leq y_i\};\]
    \item[(v)]The set
    \(F\cap\{x\in\mathbb{R}^N\mid \forall i\in N\colon x_i\geq \max\{y_i\mid y\in \nu(\{i\})\}\}\)
    is nonempty and compact.
\end{enumerate}
\end{defn}

Now let \(B\) be a bargaining game and \(A\in \mathbb{R}^{n,n}\) such that \(A_{i,j}\in \{x_{i,j}\mid x_i\in F_{i}(B)\}\) for all \(i,j\in N\). 
Then the function \(\nu^A\) of \(A\)-dominated vectors as introduced in \Cref{fairness-and-coalitional}, defined via
\[\nu^A(P):=\{x\in \mathbb{R}^n\mid \exists y\in \prod_{i\in P}F_i\colon \forall i\in P\colon x_i \leq \sum_{j\in P}y_{j,i}+\sum_{j'\in N\setminus P}A_{j',i}\},\]
satisfies these criteria. 

\begin{lem}\label{lem-nu-characteristic}
    \(\nu^A\) is a characteristic function.
\end{lem}
\begin{proof}
Left as an exercise. It follows from the assumption that the \(F_i(B)\) are compact, convex sets, together with the definition of \(\nu^A\). For (iv) and (v), we can take \(F=F(B)\).
\end{proof}

Next, we need two technical definitions to be able to state the result.

\begin{defn}[Balanced collection]
Let \(T\subseteq\mathcal{P}(N)\) be a collection of coalitions. Then \(T\) is said to be a \emph{balanced collection} if there exist nonnegative weights \((\delta_S)_{S\in T}\) such that
\[\sum_{S\in T\text{ s.t. }i\in S}\delta_S=1.\]
\end{defn}
This means that there exist weights for each set in \(T\) such that, for each player \(i\in N\), the weights of all the sets containing that player add up to \(1\).

\begin{defn}[Balanced characteristic function]A characteristic function \(\nu\) is called \emph{balanced} if for every balanced collection \(T\), we have
    \[\bigcap_{S\in T}\nu(S)\subseteq \nu(N).\]
\end{defn}

This means that if a payoff vector \(x\) can be guaranteed for their members by every single coalition in a balanced set of coalitions, then it must also be achievable by the grand coalition. This is in general not true, but we will show that it is true in the case of an additively separable bargaining game.

Now we can state the main result used to prove \Cref{core-nonempty}. Recall the definition of the core as the set of vectors \(x\in \nu(N)\) such that for all coalitions \(P\subseteq N\) and \(y\in \nu(P)\), there exists at least one player \(i\in P\) such that \(x_i\geq y_i\). Note that every characteristic function \(\nu\) defines a core \(C^\nu\).

\begin{thm}[\cite{scarf1967core,shapley1973balanced,kannai1992core}]
\label{balanced-core-nonempty}
    Every balanced characteristic function has a nonempty core.
\end{thm}
\begin{proof}
    See \textcite{kannai1992core}.
\end{proof}

Now we can prove \Cref{core-nonempty}.

\corenonempty*

\begin{proof}
Consider a bargaining game \(B\) with additively separable utility functions. Recall
\[A_{i,j}:=\min_{P\subseteq N \text{ s.t. }i\in P}\min_{\sigma_P\in\Sigma_P^H}u_{i,j}(\sigma_i)\]
for \(i,j\in N\), where \(\Sigma^H_P\) is the set of Pareto optimal strategies for the players in \(P\). By \Cref{lem-nu-characteristic}, \(\nu:=\nu^A\) is a characteristic function. It remains to show that \(\nu\) is balanced. Then it follows from \Cref{balanced-core-nonempty} that \(C^A(B)=C^\nu\) is nonempty.

To that end, assume \(T\) is a balanced collection with weights \((\delta_S)_{S\in T}\), and assume \(x\in \nu(S)\) for all \(S\in T\).
Then by definition, there exists \(x_i^{S}\in F_i(B)\) for each \(i\in N\) that corresponds to the utilities produced by player \(i\) in coalition \(S\), such that
\[x_j\leq \sum_{i\in S} x_{i,j}^{S} + \sum_{i \in N\setminus S}A_{i,j}\]
for all \(j\in S\). Note that w.l.o.g., we can assume that for some \(\sigma_S\in \Sigma_S^H\), we have \(x_i^S=u_i(\sigma_i)\) for all \(i\in S\). That is, we can choose vectors \(x_i^S\) that result in Pareto optimal payoffs for the members of \(S\).
Then, by definition of \(A\), we have
\begin{equation}\label{worst-case-condition-A}x_{i,j}^{S}\geq A_{i,j}\end{equation}
for any player \(i\in S\) and \(j\in N\).

Now we want to find a matrix of vectors \(\hat{x}\in\mathbb{R}^{n,n}\) such that
\(\hat{x}_i\in F_i(B)\) for each \(i\in N\), and such that
\[\sum_{i\in N}\hat{x}_{i,j}\geq x_j\] for all \(j\in N\). If we we can find such a matrix, then it follows that \(x\in F(B)\) and thus \(x\in \nu(N)\), and we are done.

To define this matrix, let \(i\in N\) arbitrary and set
\[\hat{x}_i:=\sum_{S\in T\text{ s.t. }i\in S}\delta_Sx_i^{S}.\]
Note that this is a convex combination of vectors \(x_i^S\in F_i(B)\), and thus also \(\hat{x}_i\in F_i(B)\) since the feasible sets are convex. It follows that
\begin{multline}
\sum_i\hat{x}_{i,j}=\sum_{i\in N}\sum_{S\in T\text{ s.t. }i\in S}\delta_Sx_{i,j}^{S}
=\sum_{S\in T}\sum_{i\in S}\delta_Sx_{i,j}^{S}
=\sum_{S\in T}(\mathbbm{1}_S(j)
\sum_{i\in S}\delta_Sx_{i,j}^{S}
+(1-\mathbbm{1}_S(j))\sum_{i\in S}\delta_Sx_{i,j}^{S})
\\
\underset{\text{\ref{worst-case-condition-A}}}{\geq}\sum_{S\in T}(\mathbbm{1}_S(j)\delta_S(x_j- \sum_{i\in N\setminus S}A_{i,j})
+(1-\mathbbm{1}_S(j))\sum_{i\in S}\delta_SA_{i,j})
\\
=
x_j- \sum_{S\in T}\delta_S(\mathbbm{1}_S(j)\sum_{i\in N}A_{i,j}
-\sum_{i\in S}A_{i,j})
=
x_j- \sum_{i\in N}A_{i,j}+\sum_{S\in T}\delta_S\sum_{i\in S}A_{i,j}
\\=
x_j-\sum_{i\in N}A_{i,j}+\sum_{i\in N}\sum_{S\in T\text{ s.t. }i\in S}\delta_SA_{i,j}
=
x_j-\sum_{i\in N}A_{i,j}+\sum_{i\in N}A_{i,j}
=x_j.
\end{multline}
This shows that \(x\in \nu(N)\) and thus concludes the proof.
\end{proof}

\end{appendix}

\end{document}